\documentclass[a4paper,USenglish, cleveref, autoref, thm-restate]{lipics-v2021}

\nolinenumbers
\hideLIPIcs

\usepackage{mathtools}
\usepackage{amsmath, amssymb, amsthm}
\usepackage{graphicx}
\usepackage{enumitem}

\usepackage{tikz}

\usetikzlibrary{arrows.meta, automata}

\newcommand{\A}{\mathcal{A}}
\newcommand{\N}{\mathbb{N}}
\newcommand{\R}{\mathbb{R}}
\newcommand{\G}{\mathcal{G}}
\newcommand{\C}{\mathcal{C}}
\renewcommand{\P}{\mathcal{P}}

\declaretheorem[name=Question, sibling=theorem]{introquestion}

\declaretheorem[name=Open problem, sibling=theorem]{open}

\title{Complexity landscape for local certification}

\author{Nicolas Bousquet}{CNRS, INSA Lyon, UCBL, LIRIS, UMR5205, F-69622 Villeurbanne, France }{}{}{}{}
\author{Laurent Feuilloley}{CNRS, INSA Lyon, UCBL, LIRIS, UMR5205, F-69622 Villeurbanne, France }{}{}{}{}
\author{Sébastien Zeitoun}{CNRS, INSA Lyon, UCBL, LIRIS, UMR5205, F-69622 Villeurbanne, France }{}{}{}{}

\authorrunning{N. Bousquet, L. Feuilloley, S. Zeitoun}



\funding{This work is supported by the ANR grant ENEDISC (ANR-24-CE48-7768).}

\acknowledgements{The authors would like to thank Thomas Colcombet for answering our questions about Chrobak normal form.}

\begin{document}

	\maketitle
	\begin{abstract}
		An impressive recent line of work has charted the complexity landscape of distributed graph algorithms. For many settings, it has been determined which time complexities exist, and which do not (in the sense that no local problem could have an optimal algorithm with that complexity). In this paper, we initiate the study of the landscape for \emph{space complexity} of distributed graph algorithms. More precisely, we focus on the local certification setting, where a prover assigns certificates to nodes to certify a property, and where the space complexity is measured by the size of the certificates.   
		
		Already for anonymous paths and cycles, we unveil a surprising landscape:
		\begin{itemize}
			\item There is a gap between complexity $O(1)$ and $\Theta(\log \log n)$ in paths. This is the first gap established in local certification. 
			\item There exists a property that has complexity $\Theta(\log \log n)$ in paths, a regime that was not known to exist for a natural property. 
			\item There is a gap between complexity $O(1)$ and $\Theta(\log n)$ in cycles, hence a gap that is exponentially larger than for paths.  
		\end{itemize}
		
		We then generalize our result for paths to the class of trees. Namely, we show that there is a gap between complexity $O(1)$ and $\Theta(\log \log d)$ in trees, where $d$ is the diameter. We finally describe some settings where there are no gaps at all.
		
		To prove our results we develop a new toolkit, based on various results of automata theory and arithmetic, which is of independent interest. 
	\end{abstract}

	\newpage{}
	\tableofcontents{}
	
	\newpage{}

	\clearpage
	\setcounter{page}{1}

\section{Introduction}
\label{sec:overview}

\subsection{Approach}

\paragraph*{Time complexity landscapes for distributed graph algorithms}

For distributed graph algorithms, the most classic measure of complexity is \emph{time}, measured by the number of rounds before output. One the most fruitful research programs on this topic has been the one of charting the complexity landscape~\cite{Suomela20}. That is, instead of considering specific problems and asking for their time complexity, a long series of papers have answered the following question: given a complexity function, is there a problem with that complexity? 
In this paper, we aim at starting the analogue research program for the \emph{space complexity} of distributed graph algorithms.  

Let us mention some elements about the time complexity landscape, in order to draw analogies with our setting later.\footnote{For the sake of simplicity, we only discuss deterministic complexity of locally checkable labelings, in bounded degree graphs in the LOCAL model.} Before 2016, there were a few of classic complexities: $O(1)$ for very simple tasks, $\Theta(\log^*n)$ for tasks that could be reduced to coloring, and $O(n)$ for global problems (throughout the paper, $n$ refers to the number of nodes in the graph). Then, several seminal papers established that some problems have complexity $\Theta(\log n)$~\cite{BrandtFHKLRSU16, ChangKP19}, and that there is no problem whose complexity strictly lies between $\Theta(\log^*n)$ and $\Theta(\log n)$~\cite{ChangKP19}. 
This gap sharply contrasts  with classic computational theory where the time hierarchy theorem prevents the existence of such gaps~\cite{HartmanisS65}. 
It was also proved that in some intervals infinitely many complexities exist~\cite{BalliuBOS21, BalliuHKLOS18, Chang20,ChangP19}, but not \emph{any} complexity function had a corresponding problem.
A key strategy in this research area is to understand the landscape for specific classes of networks such as paths~\cite{BalliuBCORS19}, grids~\cite{BrandtHKLOPRSU17}, trees~\cite{BalliuBCOSS22, BalliuBCOSST23, BalliuBOS21}, minor-closed graphs~\cite{Chang24}, etc.
In this paper, we will follow a similar approach: characterizing which space complexities are possible depending on the structure of the network.

\paragraph*{New direction: A landscape approach to space }

The \emph{space} used by distributed algorithms is much less studied than the time, especially if we restrict to distributed graph algorithms. Some models where space is a well-studied measure are population protocols~\cite{AspnesR07}, massively parallel computing (MPC)~\cite{ImKLMV23} and models similar to Stone Age~\cite{EmekW13}, but they are not relevant to our approach, either because they differ too much from distributed graph computing (having one-to-one communication or a centralized component) or because they fix the space complexity to constant.  
As far as we know, the only major line of research considering space complexity for distributed graph algorithms is self-stabilization, and more specifically the state model, where the nodes update their states by reading the states of their neighbors~\cite{AltisenDDP19, Dolev2000}. 
In that model, the algorithm has to cope with transient faults, and consequently it is common to design algorithms (at least implicitly) with two components : one that builds a solution and one that checks the solution, and can reset the system if needed. 
\emph{Local certification} was introduced to study specifically the space needed for the checking phase, by abstracting the computation of the solution into an oracle, and this is our focus today.

Informally, a local certification of a property is an assignment of labels to the nodes of the graph, such that the nodes can collectively check the correctness of a given property by inspecting the labels in their neighborhoods. 
This notion is known under different names, depending on the specific model considered: \emph{proof-labeling scheme}~\cite{KormanKP10}, locally checkable proofs~\cite{GoosS16}, non-deterministic local decision~\cite{FraigniaudKP13}, etc.
We will formally define our model in Section~\ref{sec:model}, and refer to~\cite{Feuilloley21} for an introduction to the notion.  
The (space) complexity of a certification is the maximum size of the label given to a node, and it is known to be basically equal to the space complexity of self-stabilizing algorithms.\footnote{There are some fine prints to this statement, that we will discuss later, in Section~\ref{sec:discussions}.} 
Finally, since local certification is about decision problems, it is now standard to focus on checking graph properties (instead of solutions to graph problems), and to refer to the set of correct graphs as a \emph{language}.   

In local certification, the situation is similar to the pre-2016 situation for the time complexity as described above. That is, there are a few well-established complexity regimes -- $O(1)$, $\Theta(\log n)$, $\Theta(n)$ and $\Theta(n^2)$ -- but no solid explanation as of why these are so common.
But unlike for distributed time complexity, we can prove that for a wide range of complexities, there exists a problem with that complexity. 

\begin{restatable}{theorem}{ThmFolkloreNoGap}
	\label{thm:folklore-no-gap}
	For general graphs with identifiers, for any non-decreasing function $f(n)$ in $\Omega(\log n)$ and $O(n^2)$, there exists a property that can be certified with $O(f(n))$ bits, but not in $o(f(n))$ bits.
\end{restatable}

The proof follows a standard construction of~\cite{GoosS16}.
The correct instances are mde of two copies of a graph $H$, linked by a path. The graph $H$ is chosen so that it can be encoded on $f(n)$ bits, and a certification consists in giving this encoding to all nodes. A counting argument allows to prove that this is actually necessary~\cite{GoosS16}. 

At first sight, Theorem~\ref{thm:folklore-no-gap} gives a trivial answer to our question about the complexity landscape of local certification. But it relies crucially on several assumptions. 
First, in the upper bound it is necessary to have unique identifiers to avoid being fooled by symmetries at the checking phase. Also, because of the identifiers and because the value of $n$ needs to be certified, $\Omega(\log n)$ bits are needed.  
Finally, to make the counting argument of the lower bound work, one cannot restrict too much the graphs in which to apply this theorem. This rises the three following questions that we tackle in this paper.

\begin{introquestion}
	What happens for complexities in $o(\log n)$?
\end{introquestion}

In other words, is the $\Omega(\log n)$ a limitation of the proof technique or could there be a gap between constant and $\Theta(\log n)$, as the current state of our knowledge suggests? 

\begin{introquestion}
	What about anonymous networks?
\end{introquestion}

This is especially relevant in conjunction with the previous question, since a unique identifier cannot be encoded in a certificate of $o(\log n)$ bits. 
Note that even beyond that regime, the impact of the identifiers on local decision is a well-studied topic~\cite{FraigniaudHK12, FraigniaudGKS13, FraigniaudHS18, FeuilloleyF16}.  

\begin{introquestion}
	What about structured graphs, like paths, cycles and trees?
\end{introquestion}

In the spirit of the work made on the distributed time complexity landscape, we would like to understand local certification on restricted graph classes, in which counting arguments do not apply, or only partially.  

\subsection{Main results}

We start with a simple setting: anonymous paths without input labels. In this setting, a path is characterized by its length, and a language is defined by a set of authorized lengths. Naturally, we assume that the nodes \emph{do not} have the knowledge of $n$ (otherwise the problem is trivial).
Given $O(\log n)$ bits, we can recognize any language. This is because the prover can certify the exact length of the path, by giving to every node its distance to a chosen endpoint as a certificate. The nodes can then check locally the correctness of this counter, and the last node can check whether the length belongs to the authorized sizes. 
On the other hand, with a constant number of bits, we can encode properties like ``the path has length $k \mod m$'', via counters modulo $m$ (where $m$ is a constant). 
To build a language whose complexity would be strictly between these regimes, it is tempting to consider properties of the form: ``the path has length $k$ modulo $f(n)$'' for arbitrary function $f$, and for which a counter would use $\log f(n)$ bits. Unfortunately, the nodes are unable to check that the $f(n)$ is correct, since they do not have access to $n$. This obstacle does not seem easy to bypass and, it is reasonable to conjecture that there is a gap between $O(1)$ and $\Theta(\log n)$. We do establish the existence of a gap, but it only goes up to $\Theta(\log \log n)$. 

\begin{restatable}{theorem}{ThmGapPathsChrobak}
	\label{thm:gap-paths-chrobak}
	Let $c > 1$ and $N \in \N$. Let $\P$ be a property on paths that can be certified with certificates of size $s(n):=\left\lfloor\frac{\log \log n}{c}\right\rfloor$ for all $n \geqslant N$. Then, $\P$ can also be certified with constant-size certificates.
\end{restatable}

This is \emph{the first gap established in local certification}. 
At first, this $\log \log n$ looks like an artifact of the proof but, surprisingly, it is not! We next prove that there exists a language for which the optimal certificate size is $\Theta(\log \log n)$. (Given the previous theorem, it is sufficient to prove that this language can be certified with $O(\log \log n)$ bits but not with $O(1)$ bits.) 

\begin{restatable}{theorem}{ThmCertifProductPrimes}
	\label{thm:certif_product_primes}
	There exist properties on paths that can be certified with certificates of size $O(\log \log n)$, but not with certificates of size $O(1)$.
\end{restatable}

It is \emph{the first time that this regime appears in the area of local certification} (here under the promise that the graph is a path).  
The language we use is the set of paths whose length is not a product of consecutive primes; we will come back to it. 

Now, moving on to cycles, there is a second surprise. We expect to see the same landscape in paths and cycles, but the intermediate regime actually disappears in cycles, and there is a gap between $O(1)$ and $\Theta(\log n)$.

\begin{restatable}{theorem}{ThmCycles}
	\label{thm:cycles}
	Let $c > 12$ and $N \in \N$. Let $\P$ be a property on cycles that can be certified with certificates of size $s(n) := \left\lfloor \frac{\log n}{c} \right\rfloor$ for every integer $n \geqslant N$. Then, $\P$ can also be certified with constant-size certificates.
\end{restatable}

Finally, we study the case of trees, where the picture is more complex, as it depends on the exact setting considered. In some settings, we can prove that there is no gap (any reasonable function corresponds to the optimal certificate size for a language), and in some other cases, the gap from paths persists. 
These settings depend on three parameters: 
(1)~whether we consider certificate size as a function of the number of nodes~$n$ or of the diameter~$d$, 
(2)~whether the nodes can only see what is at distance 1 or at larger distance (in other words the verification radius is 1 or larger), and (3)~whether the maximum degree is bounded or not. We establish a classification of these different settings. In particular, we prove that \emph{the gap for paths is generalized to arbitrary trees}, when parameterized by the diameter, and the verification radius is 1.

\begin{restatable}{theorem}{ThmGapTreeRadiusOne}
	\label{thm:gap-tree-radius-1}
	Let $c > 2$ and $D \in \N$. Let $\P$ be a property in trees (of unbounded degree) that can be certified using $s(d):=\left\lfloor\frac{\log \log d}{c}\right\rfloor$ bits for all $d \geqslant D$ (where $d$ is the diameter). Then, $\P$ can also be certified with constant-size certificates.
\end{restatable}

We then show that two assumptions in Theorem~\ref{thm:gap-tree-radius-1} are optimal. Namely, we prove in Theorem~\ref{thm:caterpillars-no-gap-in-n} that we can not hope for a gap in~$n$ instead of~$d$ in trees (even in caterpillars, that is, paths with leafs attached), and we establish in Theorem~\ref{thm:caterpillar-no-gap-radius-2} that it is essential that the vertices are able to see only at distance only~1, because there is no more gap if the verification radius~$r$ is at least~2 (again, even in caterpillars).

\begin{restatable}{theorem}{ThmCaterpillarNoGapinNRadiusOne}
	\label{thm:caterpillars-no-gap-in-n}
	Let $f : \N \to \mathbb{R}$ be a non-decreasing function such that $\lim_{n \rightarrow +\infty}f(n) = +\infty$ and for all integers $1 \leqslant s \leqslant t$ we have $f(t)-f(s) \leqslant \log t - \log s$. Then, there exists a property on caterpillars (of unbounded degree) that can be certified with certificates of size~$O(f(n))$ and not with certificates of size~$o(f(n))$.
\end{restatable}

\begin{restatable}{theorem}{ThmCaterpillarNoGapRadiusTwo}
	\label{thm:caterpillar-no-gap-radius-2}
	Let $r > 1$.
	Let $f : \N \to \N$ be a non-decreasing function such that $\lim_{n \rightarrow +\infty}f(n) = +\infty$ and for all integers $1 \leqslant s \leqslant t$ we have $f(t)-f(s) \leqslant \log t - \log s$. Then, if the vertices can see at distance~$r$, there exists a property on caterpillars (of unbounded degree) that can be certified with certificates of size~$O(f(d))$ and not with certificates of size~$o(f(d))$, where $d$ is the diameter.
\end{restatable}

Note that the condition on $f(t)-f(s)$ only ask for a function which is sublogarithmic and whose growth is sublogarithmic; to avoid arbitrarily long plateaus followed by big jumps. This condition is satisfied by all the usual functions ($\sqrt{\log n}$, $\log \log n$, $\log^*n$, etc.).

All the results mentioned so far are in the anonymous setting. We also explore the impact of identifiers and of the knowledge of $n$ on this landscape, but we delay this discussion for now. 

\subsection{Main techniques}

\paragraph*{Automata and arithmetic perspectives}

To give an overview of our techniques, let us start by describing two perspectives on constant size certification in anonymous paths. 
As said earlier, a basic building block is to certify with a counter that the length is equal to $i \mod k$ for some constant $i$ and $k$. 
One can see this behavior as a run in an automaton with $k$ states inducing a cycle whose initial state being $0$ mod $k$ and the final state being $i$ mod $k$. 
More generally, any constant size certification can be turned into an automaton. There are various ways to do it, but basically, it consists in having states describing (pairs of) certificates, and transitions that connect states that would be accepted by the local verifier. The existence of an accepting run in the automaton is equivalent to the existence of a certificate assignment making the local verifier accept (see Section~\ref{subsec:automata} for a more detailed explanation). (Also see \cite{FeuilloleyBP22}, and also \cite{ChangSS23} for a similar perspective in the context of local problems.)
Another point of view is the one of arithmetic. Intuitively, a constant size certification corresponds to a combination of congruence relations. This point of view allows, for example, to derive that the set of lengths of paths having a given constant size certification is eventually periodic. 

Now, moving on to non-constant certifications, we need to introduce a non-uniform automata model. Indeed, since larger paths imply that we can use larger certificates, it also means larger automata. 
For each certificate size $i$, we will have an automaton $\mathcal{A}_i$, and we will require that an incorrect instance is rejected by all these automata, while a correct instance is accepted by at least one (small enough) $\mathcal{A}_i$. 
If a property has non-constant complexity, then it means that we will need arbitrarily large automata.
If there exists a certification of size $s(n)$, then for any instance of length $n$, the automaton $\mathcal{A}_{s(n)}$ will have to accept.
Now from the arithmetic perspective, there is an equivalence between having non-constant complexity and the fact that the set of correct lengths is not eventually periodic. 

\paragraph*{Establishing gaps in paths and trees}

The reasoning to establish a gap between constant size and $\Theta(\log \log n)$ size in paths is the following.\footnote{Actually this proof gives worse constants than the ones of Theorem~\ref{thm:gap-paths-chrobak}, but it is the one that generalizes to other settings. } 
Consider a language $S$ that does not have a constant size certification. For an arbitrary $i$, we focus on the paths of $S$ that are recognized by $\mathcal{A}_i$ but not by $\mathcal{A}_1,...,\mathcal{A}_{i-1}$ (assume this is non-empty, which is the case for infinitely many $i$). This set is itself recognized by an automata: the intersection of $\mathcal{A}_i$ with the complement of the union of $\mathcal{A}_1,...,\mathcal{A}_{i-1}$. Studying the state complexity of this new automaton (which we do not discuss here) leads us to the fact that it must accept a path of length at most $2^{2^i}$. 
Finally, this upper bound on the minimum length for which some certificate size is needed translates into a lower bound constraint for the optimal complexity, and the double exponential translates into the double logarithm of the theorem. 

This proof can be generalized without much modification to labeled paths (that is, paths with inputs) and to larger verification radius (Theorem~\ref{thm:gap_labeled_distance>1}). This is pretty straightforward, given the previous proof: we use classic automata instead of unary automata, and a slightly different transformation from certification to automata. On the contrary, the proof that gives the right constants for Theorem~\ref{thm:gap-paths-chrobak} utilizes Chrobak normal form for unary automata (see Section~\ref{sec:Chrobak}), and does not generalize easily to the labeled case. 

The proof of the analogous gap in trees (Theorem~\ref{thm:gap-tree-radius-1}), is based on the same insights, but is much more involved. 
Basically, we adapt our automata to walk in the tree,  reading the pending subtrees as labels. Two challenges are that a given tree can be read in many different ways, and that the alphabet is infinite. We argue that the complexity will be captured by the walks that follow a maximum path in the tree, and that the intersection, complement and union that are used in the proof are not harmed by the infinite alphabet.  

\paragraph*{The $\log\log n$ language}

We now turn our attention to Theorem~\ref{thm:certif_product_primes}, which establishes the existence of a language with optimal certification size $\Theta(\log \log n)$ in anonymous unlabeled paths. 
Let us start by giving some intuition about how we came to this result. Intuitively, having optimal certificates of size $\Theta(\log \log n)$ means that certificates of size $k$ allow to differentiate between correct and incorrect instances of size order $2^{2^k}$. 
When we implement simple counters modulo some constant using $k$ bits, the modulo is of order $2^k$, hence we can distinguish between correct and incorrect instances of size $2^k$ but not more, in the sense that $q$ and $q+2^k$ will be classified in the same category. Hence we need to do an ``exponential jump'' in terms of distinguishing capability. 
We now make two observations. First, given a set of pairs $(r_i,m_i)_{i \le \ell}$, the fact that the path length is \emph{not equal} to $r_i \mod m_i$ for some $i \le \ell$ can be certified in $\max_i \log (m_i)$. Indeed, since it is sufficient to prove that one of the modulo equation is not satisfied, we can simply give explicitly $m_i$ in the all certificates, in addition to the counter modulo $M_i$. Second, using the Chinese remainder theorem, for any two numbers $a$ and $b$, there must exist an integer $m$ exponentially smaller that $\max (a,b)$ such that $a \not\equiv b \mod m$. 
This makes the existence of a $\Theta(\log \log n)$ language believable, but making it into a real construction is another kettle of fish. We end up considering the language of the path whose length \emph{is not} the product of consecutive primes. The scheme consists in certifying that either there is a gap or a repetition in the sequence of primes, via adequate counters, or there is an inconsistency between the largest prime used and what the length of the path is equal to, modulo a well-chosen (small) integer.  

\paragraph*{The case of cycles}

Theorem~\ref{thm:cycles} states that in cycles there is a gap between constant and logarithmic certification size. This is in sharp contrast with the case of paths, which is surprising given the similarity of the two structures. 
At a very intuitive level, the difference stems from the fact that the only congruence that we can check in cycles are of the form $0 \mod m$, and not arbitrary $i \mod m$. Indeed, in a path, it is easy to have one endpoint checking counter value zero and the other checking counter value $i$, while in cycles there are no endpoints. The natural way to simulate the endpoints is to designate a specific vertex $v$ in the cycle to check that the counter starts at $0$ from one side and reaches $i$ from the other side, but this is not robust. Indeed there could be more than one designated vertices and still all nodes could accept if each segment is $i \mod m$. This breaks the congruence, except if $i=0$.  
This restriction in the congruences prevents us from using the arithmetic tools of the proof of Theorem~\ref{thm:certif_product_primes}. 

Let us now sketch the proof of the gap between certificates of size $O(1)$ and $\Theta(\log n)$ in cycles. We again use an automata-like point of view, but need to adapt it to work without starting and ending nodes. We consider a sequence of directed certificate graphs\footnote{We consider ``graphs" and not ``automata" here since we have no initial or accepting states, a sequence of certificate only gives us a walk in the graph of pairs of certificates.} $(G_i)_i$. Here, a cycle $C$ with certificates of size~$i$ $(c_1,\ldots,c_n)$ is accepted by $G_i$ if and only if $(c_1,c_2),(c_2,c_3),\cdots,(c_n,c_1)$ is a closed walk in $G_i$. For the proof, we consider the first length $n_k$ that is accepted by the $k$-th graph $G_k$ (corresponding to certificates of size~$k$), but not by smaller ones. Now, if the certificate size is in $o(\log n)$, this walk is much longer than the size of $G_k$, and therefore it has a lot of cycles. We define a notion of elementary cycle of the graph, and decompose this walk as a linear combination of elementary cycles. If $d$ is the greatest common divisor of these elementary cycles, then $n_k=d\times q$, for some $q$. We consider a set of numbers of the form $p\times d$, such that (1) $p$ is prime, (2) $p\times d < n_k$ and (3) $p\times d$ is still much larger than the size of $G_k$. 
By a generalization of Bézout's identity, (3) implies that all these numbers are accepted by $G_k$ hence they belong to the language. But since they are smaller than $n_k$, by minimality, they must be recognized by smaller graphs too. We argue by pigeon-hole principle that there must exist one graph $G_{k'}$, $k'<k$ that has two walks of lengths $p_1d$ and $p_2d$ that intersect. Then by concatenating portions of the two walks, we can prove that $G_{k'}$ actually accepts all the cycles of size $ap_1d + bp_2d$. Finally, using again a generalization of Bézout's identity we can prove that the cycle of length $n_k$ is also recognized by $G_{k'}$, which contradicts the minimality of~$k$.   

\paragraph*{``No gap'' results}

We also establish that for several settings there is no gap in the complexity, that is, for any well-behaved function $f$, there exists a property that has certificates of size $O(f(n))$ but not $o(f(n))$. 

Let us sketch the technique of the upper bound for the setting of Theorem~\ref{thm:caterpillar-no-gap-radius-2}: restricting to caterpillars, using certificates whose size is a function of $d$, with verification radius 2.  
Given a function $f$, we define a sequence of integers $(b_k)$, such that the $k$-th term is roughly $f^{-1}(\log k)$. 
The correct instances are of the following form: a path of length $b_k$, for some $k$, such that the $i$-th node has $i$ leaves, and except for the first one, which has $b_k$ leaves instead of 1.
A compact way to certify these instances is to give $k$ to every node. The nodes in the middle check the growth of the number of attached leaves (thanks to their radius 2 verification) and the endpoints check that they have $b_k$ leaves. 
The diameter is $b_k$, and the number of bits used is $\log k$ which is basically~$f(d)$.  

\subsection{Discussions and open problems}
\label{sec:discussions}

Before we move on to the technical parts, let us discuss open problems and future directions. 

\subparagraph*{Full understanding of paths}
For paths, we do not know what happens between $\Theta(\log \log n)$ and $\Theta(\log n)$. By sparsifying the set of primes considered in the $\log \log n$ language (Theorem~\ref{thm:certif_product_primes}), we can get languages for which the natural upper bound can be positioned in between these two regimes, but Theorem~\ref{thm:gap-paths-chrobak} does not provide  a matching lower bound anymore, hence we cannot prove that there is no gap.
\begin{open}
	Is there a gap between $\Theta(\log \log n)$ and $\Theta(\log n)$ in paths? 
\end{open}

To solve this question, it would good to get a better understanding of the $\Theta(\log \log n)$ regime, or even a characterization. (For now, we just have one example of a property in this regime.)

\subparagraph*{General graphs}

For general graphs, we prove in Theorem~\ref{thm:general-graphs} that if the radius is larger than 1, then there is no gap. For radius 1, it is unclear whether we should expect the same gap as for trees or not. Our automata-related tools seem too weak to tackle this case (the generalization from trees to graphs in automata theory is notoriously intricate).

\begin{open}
	Is there a gap between $O(1)$ and $\Theta(\log \log d)$ for general graphs? For bounded degree graphs?
\end{open}

\subparagraph*{The role of identifiers and of the knowledge of $n$}

Our main results are for the anonymous setting, where the nodes do not have the knowledge of $n$. In Section~\ref{sec:extensions-ID}, we explore several settings with identifiers or (approximate) knowledge of $n$. We can for example prove that a very sharp estimate of $n$ allows to break the $\log\log n$ barrier of Theorem~\ref{thm:gap-paths-chrobak}, while arbitrarily large identifiers do not help. It is still very unclear how these different assumptions affect the complexity landscape.

\subparagraph*{Extensions to self-stabilizing algorithms}

We described in the introduction our wish to chart the space landscape of distributed graph algorithms, and as a first step we focused on local certification. A natural next step is to transfer our results to self-stabilizing algorithms. 
As mentioned earlier, the two are tightly connected, since the space used for silent self-stabilization is basically captured by the space needed to certify the solution correctness~\cite{BlinFP14}. 
Actually, \cite{BlinFP14} implements a transformation from a local certification to a self-stabilizing algorithm, that does require an additive $O(\log n)$ for the memory of the algorithm in comparison to the local certification size. This is usually harmless, but in the setting of this paper, which focuses on sublogarithmic regime, the result is unusable. The restricted topology might allow to shave the additional logarithmic term.  

\begin{open}
	Are the optimal certificate size and the optimal memory of a self-stabilizing algorithm asymptotically equal in restricted topologies (paths, cycles, trees)?
\end{open}

When it comes to transferring the sub-log-logarithmic gap for trees to (silent) self-stabilizing algorithms, it would actually be enough to understand whether labelings accepted by tree automata (which are equivalent to monadic second-order on trees) can be built in constant space in a self-stabilizing manner. Incidentally, understanding when one can make this transfer while keeping polynomial time is also a very intriguing question (see \cite{BlinDF20} for a discussion). 

\subparagraph*{Different landscapes}

Theorem~\ref{thm:folklore-no-gap} establishes that for any super-logarithmic complexity $f$ there exists a property that requires exactly that complexity. But the construction is very unnatural, since the nodes need to know the function $f$, and the instances are extremely specific. Hence the following question. 

\begin{open}
	Are there super-logarithmic gaps in the complexity of natural properties, for a reasonable definition of ``natural''? What about monadic second-order (MSO) properties?
\end{open}

It is known that there are properties for which the optimal certification size is $\Theta(n)$, for example having diameter $3$~\cite{Censor-HillelPP20,BousquetEFZ24} and also $\Theta(n^2)$~\cite{GoosS16}. 
As far as we know the only works about what happens in between are~\cite{BousquetCFPZ24} and \cite{Miyamoto2024}, that have established that for forbidden subgraphs, there are many polynomial complexities, when using verification radius 2. 

We mention MSO properties in the open problem because they have received a considerable amount of attention in recent years in local certification~\cite{BaterisnaC25, CookKM25,FeuilloleyBP22, FraigniaudMRT24, FraigniaudM0RT23}, and capture many classic properties and problems through logic. 

Finally, another research direction is to chart landscapes for other parameters. For example, \cite{BousquetFZ24} explored the certification complexity as a function of the maximum degree.

\subsection{Outline of the paper}

The organization of the paper is the following. After the model section (Section~\ref{sec:model}), we have a series of technical sections, each of them containing basically one theorem and its proof. 
These are arranged in the following way. First, the sections 3 to 6 are dedicated to paths and cycles, and contain most of the basic insights of our work. Then our more involved gap results are grouped in Section 8 to 10, and our ``no gap'' results are in Sections 11 to 14. Finally, Section~\ref{sec:extensions-ID} discusses the knowledge of $n$ and the identifier assumptions. See the table of contents after the title page for more specific pointers.



	\section{Model and definitions}
	\label{sec:model}
	
	\subsection{Graphs}
	
	All the graphs we consider are finite, simple, loopless, connected, and undirected. For completeness, let us recall several basic graph definitions. Let $G$ be a graph. We denote its set of vertices (resp. edges) by~$V(G)$ (resp. by~$E(G)$), or simply by~$V$ (resp. by~$E$) if~$G$ is clear from the context. Let $u,v \in V(G)$. The \emph{distance} between $u$ and $v$, denoted by~$d(u,v)$, is the smallest number of edges in a path from~$u$ to~$v$. The \emph{diameter} of~$G$ is the largest distance between any two vertices. If $G$ is a path, its \emph{length} is its number of vertices (or equivalently, its diameter plus one). We say that $G$ is a \emph{caterpillar} if, when removing all the degree-1 vertices in~$G$, the resulting graph is a path, called the \emph{central path} of~$G$ (which is induced by the set of all vertices of degree at least two in~$G$).
	
	\subsection{Local certification}
	
	
	Let $G=(V,E)$ be a graph. We will sometimes assume that the vertices of~$G$ are equipped with \emph{unique identifiers} and/or with \emph{inputs}. An \emph{identifier assignment} for $G$ is an \emph{injective} mapping from $V$ to some set~$I$ (the set of identifiers) and an \emph{input function} is a mapping from $V$ to some set $L$ (the set of labels). If $G$ is equipped with identifiers, we say that we are in the \emph{locally checkable proof model}, else we say that we are in the \emph{anonymous model}. If the vertices of $G$ have inputs, we say that $G$ is \emph{labeled}.
	Finally, let $C$ be a non-empty set. A \emph{certificate assignment} of $G$ with certificates in~$C$ is a mapping~$c : V \to C$.
	
	Let~$r \geqslant 1$ and $c$ be a certificate assignment for~$G$. Let $u \in V$. The~\emph{view of~$u$ at distance~$r$} consists in:
	\begin{itemize}
		\item all the vertices at distance at most $r$ from~$u$, and all the edges having at least one endpoint at distance at most~$r-1$,
		\item the restriction of~$c$ to these vertices,
		\item the restriction of the identifier assignment (if any) and of the input function (if any) to these vertices.
	\end{itemize}

	A~\emph{verification algorithm (at distance~$r$)} is a function taking as input the view at distance~$r$ of a vertex, and outputting a decision, \emph{accept} or~\emph{reject}. In all this paper, if~$r$ is not mentioned in a statement of a result, it is by default equal to~$1$. When we will consider settings where $r \geqslant 2$, it will always be explicitly written.
	
	Let $\C$ be a class of (possibly labeled) graphs and $\P$ be a property on graphs in~$\C$. Note that if we consider labeled graphs, the fact that a graph $G \in \C$ satisfies the property $\P$ does not depend only on the structure of $G$: it depends also on its input function. In other words, it is possible that a graph $G \in \C$ satisfies $\P$ and that another graph $G' \in \C$ with the same structure as~$G$ but with a different input function does not satisfy~$\P$. Let $s : \N \to \N$. We say that there exists a \emph{certification scheme for~$\P$} with certificates of size~$s$ if there exists a verification algorithm such that the two following conditions are satisfied:
	
	\begin{itemize}
		\item (Completeness) For every $n$-vertex graph~$G \in \C$ that satisfies~$\P$, and for every identifier assignment of $G$ (if we are in the locally checkable proof model), there exists a certificate assignment in $\{0, \ldots, 2^{s(n)}-1\}$ such that the verification of every vertex accepts (we says that the graph is globally accepted).
		\item (Soundness) For every graph~$G \in \C$ that does not satisfy~$\P$, for every identifier assignment of~$G$ (if we are in the locally checkable proof model), for every $k \in \N$ and every certificate assignment in $\{0, \ldots, 2^{k}-1\}$, at least one vertex rejects.
	\end{itemize}
	
	Let us emphasize that, if $G$ does not satisfy~$\P$, then for any assignment of certificates \emph{of any size}, at least one vertex rejects.
	Let us also point out the fact that in a certification scheme for a property~$\P$ in some class~$\C$ (in this paper, we will for instance consider the cases where $\C$ is the class of paths, of cycles, of trees...), the vertices have the promise that the graph belongs to~$\C$. In other words, the certification scheme depends on the property~$\P$ \emph{and} on the class~$\C$, and we are not concerned by the output of the verification procedure of the vertices in graphs that do not belong to~$\C$.
	
	\section{Gap between $O(1)$ and $\Theta(\log\log n)$ in paths}
	\label{sec:gap-path}
	
	
	
	The goal of this section is to prove the following result:
	
	\begin{restatable}{theorem}{ThmGapPaths}
		\label{thm:gap-paths-state-complexity}
		Let $c > 2$ and $N \in \N$. Let $\P$ be a property on paths that can be certified with certificates of size $s(n):=\left\lfloor\frac{\log \log n}{c}\right\rfloor$ for all $n \geqslant N$. Then, $\P$ can also be certified with constant-size certificates.
	\end{restatable}
	
	Note that the constant $c$ is larger here than in Theorem~\ref{thm:gap-paths-chrobak}. We prove Theorem~\ref{thm:gap-paths-chrobak} later in the paper (Section~\ref{sec:Chrobak}) with tools that are specific to unlabeled paths. We choose to prove this weaker version of the theorem first, because the techniques can be adapted to more general settings.

	\subsection{Preliminary: automata point of view}\label{subsec:automata}
	
	
	For every property $\P$ on paths, we can associate a subsets of integers $S$ such that a path is accepted if and only if its length is in $S$. The property $\P$ is equivalent to $S$ and in the rest of the proof, we will completely forget the property $\P$ and only focus on $S$.
	For every $k \in \N$, let us denote by $C_k$ the set of certificates of size~$k$, and by $S_k$ the set of lengths of the paths that are accepted with certificates in $C_k$. We have: $S = \bigcup_{k \in \N}S_k$.
	
	
	The set $S_k$ is a regular language that is accepted with the following nondeterministic finite automaton $\A_k$ over a unary alphabet. The set of states is the set of pairs of certificates of size~$k$ plus two additional states, $i$ and $f$, that is: $C_k^2 \cup \{i,f\}$. There is a single initial state which is~$i$ and a single final state which is~$f$. The transitions are the following:
	\begin{itemize}
		\item for every $c_1,c_2 \in C_k$, we put a transition between states $i$ and $(c_1,c_2)$ if a vertex of degree~1 that has certificate $c_1$ and has a neighbor with certificate $c_2$ accepts;
		\item for every $c_1,c_2,c_3 \in C_k$, we put a transition between states $(c_1,c_2)$ and $(c_2,c_3)$ if a vertex of degree~2 that has certificate $c_2$ and has two neighbors with certificates $c_1$ and $c_3$ accepts;
		\item for every $c_1, c_2 \in C_k$, we put a transition between states $(c_1,c_2)$ and $f$ if a vertex of degree~1 that has certificate $c_2$ and has a neighbor with certificate $c_1$ accepts;
		\item if there exists $c \in C_k$ such that an isolated vertex with certificate~$c$ accepts, we put a transition from~$i$ to~$f$.
	\end{itemize}
	
	Let us give an example to make things more concrete. Assume that we want to certify that the length of a path is divisible by~$3$. There is an easy way to do it by using three certificates 0, 1, and 2. The prover fixes an endpoint~$u$ and for every vertex~$v$, the certificate it gives to $v$ is $d(u,v) \mod 3$. Then, every vertex $v$ checks that one of the two following conditions is satisfied:
	\begin{itemize}
		\item $v$ has degree~$1$, has certificate~$0$ or~$2$, and its neighbor has certificate~$1$, or
		\item $v$ has degree~$2$, and the set of certificates of $v$ and its two neighbors is $\{0,1,2\}$.
	\end{itemize}
	
	The automaton corresponding to these certificates is represented on Figure~\ref{fig:example_automaton}.
	
	\begin{figure}[h!]
		\centering
		\begin{tikzpicture}
			\scriptsize
			\node[state, initial, minimum size=0.7cm] at (0, 0) (i) {$i$};
			\node[state, minimum size=0.7cm] at (1.5, 0.5) (01) {$0,1$};
			\node[state, minimum size=0.7cm] at (3, 0.5) (12) {$1,2$};
			\node[state, minimum size=0.7cm] at (4.5, 0.5) (20) {$2,0$};
			\node[state, minimum size=0.7cm] at (1.5, -0.5) (02) {$0,2$};
			\node[state, minimum size=0.7cm] at (3, -0.5) (10) {$1,0$};
			\node[state, minimum size=0.7cm] at (4.5, -0.5) (21) {$2,1$};
			\node[state, accepting, minimum size=0.7cm] at (6, 0) (f) {$f$};
			
			\draw[->, >=stealth]
			(i) edge node{} (01)
			(i) edge[bend right=80] node{} (21)
			(01) edge node{} (12)
			(12) edge node{} (20)
			(20) edge[bend right=40] node{} (01)
			(21) edge node{} (10)
			(10) edge node{} (02)
			(02) edge[bend right=40] node{} (21)
			(12) edge[bend left=80] node{} (f)
			(10) edge[bend right=80] node{} (f)
			;
		\end{tikzpicture}
		\caption{The automaton corresponding to the certificates used to certify that the length of a path is divisible by~$3$. The states corresponding to the tuples $(0,0)$, $(1,1)$ and $(2,2)$ are not represented because the have no incoming nor outgoing transitions. The final state is the state $f$.}
		\label{fig:example_automaton}
	\end{figure}
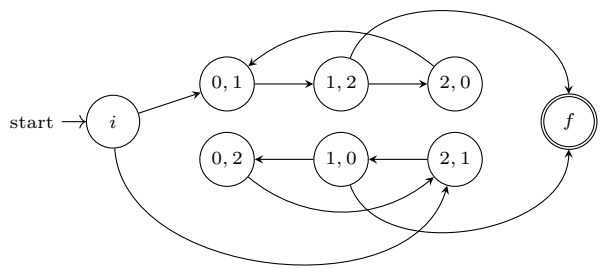

	\begin{lemma}
		For every $t \in \N$, a path on~$t$ vertices is accepted with certificates of size~$k$ if and only if there exists an accepting run of length~$t$ (ie, going through~$t$ transitions) in~$\A_k$.
	\end{lemma}
	
	\begin{proof}
		Assume that a path of length~$t$ is accepted with certificates of size~$k$. Let us denote the vertices of the path by~$u_1, \ldots, u_t$, and their certificates by~$c_1, \ldots, c_t$. If~$t=1$, by definition there is a transition from~$i$ to~$f$ so there is an accepting run of length~$1$ in~$\A_k$. If $i \geqslant 2$, since $u_1$ accepts, there is a transition from~$i$ to~$(c_1, c_2)$. For every $i \in \{2, \ldots, t-1\}$, since $u_i$ accepts, there is a transition from~$(c_{i-1},c_i)$ to~$(c_i,c_{i+1})$. And since~$u_t$ accepts, there is a transition from $(c_{t-1},c_t)$ to~$f$. Thus, there exists an accepting run of length~$t$ in~$\A_k$.
		
		Conversely, every accepting run of length~$t$ in~$\A_k$ can be converted in an assignment of certificates to a path on~$t$ vertices such that each vertex accepts.
	\end{proof}

	\subsection{Proof via state complexity}

	\begin{lemma}
		\label{lem:automata}
		Let $\Sigma$ be a (possibly infinite) alphabet, and let $L_\mathcal{A}, L_\mathcal{B}$ be two languages over $\Sigma^\ast$ recognized by nondeterministic finite automata $\mathcal{A}$ and $\mathcal{B}$, having $n_\mathcal{A}$ and $n_\mathcal{B}$ states respectively.
		Then:
		\begin{itemize}
			\item $L_\mathcal{A} \cup L_\mathcal{B}$ can be recognized by an automaton having $n_\mathcal{A} + n_\mathcal{B}$ states
			\item $L_\mathcal{A} \cap L_\mathcal{B}$ can be recognized by an automaton having $n_\mathcal{A}n_\mathcal{B}$ states
			\item $\overline{L_\mathcal{A}}$ can be recognized by an automaton having $2^{n_\mathcal{A}}$ states
		\end{itemize}
	\end{lemma}
	
	These statements are folklore. For the first point, one simply has to take the disjoint union of the two automata. For the second, we create an automaton whose states are the pairs of states of $\mathcal{A}$ and $\mathcal{B}$ and we add a transition labeled $\sigma$ from $(a,b)$ to $(a',b')$ if there is a transition labeled $\sigma$ from $a$ to $a'$ in $\mathcal{A}$ and from $b$ to $b'$ in $\mathcal{B}$. Finally, to get the last result, we first have to make the automaton deterministic (which might need an exponential explosion~\cite{RabinS59}) and then flip final states and non-final states.
	
	\begin{proof}[Proof of Theorem~\ref{thm:gap-paths-state-complexity}.]
		Using the notations introduced previously, the automaton $\A_k$ has $M_k := 2^{2k} + 2$ states and recognizes $S_k$. Let $\P$ be a property that can not be recognized with constant-size certificates. Then, the set $X \subseteq \N$ containing all the integers $k \in \N$ such that $S_k \not\subseteq \bigcup_{i \leqslant k-1} S_i$ is infinite. For $k \in X$, let $n_k$ be the smallest integer in $S_k \setminus \bigcup_{i \leqslant k-1} S_i$. Let $k \in X$ be such that $n_k \geqslant N$ (such an integer $k \in X$ exists because $X$ is infinite and for all distinct $k, k' \in X$ we have $n_k \neq n_{k'}$).
		
		Since $n_k \in S_{s(n_k)}$ and $n \notin S_i$ for $i<k$, we have $s(n_k) \geqslant k$. By Lemma~\ref{lem:automata}, $\bigcup_{i\leqslant k-1} S_i$ can be recognized by an automaton that has $\sum_{i=1}^{k-1}M_i \leqslant 2^{2k}$ states. Thus, by Lemma~\ref{lem:automata}, $\overline{\bigcup_{i\leqslant k-1} S_i}$ can be recognized by an automaton that has at most $2^{2^{2k}}$ states.
		Again by Lemma~\ref{lem:automata}, $S_k \setminus \bigcup_{i \leqslant k-1} S_i$ can be recognized by an automaton having at most $M_k \cdot 2^{2^{2k}}$ states. Since $n_k$ is the smallest integer in $S_k \setminus \bigcup_{i \leqslant k-1} S_i$, it is at most equal to the number of states of this automaton, so it follows that $n_k \leqslant M_k \cdot 2^{2^{2k}} \leqslant 2^{2^{2k+1}}$. Finally we get $s(n_k) \geqslant k \geqslant \frac{\log \log n_k - 1}{2}$, and the result follows.
		%
		%
	\end{proof}
	
	
	\subparagraph*{Pointers to generalizations}
	
	In Section~\ref{sec:gap-labeled-larger-radius}, we prove the more general following extension to labeled paths (that is, paths with inputs) and to larger verification radius. 
	
	\begin{restatable}{theorem}{ThmGapLabeledDistance}
		\label{thm:gap_labeled_distance>1}
		Let $N \in \N$, $\Sigma$ be a finite alphabet, $r \geqslant 1$ and $c > 2$. Assume that the vertices can see at distance~$r$.
		Let $\P$ be a property on labeled paths (where the labels are letters in $\Sigma$) that can be certified with certificates of size~$s(n):=\left\lfloor\frac{\log \log n}{cr}\right\rfloor$ for every $n \geqslant N$. Then $\P$ can also be certified with certificates of constant size.
	\end{restatable}
	
	We have already mentioned the generalization to trees (Theorem~\ref{thm:gap-tree-radius-1}), whose proof can be found in Section~\ref{sec:gap-tree}.
	
	\section{A property with optimal size  $\Theta(\log \log n)$ in unlabeled paths }
	\label{sec:loglogn-property}
	
	The goal of this Section is to prove the following theorem.
	
	\ThmCertifProductPrimes*
	
	Before proving Theorem~\ref{thm:certif_product_primes}, let us show the following result:
	
	\begin{lemma}
		\label{lem:certif_modulo}
		Let $m, t \in \mathbb{N}$ and $m \geqslant 2$. Certifying that the length a path on~$n$ vertices satisfies \mbox{$n \equiv t \mod m$} can be done with certificates of size $O(\log m)$.
	\end{lemma}
	
	\begin{proof}
		The certificate given by the prover to the vertices consists in two parts. For each vertex~$u$, we denote the first part of its certificate by $\textsf{Orientation}[u]$, and the second part by $\textsf{Distance}[u]$. The prover chooses an endpoint of the path, and for each vertex $u$, we denote by $d_u$ its distance to this endpoint. In $\textsf{Orientation}[u]$, the prover writes $d_u \mod 3$. In $\textsf{Orientation}[u]$, the prover writes $d_u \mod m$. This certificate has size $O(\log m)$.
		
		To check the correctness of the certificates, each vertex first checks the correctness of the part $\textsf{Orientation}$. To do so, if $u$ has degree~$2$, it simply checks that it has a neighbor~$v$ such that $\textsf{Orientation}[v] = \textsf{Orientation}[u]+1 \mod 3$, and a neighbor $w$ such that $\textsf{Orientation}[w] = \textsf{Orientation}[u]-1 \mod 3$. And if $u$ has degree~$1$ and has a neighbor $v$, it checks that either $\textsf{Orientation}[u] = 0$ and $\textsf{Orientation}[v]=1$ (in this case we say that $u$ is the \emph{beginning endpoint}), or $\textsf{Orientation}[v] = \textsf{Orientation}[u]-1$ (in this case we say that $u$ is the \emph{final endpoint}). This part of the certificates gives an orientation of the path to the vertices: for two neighbors~$u$ and~$v$, we say that $u$ is a~\emph{predecessor} of~$v$ (resp. a \emph{successor}) if $\textsf{Orientation}[u]=\textsf{Orientation}[v]-1 \mod 3$ (resp. $\textsf{Orientation}[u]=\textsf{Orientation}[v]+1 \mod 3$). Then, to check the correctness of the part $\textsf{Distance}$, each degree-$2$ vertex checks that its predecessor~$v$ satisfies $\textsf{Distance}[v] = \textsf{Distance}[u]-1 \mod m$, and its successor~$w$ satisfies $\textsf{Distance}[w] = \textsf{Distance}[u]+1 \mod m$. If $u$ is the beginning endpoint, it checks that $\textsf{Distance}[u] = 0$. If $u$ is the final endpoint and has a predecessor~$v$, it checks that $\textsf{Distance}[u] = \textsf{Distance}[v] + 1 \mod m$, and that $\textsf{Distance}[u] = t$.
	\end{proof}
	
	\begin{remark}
		\label{rem:certif_modulo}
		For every $m, t \in \mathbb{N}$ with $m \geqslant 2$, we can also certify that the length~$n$ of a path satisfies $n \not\equiv t \mod m$ with certificates of size $O(\log m)$, with the same proof (just by replacing $\textsf{Distance}[u] = t$ by $\textsf{Distance}[u] \neq t$ at the end). In particular, with certificates of size $O(\log m)$, we can certify that~$m$ divides~$n$, or that~$m$ does not divide~$n$.
	\end{remark}

	Finally, let us introduce some notations and give some useful properties.
	For every $k \geqslant 1$, let us denote by $p_k$ the $k$-th prime number (i.e. $p_1 = 2, p_2 = 3, p_3 = 5$...), and let $a_k$ be the product of the $k$ first prime numbers:
	$a_k := \prod_{i=1}^{k}p_i$. Let $S \subseteq \mathbb{N}$ be the set $\{a_k \; | \; k \geqslant 1\}$.
	
	\begin{lemma}{\cite{HardyW79}}
		\label{lem:bound}
		We have $p_k = \Theta(k \log k)$ and $a_k = 2^{k \log k (1 + o(1))}$.
	\end{lemma}
	
	\begin{lemma}
		\label{lem:small_prime_certif}
		There exists $c > 0$ such that, for every even integer $n \geqslant 2$, there exists $k \leqslant c \log n$ such that $p_k$ divides~$n$ and $p_{k+1}$ does not divide~$n$.
	\end{lemma}
	
	\begin{proof}
		Let $n$ be an even integer and let $k$ be the smallest integer such that $p_k$ divides $n$ and $p_{k+1}$ does not divide $n$ (which exists because $n$ is divisible by~$p_1 = 2$). Then, $n$ is divisible by $a_k$, so $a_k \leqslant n$. Using Claim~\ref{lem:bound}, we get $2^{k \log k (1 + o(1))} \leqslant n$, so $k \log k (1+o(1)) \leqslant \log n$, and the result follows.
	\end{proof}
	
	\begin{lemma}
		\label{lem:not_equiv}
		Let $1 \leqslant s < t$. There exists $m \leqslant \lceil \log t \rceil$ such that $s \not\equiv t \mod m$.
	\end{lemma}
	
	\begin{proof}
		By contradiction, assume that for all $m \leqslant \lceil \log t \rceil$, we have $s \equiv t \mod m$. Then, by the Chinese remainder theorem, we get $s \equiv t \mod p$, where $p$ is the least common multiple of $1, 2, \ldots, \lceil \log t \rceil$. It is well-known that this least common multiple is at least~$t$, so we have $1 \leqslant s < t \leqslant p$. Together with $s \equiv t \mod p$, this implies $s = t$, which contradicts the assumption $s < t$.
	\end{proof}
	
	\begin{proposition}
		\label{prop:main}
		Let $n \geqslant 1$, and $c$ be the constant of Lemma~\ref{lem:small_prime_certif}. Then, $n \notin S$ if and only if at least one of the three following conditions is satisfied:
		\begin{enumerate}
			\item n is odd, or
			\item there exists $1 \leqslant \ell < k \leqslant c \log n$ such that $p_\ell$ does not divide~$n$ and $p_k$ divides~$n$, or
			\item there exists $1 \leqslant k \leqslant c \log n$ and $1 \leqslant m \leqslant \lceil \log n \rceil$ such that $p_k$ divides~$n$, $p_{k+1}$ does not divide~$n$ and $n \not\equiv a_k \mod m$.
		\end{enumerate}
	\end{proposition}
	
	\begin{proof}
		First, assume that one of the three conditions is satisfied. If condition~1. holds, $n$ is odd, so $n \notin S$ because $S$ contains only even integers. If condition~2. holds, then $n \notin S$ because all the integers in~$S$ which are divisible by~$p_k$ are also divisible by $p_\ell$ for all $\ell \leqslant k$. If condition~3. holds, then $n \notin S$, because the only integer in $S$ which is divisible by~$p_k$ and not by~$p_{k+1}$ is~$a_k$.
		
		Conversely, assume that $n \notin S$, and let us show that at least one of the three conditions is satisfied. Assume that conditions~1. and~2. are not satisfied, and let us show that condition~3. holds. Since condition~1. is not satisfied, $n$ is even, so by Lemma~\ref{lem:small_prime_certif}, there exists $k \leqslant c \log n$ such that $p_k$ divides~$n$ and $p_{k+1}$ does not divide~$n$. Since condition~2. is not satisfied, $n$ is divisible by~$a_k$, so $a_k < n$ (this inequality is strict because, by assumption, $n \notin S$). Finally, by Lemma~\ref{lem:not_equiv}, there exists $m \leqslant \lceil \log n \rceil$ such that $n \not\equiv a_k \mod m$, so condition~3. is satisfied. 
	\end{proof}
	
	We are now able to prove Theorem~\ref{thm:certif_product_primes}.
	
	\begin{proof}[Proof of Theorem~\ref{thm:certif_product_primes}]
		Recall that we assume that the input graph is a path $P$. 
		Let $\P$ be the property of being a path whose length is \emph{not} in~$S$.
		First, observe that $\P$ cannot be certified with constant-size certificates. Indeed, properties on paths that can be certified with constant size-certificates are paths whose length is in a set that is eventually periodic (the most simple proof for it uses Chrobak normal form, see Observation~\ref{obs:periodic}), and the set $\overline{S}$ is not.
		Now, let us show that $\P$ can be certified with certificates of size $O(\log \log n)$.
		
		Let $n \geqslant 1$. If $n \notin S$, the prover certifies that at least one of the three conditions of Proposition~\ref{prop:main} is satisfied. More precisely:
		\begin{itemize}
			\item if $n$ is odd, the prover certifies it. By Lemma~\ref{lem:certif_modulo}, this needs~$O(1)$ bits.
			\item if there exists $1 \leqslant \ell < k \leqslant c \log n$ such that $p_k$ divides~$n$ and $p_\ell$ does not divide~$n$, the prover writes $k$ and $\ell$ in the certificate of each vertex, and certifies that $p_k$ divides~$n$ and that $p_\ell$ does not divide~$n$. Since $k \leqslant c \log n$ and $p_k = \Theta(k \log k)$, by Lemma~\ref{lem:certif_modulo} and Remark~\ref{rem:certif_modulo}, this needs~$O(\log \log n)$~bits.
			\item if there exists $1 \leqslant k \leqslant c \log n$ and $1 \leqslant m \leqslant \lceil \log n \rceil$ such that $p_k$ divides~$n$, $p_{k+1}$ does not divide~$n$ and $n \not\equiv a_k \mod m$, the prover writes~$k$ and~$m$ in the certificate, certifies that $p_k$ divides~$n$, $p_{k+1}$ does not divides~$n$ and that $n \not\equiv a_k \mod m$. By Lemma~\ref{lem:certif_modulo} and Remark~\ref{rem:certif_modulo}, this needs~$O(\log \log n)$~bits.
		\end{itemize}
		
		The verification of the vertices just consists in checking that the condition given by the prover is indeed satisfied, with the verification procedure of Lemma~\ref{lem:certif_modulo}. Note that the vertices do not need to check the bounds on~$k$ and~$m$ for conditions~2. and~3., because if condition~2. or~3. is satisfied with larger~$k$ or~$m$, it also implies that $n \notin S$ (and then $P \in \P$ and $P$ can be also accepted with certificates of size $O(\log \log n)$).
		These bounds on~$k$ and~$m$ are only useful to get an upper bound on the size of the certificates.
	\end{proof}

	\section{Gap between $O(1)$ and $\Theta(\log n)$ in cycles}
	\label{sec:gap-cycle}
	
	In this section we prove the following theorem.
	
	\ThmCycles*
	
	
	
	
	
	\subsection{Preliminaries on number theory and walks in graphs}
	
	Let us give some results from number theory on which we will rely. First, let us recall the prime number theorem:
	
	\begin{theorem}[Prime Number Theorem]
		\label{thm:PNT}
		For $n \in \N$, let $\pi(n)$ be the number of prime numbers in $\{1, \ldots, n\}$. Then:$$\pi(n) \sim \frac{n}{\ln(n)}$$
	\end{theorem}
	
	From Theorem~\ref{thm:PNT}, we can deduce the immediate following corollary:
	
	\begin{corollary}
		\label{cor:PNT}
		Let $c > 12$. For every $n \in \N$, let $\pi_c(n)$ be the number of prime numbers~$p$ such that $2^{4n+1} < p \leqslant 2^{n(c/2-2)}$. Then, there exists $n_0\in \N$ such that for all $n \geqslant n_0$, $\pi_c(n) > n 2^{2n}$.
	\end{corollary}
	
	\begin{proof}
		For every $n \in \N$, we have $\pi_c(n) = \pi(2^{n(c/2-2)}) - \pi(2^{4n+1})$. Since $c > 12$, by Theorem~\ref{thm:PNT}, we have $\pi(2^{4n+1}) = o(\pi(2^{n(c/2-2)}))$. Thus, $\pi_c(n) \sim \pi(2^{n(c/2-2)})$. By applying again the prime number theorem, we get $\pi_c(n) \sim 2^{n(c/2-2)}/(n(c/2-2))$. Thus, $n2^{2n} = o(\pi_c(n))$, and the result follows.
	\end{proof}
	
	Finally, let us give the following generalization of Bézout's identity that will be useful in the proof of Theorem~\ref{thm:cycles}:
	
	\begin{lemma}[\cite{brauer42}]
		\label{lem:bezout}
		Let $\ell_1, \ldots, \ell_t$ be positive integers. Let $m := \max(\ell_1, \ldots, \ell_t)$ and $d := \gcd(\ell_1, \ldots, \ell_t)$.
		Then, for every integer $n \geqslant m^2$ which is divisible by~$d$, there exists non-negative integers $a_1, \ldots a_t$ such that
		$\sum_{i=1}^{t} a_i \ell_i = n$.
	\end{lemma}
	
	We now move on to some definitions about walks in graphs.
	A~\emph{closed walk} in $\G_k$ is a directed path that begins and ends in the same vertex (it is allowed to pass through the same vertex or the same edge multiple times). 
	The \emph{length} of a closed walk is its number of edges. A closed walk that does not pass through the same vertex twice (except for the starting and ending vertices which are the same) is called an \emph{elementary cycle}. If $\C$ and $\C'$ are two directed closed walks, we say that $\C'$ is a closed subwalk of $\C$ if a subsequence of vertices in $\C$ is equal to $\C'$. See Figure~\ref{fig:dircycl} for an example. Note that the length of an elementary closed walk in~$\G_k$ is at most equal to the number of vertices of $\G_k$ which is $2^{2k}$.
	
	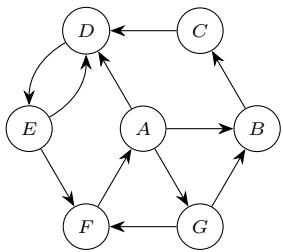
\begin{figure}[h]
		\centering
		\begin{tikzpicture}
			
			\newdimen\R
			\R=1.5cm
			
			\node[circle, draw, minimum size=0.5cm] at (0, 0) (0) {\scriptsize $A$};
			\node[circle, draw, minimum size=0.5cm] at (0:\R) (1) {\scriptsize $B$};
			\node[circle, draw, minimum size=0.5cm] at (60:\R) (2) {\scriptsize $C$};
			\node[circle, draw, minimum size=0.5cm] at (120:\R) (3) {\scriptsize $D$};
			\node[circle, draw, minimum size=0.5cm] at (180:\R) (4) {\scriptsize $E$};
			\node[circle, draw, minimum size=0.5cm] at (240:\R) (5) {\scriptsize $F$};
			\node[circle, draw, minimum size=0.5cm] at (300:\R) (6) {\scriptsize $G$};
			
			\draw[-{Stealth[scale=1.2]}] (6) to (5);
			\draw[-{Stealth[scale=1.2]}] (0) to (1);
			\draw[-{Stealth[scale=1.2]}] (0) to (3);
			\draw[-{Stealth[scale=1.2]}] (2) to (3);
			\draw[-{Stealth[scale=1.2]}] (1) to (2);
			\draw[-{Stealth[scale=1.2]}] (3) to[bend right] (4);
			\draw[-{Stealth[scale=1.2]}] (4) to[bend right] (3);
			\draw[-{Stealth[scale=1.2]}] (4) to (5);
			\draw[-{Stealth[scale=1.2]}] (5) to (0);
			\draw[-{Stealth[scale=1.2]}] (6) to (1);
			\draw[-{Stealth[scale=1.2]}] (0) to (6);
		\end{tikzpicture}
		\caption{
			In this directed graph, $ADEDEFA$ is a closed walk of length~$6$, which is not an elementary cycle. The closed walks $FAGF$ and $CDEFABC$ are elementary cycles and have length~$3$ and~$6$ respectively. Moreover, $FAGF$ is a subwalk of $DEFAGFAD$.}
		\label{fig:dircycl}
	\end{figure}

	\subsection{Proof of Theorem~\ref{thm:cycles}}
	
	All this subsection is devoted to the proof of Theorem~\ref{thm:cycles}.
	
	\medskip
	
	Let $c > 12$ and $N \in \N$. Let $S$ be the set of lengths of cycles in~$\P$.  Assume by contradiction that there exists a property $\P$ such that, for every integer $n \geqslant N$ satisfying $\P$, the cycle of length $n$ is accepted with certificates of size~$s(n)$, where $s(n) \leqslant \left\lfloor
	\frac{\log n}{c}\right\rfloor$, and that a constant number of certificates are not sufficient to certify~$\P$. 
	
	For every $k \geqslant 1$, let $C_k$ be the set of certificates of size~$k$, and let $S_k$ be the subset of $S$ corresponding to the cycles accepted with certificates in~$C_k$.
	Note that $|C_k| = 2^k$, and that we have $S = \bigcup_{k \in \N} S_k$.
	Let $\G_k = (V_k,E_k)$ be the directed graph where $V_k := C_k^2$ and for all $a, b, c \in C_k$, there is an edge in $E_k$ between $(a,b)$ and $(b,c)$ if and only if a degree-2 vertex with certificate $b$ has a neighbor with certificate~$a$ and another neighbor with certificate~$c$ accepts (and there are no other edges in~$E_k$, that is there is no edge between $(a,b)$ and $(c,d)$ if $b \ne c$). 
	The directed graph $\G_k$ has $2^{2k}$ vertices.

	\begin{claim}
		\label{claim:setcycle}
		For every integer $n \geqslant 3$, $n \in S_k$ if and only there exists a closed walk of length~$n$ in $\G_k$.
	\end{claim}
	
	\begin{proof}
		If $n \in S_k$, there exists an assignment of certificates $c_1, \ldots, c_n$ to the vertices of a $n$-vertex cycle such that every vertex accepts. For every $i \in \{1, \ldots, n\}$, since the vertex with certificate~$c_i$ accepts and has two neighbors with certificates $c_{i-1}$ and $c_{i+1}$ (where $i-1$ and $i+1$ are taken modulo~$n$), by definition there is an edge in $\G_k$ from $(c_{i-1}, c_i)$ to $(c_i, c_{i+1})$. So this gives an closed walk of length $n$ in~$\G_k$. Conversely, if there exists a closed walk $(c_1, c_2), (c_2, c_3), \ldots, (c_n,c_1)$ in $\G_k$, by definition all the vertices of the cycle of length~$n$ with certificates $c_1, \ldots, c_n$ accept, so $n \in S_k$.
	\end{proof}
	
	By Corollary~\ref{cor:PNT}, there exists $k_0\in \N$ such that for all $k \geqslant k_0$, $\pi_c(k) > k 2^{2k}$.
	
	Since $S$ is not accepted with a constant number of certificates, the set $X \subseteq \N$ of integers $k \in \N$ such that $S_k \not \subseteq \bigcup_{1 \leqslant i < k}S_i$ is infinite. For $k \in X$, let $n_k$ be the smallest integer $S_k \setminus \bigcup_{1 \leqslant i < k}S_i$. Finally, let us fix an integer integer~$k \in X$, such that $k \geqslant k_0$ and $n_k \geqslant N$ (such an integer $k \in X$ exists because $X$ is infinite and for all distinct $k, k' \in X$ we have $n_k \neq n_{k'}$).
	
	\begin{claim}
		\label{claim:lower_bound_n}
		We have $n_k \geqslant 2^{ck}$.
	\end{claim}
	
	\begin{proof}
		By definition of $s(n_k)$, we have $n_k \in S_{s(n_k)}$. Since $n_k \in S_k \setminus \bigcup_{1 \leqslant i < k}S_i$, it follows that $s(n_k) \geqslant k$. Moreover, by assumption, $s(n_k) \leqslant \left \lfloor \frac{\log n_k}{c} \right \rfloor$. So $n_k \geqslant 2^{ck}$.
	\end{proof}
	
	Since $n_k \in S_k$, by Claim~\ref{claim:setcycle}, there is a closed walk of length~$n_k$ in $\G_k$. Let us consider the strongly connected component $\G_k'$ of $\G_k$ containing this closed walk. Let $t$ be the number of elementary cycles in $\G_k'$ that we denote by $\C_1, \ldots, \C_t$, and let $\ell_1, \ldots, \ell_t$ be their lengths. Let $d = \gcd(\ell_1, \ldots ,\ell_t)$. We have $d \leqslant 2^{2k}$, because we have $\ell_i \leqslant 2^{2k}$ for every $i \in \{1, \ldots, t\}$ (since $\G_k$ has size $2^{2k}$).
	
	\begin{claim}
		\label{claim:elementary}
		Let $\C$ be a closed walk in $\G_k'$, and let $\ell$ be its length. There exists $b_1, \ldots, b_t \in \N$ such that $\ell = \sum_{i=1}^{t} b_i \ell_i$. Thus, $d$ divides~$\ell$. In particular, $d$ divides~$n_k$.
	\end{claim}
	
	
	\begin{proof}
		Let us prove this result by induction on the length of $\C$.
		The base case includes all elementary cycles: if $\C$ is an elementary cycle, there exists $j \in \{1, \ldots, t\}$ such that $\ell = \ell_j$, so the result is trivially true. Assume now that $\C$ is not an elementary cycle, and consider the shortest closed subwalk~$\C'$ of~$\C$. Then, $\C'$ is an elementary cycle (otherwise it would not be the shortest subwalk of~$\C$), so $\C' \in \{\C_1, \ldots, \C_t\}$ and its length is equal to $\ell_j$ for some $j \in \{1, \ldots, t\}$.
		Let us denote by $\C \setminus \C'$ the closed walk obtained  by removing from $\C$ the steps of $\C'$. The length of $\C \setminus \C'$ is $\ell - \ell_j$. Finally, apply the induction hypothesis to the closed walk $\C \setminus \C'$, to obtain integers $b_1, \ldots, b_t \in \N$ such that $\ell - \ell_j = \sum_{i=1}^t b_i \ell_i$. The result follows.
	\end{proof}
	
	\begin{claim}
		\label{claim:lengthcycle}
		Let $m \in \N$ be such that $d$ divides~$m$, and $m \geqslant 2^{4k+1}$. Then, there exists a closed walk in $\G_k'$ of length~$m$. Thus, $m \in S_k$.
	\end{claim}
	
	\begin{proof}
		First, we construct greedily a closed walk~$\C_0$ in~$\G_k'$ that passes through all the vertices of $\G_k'$ (it exists, because $\G_k'$ is strongly connected).
		For every $u, v \in V(\G_k')$, the shortest directed path from~$u$ to~$v$ has length at most the number of vertices of $\G_k$ which is $2^{2k}$. Thus, there exists a closed walk~$\C_0$ of length $\ell_0 \leqslant (2^{2k})^2 = 2^{4k}$ that passes through all the vertices.
		By Claim~\ref{claim:elementary}, $d$ divides~$\ell_0$, so $d$ divides $m-\ell_0$. Furthermore, $m-\ell_0 \geqslant 2^{4k}$. Since we have $\max_{1 \leqslant i \leqslant t} \ell_i \leqslant 2^{2k}$, we can apply Lemma~\ref{lem:bezout} to get the existence of integers $a_1, \ldots, a_t \in \N$ such that $m - \ell_0 = \sum_{i=1}^{t} a_i \ell_i$. Finally, to construct a closed walk in $\G_k'$ of length $m = \ell_0 + \sum_{i=1}^{t} a_i \ell_i$, we attach $a_i$ times the elementary cycle $\C_i$ to the closed walk $\C_0$ for every $i \in \{1, \ldots, t\}$ (this is possible, because $\C_0$ passes through all the vertices). By Claim~\ref{claim:setcycle}, we have $m \in S_k$.
	\end{proof}
	
	Let us now combine these arguments to prove Theorem~\ref{thm:cycles}. Before giving the technical details, let us explain the intuition. Let us denote by $d$ the $d:=gcd$ of all the lengths of cycles in $S_k$. Since $\mathcal{G}_k$ has size $2^{2k}$, Lemma~\ref{claim:lengthcycle} ensures that all the cycles of length $r \cdot d$ are in $\mathcal{P}$ when $r$ is large enough (but small compared to $n_k$). Thus there exist many prime numbers $p$ such that $pd$ are in $\mathcal{P}$ and $pd \le n_k$. By definition of $n_k$, at least two of them can be certified with the same set of bits and we can obtain a contradiction. Let us now formalize the argument.
	
	Recall that $k$ is an integer such that $k \geqslant k_0$ and $n_k \geqslant N$.
	Let $p$ be a prime number such that $2^{4k+1} < p \leqslant 2^{k(c/2-2)}$. By Claim~\ref{claim:lengthcycle}, we have $pd \in S_k$.
	Moreover, since $p \leqslant 2^{k(c/2-2)}$ and $d \leqslant 2^{2k}$, we have $pd \leqslant 2^{kc/2}$, that is, $pd \leq \sqrt{n_k}$ using Claim~\ref{claim:lower_bound_n}.
	Since $n_k$ is the smallest integer in $S_k \setminus \bigcup_{1 \leqslant i<k} S_i$, we have $pd \in \bigcup_{1 \leqslant i<k} S_i$. For every $i \in \{1, \ldots, k-1\}$, let $X_i$ be the set of prime numbers~$p \in \{2^{4k+1}+1, \ldots, 2^{k(c/2-2)}\}$ such that $pd \in S_i$.
	Since there are $\pi_c(k)$ prime numbers in $\{2^{4k+1}+1, \ldots, 2^{k(c/2-2)}\}$ and we have $\pi_c(k) > k2^{2k}$, by the pigeonhole principle there exists $i < k$ such that $|X_i| > 2^{2k}$. Let us fix this index~$i$.
	
	For every $p \in X_i$, since $pd \in S_i$, there exists a closed walk~$\C^{(p)}$ of length~$pd$ in $\G_i$. Since~$|X_i| > 2^{2k}$, and since $\G_i$ has $2^{2i}$ vertices and $i < k$, again by the pigeonhole principle there exist $p,q \in X_i$ such that $\C^{(p)}$ and $\C^{(q)}$ have a vertex in common. Since $\C^{(p)}$ has length $pd$, $\C^{(q)}$ has length $qd$, and these two cycles have a vertex in common, for every $a,b \in \N$, there exists a closed walk of length $apd + bqd$ in $\G_i$ (obtained by starting from a vertex $u \in \C^{(p)} \cap \C^{(q)}$, taking $a$ times $\C^{(p)}$ and then $b$ times $\C^{(q)}$).
	Thus, for every $a, b \in \N$, $apd + bqd \in S_i$.
	
	Finally, since $\gcd(pd,qd) = d$ (because $p$ and $q$ are two distinct prime numbers), since $pd, qd \leqslant \sqrt{n_k}$, and since $n_k$ is divisible by~$d$ by Claim~\ref{claim:elementary}, we can apply Lemma~\ref{lem:bezout} which states that there exists $a,b \in \N$ such that $apd + bpd = n_k$. So $n_k \in S_i$, which a contradiction, because by assumption $n_k \in S_k \setminus \bigcup_{1 \leqslant i < k} S_i$. This concludes the proof of Theorem~\ref{thm:cycles}.

	\section{A property with optimal size~$\Theta(\log n)$ in cycles}
	\label{sec:logn-cycle}
	
	Let us now give an example of property that can be certified with $O(\log n)$ bits but not with constant-size certificates, to show that the gap stated in Theorem~\ref{thm:cycles} is optimal.
	
	\begin{proposition}
		Certifying that the length of a cycle is not a power of~$2$ can be done with certificates of size~$O(\log n)$ but not with certificates of size~$O(1)$.
	\end{proposition}
	
	\begin{proof}
		Let $C$ be a cycle of length~$n$, let $u \in V(C)$ and let $P$ be the path obtained from $C$ by deleting one edge adjacent to~$u$. To certify that $n \notin \{2^k, k \in \N\}$, the prover writes in the certificate of every vertex~$v \in V(C)$ the tuple $(d, i)$ where $d \geqslant 3$ is an odd integer that divides~$n$, and~$i$ is the distance from~$u$ to~$v$ in~$P$ modulo~$d$. Every vertex checks that, if its certificate is $(d, i)$ then its two neighbors have certificates $(d, i-1 \mod d)$ and $(d, i+1 \mod d)$, and that $d$ is indeed odd. Such a certificate has size~$O(\log n)$. This scheme is correct: indeed, if all the vertices accept, the length of the cycle should be divisible by~$d$ (and conversely, with the certificates described above, all the vertices will accept if the cycle has length divisible by~$d$).

		Now, assume by contradiction that certifying that the length~$n$ of a cycle is not a power of~$2$ can be done with certificates of constant size~$k$ (or equivalently, that~$2^k$ distinct certificates are sufficient). Let $p$ be an odd prime number such that $p > 2^k$. Let us consider an assignment of certificates to the vertices of a cycle~$C$ of length~$p$ such that all the vertices accept.
		Let us number the vertices of $C$ in clockwise order starting from an arbitrary vertex, and for every $i \in \{0,\ldots, n-1\}$, let $u_i$ be the $i$-th vertex in this numbering, and $c_i$ be its certificate. By the pigeonhole principle, there exists $0 \leqslant i < j \leqslant n-1$ such that $(c_i,c_{i+1}) = (c_j,c_{j+1})$ (where $j+1$ is taken modulo~$n$).
		If $j = i+1$, then a vertex with certificate $c_i$ accepts with two neighbors having certificate $c_i$. Thus, in this case, any cycle is accepted (by giving the certificate $c_i$ to all the vertices) which is a contradiction. Else, let $\ell_1 := j-i$ and $\ell_2 := p - j + i$. We have $\ell_1, \ell_2 \in \{1, \ldots, p-1\}$ and $\ell_1 + \ell_2 = p$. Since $p$ is prime, we get $\gcd(\ell_1, \ell_2) = 1$. Moreover, for any $a_1, a_2 \in \N$, a cycle of size $a_1 \ell_1 + a_2 \ell_2$ is accepted. Indeed, by cutting such a cycle in $a_1$~portions of length~$\ell_1$ and~$a_2$ portions of length~$\ell_2$, giving the certificates~$c_i, \ldots, c_{j-1}$ to the vertices in portions of size~$\ell_1$ and $c_j, \ldots, c_{n-1}, c_0, \ldots, c_{i-1}$ to the vertices in portions of length~$\ell_2$, all the vertices accept because they have the same view as a vertex which accepts in~$C$. Using Lemma~\ref{lem:bezout}, all the the cycles of length $m \geqslant p^2$ are accepted, which is a contradiction because the cycles whose length is a power of~$2$ greater than~$p^2$ should not be accepted.
	\end{proof}
	
	\section{Gap in labeled paths and with larger verification radius}
	\label{sec:gap-labeled-larger-radius}
	
	Let $\Sigma$ be an alphabet. We say that that a path is \emph{$\Sigma$-labeled} (or \emph{labeled} when $\Sigma$ is clear from context) if $P$ is a path given with a function that associates to every vertex a letter in $\Sigma$.
	In this subsection, we prove the following result:
	
	\ThmGapLabeledDistance*
	
	\begin{proof}
		The proof will again rely on state complexity. Let $S$ be the set of labeled paths that satisfy~$\P$. Let $k \in \N$, and let $S_k$ be the set of labeled paths in~$S$ that are accepted with certificates of size~$k$.
		We have: $S = \bigcup_{k \in \N}S_k$.
		The set $S$ can be seen as a set of words in~$\Sigma^\ast$, where the word $w = w_1\ldots w_t \in \Sigma^\ast$ corresponds to the path of length~$t$ labeled by $w_1, \ldots, w_t$. 
		Let us assume by contradiction that $S$ cannot be certified with constant-size certificates.
		Let $S^{\geqslant 2r}$ (resp. $S_k^{\geqslant 2r}$) be the set of words in~$S$ (resp. in~$S_k$) of length at least~$2r$.
		
		
		\begin{claim}
			\label{claim:sk_regular}
			For every $k \in \N$, the set $S_k^{\geqslant 2r}$ is a regular language, that can be recognized by an automaton~$\A_k$ having~$M_k$ states, with $M_k \leqslant 2r(2^k|\Sigma|)^{2r}$.
		\end{claim}
		
		\begin{proof}
			
			Let $k \in \N$, and let us define the following automaton~$\A_k$ that accepts the words in~$S_k^{\geqslant 2r}$, which is constructed in the following way. First, we create two states: $i$ which is the unique initial state and $f$ which is the unique final state. Then, we add a state for every $2r$-tuple of pairs $((c_0, \ell_0), \ldots, (c_{2r-1}, \ell_{2r-1}))$ where $c_0, \ldots, c_{2r-1}$  are certificates of size~$k$ and $\ell_0, \ldots, \ell_{2r-1} \in \Sigma$.
			Informally speaking, in a run, we will read the vertices of the path from one end to the other. Being in such a state in a run means that the vertices with certificates $c_0, \ldots, c_{r-1}$ and labels~$\ell_0, \ldots, \ell_{r-1}$ have already been considered and had accepted, and that we have already read the labels $\ell_0, \ldots, \ell_{r-1}$ in the automaton. 

			Let $c_0, \ldots, c_{2r}$ be certificates of size~$k$ and $\ell_0, \ldots, \ell_{2r} \in \Sigma$. Let $P$ be the path of length~$2r+1$ with vertices having certificates~$c_0, \ldots, c_{2r}$ and labels~$\ell_0, \ldots, \ell_{2r}$ ($P$ is depicted in~Figure~\ref{fig:transition_2r}; such a path $P$ typically corresponds to the view of a vertex). We add the following states and transitions to~$\A_k$:
			\begin{itemize}
				\item If the vertex with certificate~$c_r$ and label~$\ell_r$ accepts in~$P$, we add a transition from $((c_0, \ell_0), \ldots, (c_{2r-1}, \ell_{2r-1}))$ to $((c_1, \ell_1), \ldots, (c_{2r}, \ell_{2r}))$ labeled by $\ell_{r}$.
				
				\item 
				If all the $r$ vertices with certificates~$c_0, \ldots, c_{r-1}$ and labels~$\ell_0, \ldots, \ell_{r-1}$ accept in~$P$, we add $r-1$ intermediary states $p_1, \ldots, p_{r-1}$ (for readability we do not index them by $c_0, \ldots, c_{2r-1}$ and $\ell_0, \ldots ,\ell_{2r-1}$ in our notation), and we add the following transitions:
				$$i \xrightarrow{\ell_0} p_1 \xrightarrow{\ell_1} \ldots \xrightarrow{\ell_{r-2}} p_{r-1} \xrightarrow{\ell_{r-1}} ((c_0, \ell_0), \ldots, (c_{2r-1}, \ldots, \ell_{2r-1}))$$
				
				\item If all the $r$ vertices with certificates~$c_{r+1}, \ldots, c_{2r}$ and labels~$\ell_{r+1}, \ldots, \ell_{2r}$ accept in~$P$, we add $r-1$ intermediary states $q_1, \ldots, q_{r-1}$ (again, for readability we do not index them by $c_1, \ldots, c_{2r}$ and $\ell_1, \ldots ,\ell_{2r}$ in our notation), and we add the following transitions:
				$$((c_1, \ell_1), \ldots, (c_{2r}, \ldots, \ell_{2r})) \xrightarrow{\ell_{r+1}} q_1 \xrightarrow{\ell_{r+2}} \ldots \xrightarrow{\ell_{2r-1}} q_{r-1} \xrightarrow{\ell_{2r}} f $$
			\end{itemize}

			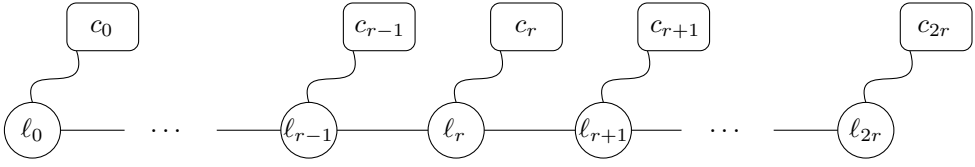
\begin{figure}[h]
				\centering
				
				\begin{tikzpicture}[x=0.75pt,y=0.75pt,yscale=-1,xscale=1]
					
					\draw    (305,135) -- (378.5,135) ;
					\draw    (378.5,135) -- (417.5,135) ;
					\draw    (231.5,135) -- (305,135) ;
					\draw    (185.5,135) -- (231.5,135) ;
					\draw    (93.5,135) -- (139.5,135) ;
					\draw    (463.5,135) -- (509.5,135) ;
					\draw  [fill={rgb, 255:red, 255; green, 255; blue, 255 }  ,fill opacity=1 ] (291,135) .. controls (291,127.27) and (297.27,121) .. (305,121) .. controls (312.73,121) and (319,127.27) .. (319,135) .. controls (319,142.73) and (312.73,149) .. (305,149) .. controls (297.27,149) and (291,142.73) .. (291,135) -- cycle ;
					\draw  [fill={rgb, 255:red, 255; green, 255; blue, 255 }  ,fill opacity=1 ] (364.5,135) .. controls (364.5,127.27) and (370.77,121) .. (378.5,121) .. controls (386.23,121) and (392.5,127.27) .. (392.5,135) .. controls (392.5,142.73) and (386.23,149) .. (378.5,149) .. controls (370.77,149) and (364.5,142.73) .. (364.5,135) -- cycle ;
					\draw  [fill={rgb, 255:red, 255; green, 255; blue, 255 }  ,fill opacity=1 ] (217.5,135) .. controls (217.5,127.27) and (223.77,121) .. (231.5,121) .. controls (239.23,121) and (245.5,127.27) .. (245.5,135) .. controls (245.5,142.73) and (239.23,149) .. (231.5,149) .. controls (223.77,149) and (217.5,142.73) .. (217.5,135) -- cycle ;
					\draw  [fill={rgb, 255:red, 255; green, 255; blue, 255 }  ,fill opacity=1 ] (79.5,135) .. controls (79.5,127.27) and (85.77,121) .. (93.5,121) .. controls (101.23,121) and (107.5,127.27) .. (107.5,135) .. controls (107.5,142.73) and (101.23,149) .. (93.5,149) .. controls (85.77,149) and (79.5,142.73) .. (79.5,135) -- cycle ;
					\draw  [fill={rgb, 255:red, 255; green, 255; blue, 255 }  ,fill opacity=1 ] (495.5,135) .. controls (495.5,127.27) and (501.77,121) .. (509.5,121) .. controls (517.23,121) and (523.5,127.27) .. (523.5,135) .. controls (523.5,142.73) and (517.23,149) .. (509.5,149) .. controls (501.77,149) and (495.5,142.73) .. (495.5,135) -- cycle ;
					\draw    (93.5,121) .. controls (86.5,98) and (123.5,121) .. (115.5,95) ;
					\draw    (231.5,121) .. controls (224.5,98) and (261.5,121) .. (253.5,95) ;
					\draw    (305,121) .. controls (298,98) and (335,121) .. (327,95) ;
					\draw    (378.5,121) .. controls (371.5,98) and (408.5,121) .. (400.5,95) ;
					\draw    (509.5,121) .. controls (502.5,98) and (539.5,121) .. (531.5,95) ;
					\draw  [fill={rgb, 255:red, 255; green, 255; blue, 255 }  ,fill opacity=1 ] (110.7,75.8) .. controls (110.7,73.15) and (112.85,71) .. (115.5,71) -- (141.7,71) .. controls (144.35,71) and (146.5,73.15) .. (146.5,75.8) -- (146.5,90.2) .. controls (146.5,92.85) and (144.35,95) .. (141.7,95) -- (115.5,95) .. controls (112.85,95) and (110.7,92.85) .. (110.7,90.2) -- cycle ;
					\draw  [fill={rgb, 255:red, 255; green, 255; blue, 255 }  ,fill opacity=1 ] (248.7,75.8) .. controls (248.7,73.15) and (250.85,71) .. (253.5,71) -- (279.7,71) .. controls (282.35,71) and (284.5,73.15) .. (284.5,75.8) -- (284.5,90.2) .. controls (284.5,92.85) and (282.35,95) .. (279.7,95) -- (253.5,95) .. controls (250.85,95) and (248.7,92.85) .. (248.7,90.2) -- cycle ;
					\draw  [fill={rgb, 255:red, 255; green, 255; blue, 255 }  ,fill opacity=1 ] (322.2,75.8) .. controls (322.2,73.15) and (324.35,71) .. (327,71) -- (353.2,71) .. controls (355.85,71) and (358,73.15) .. (358,75.8) -- (358,90.2) .. controls (358,92.85) and (355.85,95) .. (353.2,95) -- (327,95) .. controls (324.35,95) and (322.2,92.85) .. (322.2,90.2) -- cycle ;
					\draw  [fill={rgb, 255:red, 255; green, 255; blue, 255 }  ,fill opacity=1 ] (395.7,75.8) .. controls (395.7,73.15) and (397.85,71) .. (400.5,71) -- (426.7,71) .. controls (429.35,71) and (431.5,73.15) .. (431.5,75.8) -- (431.5,90.2) .. controls (431.5,92.85) and (429.35,95) .. (426.7,95) -- (400.5,95) .. controls (397.85,95) and (395.7,92.85) .. (395.7,90.2) -- cycle ;
					\draw  [fill={rgb, 255:red, 255; green, 255; blue, 255 }  ,fill opacity=1 ] (526.7,75.8) .. controls (526.7,73.15) and (528.85,71) .. (531.5,71) -- (557.7,71) .. controls (560.35,71) and (562.5,73.15) .. (562.5,75.8) -- (562.5,90.2) .. controls (562.5,92.85) and (560.35,95) .. (557.7,95) -- (531.5,95) .. controls (528.85,95) and (526.7,92.85) .. (526.7,90.2) -- cycle ;
					
					\draw (86.4,128) node [anchor=north west][inner sep=0.75pt]    {$\ell _{0}$};
					\draw (217.5,128) node [anchor=north west][inner sep=0.75pt]    {$\ell _{r-1}$};
					\draw (297.3,128) node [anchor=north west][inner sep=0.75pt]    {$\ell _{r}$};
					\draw (365,128) node [anchor=north west][inner sep=0.75pt]    {$\ell _{r+1}$};
					\draw (500,128) node [anchor=north west][inner sep=0.75pt]    {$\ell _{2r}$};
					\draw (121,78.4) node [anchor=north west][inner sep=0.75pt]    {$c_{0}$};
					\draw (253,78.4) node [anchor=north west][inner sep=0.75pt]    {$c_{r-1}$};
					\draw (333,78.4) node [anchor=north west][inner sep=0.75pt]    {$c_{r}$};
					\draw (400.5,78.4) node [anchor=north west][inner sep=0.75pt]    {$c_{r+1}$};
					\draw (534,78.4) node [anchor=north west][inner sep=0.75pt]    {$c_{2r}$};
					\draw (151,133) node [anchor=north west][inner sep=0.75pt]    {$\dotsc $};
					\draw (430,133) node [anchor=north west][inner sep=0.75pt]    {$\dotsc $};

				\end{tikzpicture}
				
				\caption{The path~$P$ used to define the transitions in~$\A_k$.}
				\label{fig:transition_2r}
			\end{figure}

			Just by counting directly from the definition, the automaton $\A_k$ has at most $2r(2^k|\Sigma|)^{2r}$ states.
			Moreover, it accepts the language $S_k^{\geqslant 2r}$. Indeed, let us consider a labeled path having~$t \geqslant 2r$ vertices, which is in $S_k^{\geqslant 2r}$. Let $\ell_0, \ldots, \ell_{t-1}$ be the labels of the vertices, and let $c_0, \ldots, c_{t-1}$ be a certificate assignment that makes every vertex accept. Then, by definition of~$\A_k$:
			\begin{itemize}
				\item there are transitions from~$i$ to~$((c_0,\ell_0) \ldots ,(c_{2r-1},\ell_{2r-1}))$ labeled by $\ell_0, \ldots, \ell_{r-1}$
				\item for every $j \in \{0, \ldots, t-2r-1\}$, there is a transition from $((c_j, \ell_j), \ldots, (c_{j+2r-1}, \ell_{j+2r-1}))$ to $((c_{j+1}, \ell_{j+1}), \ldots, (c_{j+2r}, \ell_{j+2r}))$ labeled by $\ell_{j+r}$
				\item there are transitions from~$((c_{t-2r},\ell_{t-2r}), \ldots, (c_{t-1}, \ell_{t-1}))$ to~$f$ labeled by~$\ell_{t-r}, \ldots, \ell_{t-1}$
			\end{itemize}
			
			Thus, there is an accepting run in~$\A_k$ that accepts $\ell_0, \ldots, \ell_{t-1}$. And conversely, any accepting run in $\A_k$ can be converted in a path in $S_k^{\geqslant 2r}$ (an assignment of the certificates that makes every vertex accept is given by the states of the accepting run). 
		\end{proof}

		Now, let $X$ be the set of integers $k \in \N$ such that~$S_{k}^{\geqslant 2r} \setminus \bigcup_{1 \leqslant i<k}S_i^{\geqslant2r}$ is non-empty.
		The set~$X$ in infinite: indeed, assume by contradiction that~$X$ is finite. Then, there exists $k \in \N$ such that $S^{\geqslant 2r} = \bigcup_{i \leqslant k}S_i^{\geqslant 2r}$. In other words, it means that constant-size certificates are sufficient to accept all the labeled paths in~$S$ of length at least~$2r$. Since there is only a finite number of labeled paths in~$S$ of length at most~$2r$, it means that constant-size certificates are sufficient to certify~$S$, which is not the case by assumption. So~$X$ is infinite.
		
		For every $k \in X$, let $n_k$ be the length of a smallest word in $S_{k}^{\geqslant 2r} \setminus \bigcup_{1 \leqslant i<k}S_i^{\geqslant2r}$. Note that we have $\lim_{k \to +\infty, k \in X} n_k = +\infty$. Indeed, for every $n \in \N$, there is only a finite number of $k \in X$ such that $n_k = n$, because there are only a finite number of words of length~$n$.
		Let $k_0 \in \N$ be such that, for every $k \in X$ and $k \geqslant k_0$, we have $n_k \geqslant \max(N,2r)$. For such an integer~$k$, since $n_k \in S_{s(n_k)}^{\geqslant 2r}$, we have $s(n_k) \geqslant k$.
		
		Let $k \in X$ such that $k \geqslant k_0$. By Claim~\ref{claim:sk_regular}, the set $S_k^{\geqslant2r}$ is a regular language in $\Sigma^\ast$, recognized by an automaton~$\A_k$ that has $M_k \leqslant 2r(2^k|\Sigma|)^{2r}$ states.
		By Lemma~\ref{lem:automata}, there exists an automaton~$\A_k'$ that accepts $S_{k}^{\geqslant 2r} \setminus \bigcup_{1 \leqslant i<k}S_i^{\geqslant2r}$ and which has at most $M_k \cdot 2^{\left(\sum_{i=1}^{k-1}M_i\right)}$ states.
		Let us prove the following claim:
		
		\begin{claim}
			\label{claim:bound_k}
			We have $k \geqslant \frac{\log \log n_k}{2r} + O_{r, |\Sigma|}(1)$
		\end{claim}
		
		\begin{proof}
			Since $n_k$ is the smallest word accepted by~$\A_k'$, it is smaller than the number of states of $\A_k'$,  which is at most~$M_k \cdot 2^{\left(\sum_{i=1}^{k-1}M_i\right)}$. Moreover, we also know that $M_i \leqslant 2r(2^i|\Sigma|)^{2r}$ for every $i \in \N$, so we have $\sum_{i=1}^{k-1}M_i \leqslant 2r|\Sigma|^{2r}2^{2kr}$. Thus:
			\begin{align*}
				n_k 
				&\leqslant M_k \cdot 2^{\left(\sum_{i=1}^{k-1}M_i\right)} \\
				&\leqslant 2r|\Sigma|^{2r}2^{2kr+2r|\Sigma|^{2r}2^{2kr}}\\
				&\leqslant 2r|\Sigma|^{2r}2^{(2r+1)|\Sigma|^{2r}2^{2kr}}\\
			\end{align*}
			By taking the logarithm of the previous equation, we obtain:
			$$\log n_k \leqslant O_{r, |\Sigma|}(1) + O_{r,|\Sigma|}(1)2^{2kr}$$
			Finally, again by taking the logarithm, we obtain:
			$$2kr \geqslant \log\left(\frac{\log n_k + O_{r,|\Sigma|}(1)}{O_{r,|\Sigma|}(1)} \right)$$
			which finally gives $k \geqslant \frac{\log \log n_k}{2r} + O_{r,|\Sigma|}(1)$.
		\end{proof}
		
		We can finally combine Claim~\ref{claim:bound_k} with the inequality $s(n_k) \geqslant k$ to obtain that $s(n_k) \geqslant \frac{\log \log n_k}{2r} + O_{r,|\Sigma|}(1)$ for every $k \in X$ such that $k \geqslant k_0$. This is a contradiction with the fact that $c>2$ and $s(n_k) = \left\lfloor\frac{\log \log n_k}{cr}\right\rfloor$ since $n_k \geqslant N$.
		This concludes the proof and shows that $\P$ can be certified with constant-size certificates.
	\end{proof}

	\section{Gap in $d$ for trees with radius 1}
	\label{sec:gap-tree}
	
	\ThmGapTreeRadiusOne*

	\subsection{Warm-up 1: parsing trees}

	Let $T$ be a tree, and let $u,v \in V(T)$. Let $P_{uv}$ be the path from~$u$ to~$v$ in~$T$. Let $m$ be the number of vertices of $P_{uv}$, and let $w_1, \ldots, w_m$ be the vertices of $P_{uv}$ (with $w_1 = u$ and $w_m = v$). For each $i \in \{1, \ldots, m\}$, the tree $T_{i}$ \emph{rooted in~$w_i$} is the connected component of~$w_i$ in $T$ after removing all the edges of $P_{uv}$ (with a prescribed vertex $w_i$). Finally, let $\sigma(T,u,v)$ be the sequence of rooted trees $T_1, \ldots, T_m$. We say that $\sigma(T,u,v)$ is a \emph{parsing} of the tree~$T$. See Figure~\ref{fig:parsing} for an example.
	
	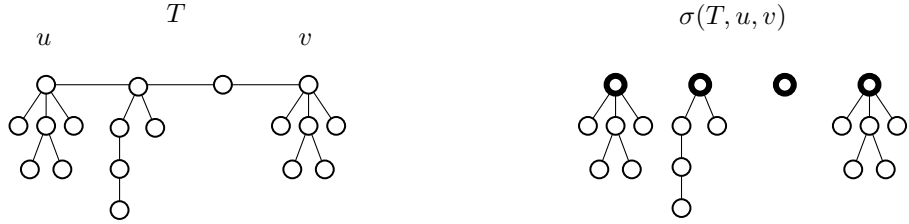
\begin{figure}[h!]
		\centering
		
		\begin{tikzpicture}[x=0.75pt,y=0.75pt,yscale=-1,xscale=1]
			
			\draw    (43.25,102.46) -- (89.25,102.46) ;
			\draw    (89.25,102.46) -- (131.5,102.46) ;
			\draw    (131.5,102.46) -- (173.75,102.46) ;
			\draw    (43.25,123.14) -- (35.14,145.04) ;
			\draw    (43.25,123.14) -- (51.36,145.04) ;
			\draw    (43.25,102.46) -- (43.25,123.14) ;
			\draw    (43.25,102.46) -- (57.04,123.14) ;
			\draw    (29.46,123.14) -- (43.25,102.46) ;
			\draw  [fill={rgb, 255:red, 255; green, 255; blue, 255 }  ,fill opacity=1 ][line width=0.75]  (38.79,102.46) .. controls (38.79,100) and (40.79,98) .. (43.25,98) .. controls (45.71,98) and (47.71,100) .. (47.71,102.46) .. controls (47.71,104.92) and (45.71,106.92) .. (43.25,106.92) .. controls (40.79,106.92) and (38.79,104.92) .. (38.79,102.46) -- cycle ;
			\draw  [fill={rgb, 255:red, 255; green, 255; blue, 255 }  ,fill opacity=1 ][line width=0.75]  (52.58,123.14) .. controls (52.58,120.68) and (54.58,118.68) .. (57.04,118.68) .. controls (59.5,118.68) and (61.5,120.68) .. (61.5,123.14) .. controls (61.5,125.61) and (59.5,127.61) .. (57.04,127.61) .. controls (54.58,127.61) and (52.58,125.61) .. (52.58,123.14) -- cycle ;
			\draw  [fill={rgb, 255:red, 255; green, 255; blue, 255 }  ,fill opacity=1 ][line width=0.75]  (25,123.14) .. controls (25,120.68) and (27,118.68) .. (29.46,118.68) .. controls (31.92,118.68) and (33.92,120.68) .. (33.92,123.14) .. controls (33.92,125.61) and (31.92,127.61) .. (29.46,127.61) .. controls (27,127.61) and (25,125.61) .. (25,123.14) -- cycle ;
			\draw  [fill={rgb, 255:red, 255; green, 255; blue, 255 }  ,fill opacity=1 ][line width=0.75]  (38.79,123.14) .. controls (38.79,120.68) and (40.79,118.68) .. (43.25,118.68) .. controls (45.71,118.68) and (47.71,120.68) .. (47.71,123.14) .. controls (47.71,125.61) and (45.71,127.61) .. (43.25,127.61) .. controls (40.79,127.61) and (38.79,125.61) .. (38.79,123.14) -- cycle ;
			\draw  [fill={rgb, 255:red, 255; green, 255; blue, 255 }  ,fill opacity=1 ][line width=0.75]  (46.9,145.04) .. controls (46.9,142.58) and (48.9,140.58) .. (51.36,140.58) .. controls (53.82,140.58) and (55.82,142.58) .. (55.82,145.04) .. controls (55.82,147.51) and (53.82,149.51) .. (51.36,149.51) .. controls (48.9,149.51) and (46.9,147.51) .. (46.9,145.04) -- cycle ;
			\draw  [fill={rgb, 255:red, 255; green, 255; blue, 255 }  ,fill opacity=1 ][line width=0.75]  (30.68,145.04) .. controls (30.68,142.58) and (32.68,140.58) .. (35.14,140.58) .. controls (37.6,140.58) and (39.6,142.58) .. (39.6,145.04) .. controls (39.6,147.51) and (37.6,149.51) .. (35.14,149.51) .. controls (32.68,149.51) and (30.68,147.51) .. (30.68,145.04) -- cycle ;
			\draw    (79.92,144.83) -- (79.92,165.51) ;
			\draw    (79.92,124.14) -- (79.92,144.83) ;
			\draw    (89.25,103.46) -- (97.71,124.14) ;
			\draw    (79.92,124.14) -- (89.25,103.46) ;
			\draw  [fill={rgb, 255:red, 255; green, 255; blue, 255 }  ,fill opacity=1 ][line width=0.75]  (84.79,103.46) .. controls (84.79,101) and (86.79,99) .. (89.25,99) .. controls (91.71,99) and (93.71,101) .. (93.71,103.46) .. controls (93.71,105.92) and (91.71,107.92) .. (89.25,107.92) .. controls (86.79,107.92) and (84.79,105.92) .. (84.79,103.46) -- cycle ;
			\draw  [fill={rgb, 255:red, 255; green, 255; blue, 255 }  ,fill opacity=1 ][line width=0.75]  (93.25,124.14) .. controls (93.25,121.68) and (95.25,119.68) .. (97.71,119.68) .. controls (100.17,119.68) and (102.17,121.68) .. (102.17,124.14) .. controls (102.17,126.61) and (100.17,128.61) .. (97.71,128.61) .. controls (95.25,128.61) and (93.25,126.61) .. (93.25,124.14) -- cycle ;
			\draw  [fill={rgb, 255:red, 255; green, 255; blue, 255 }  ,fill opacity=1 ][line width=0.75]  (75.46,124.14) .. controls (75.46,121.68) and (77.46,119.68) .. (79.92,119.68) .. controls (82.39,119.68) and (84.38,121.68) .. (84.38,124.14) .. controls (84.38,126.61) and (82.39,128.61) .. (79.92,128.61) .. controls (77.46,128.61) and (75.46,126.61) .. (75.46,124.14) -- cycle ;
			\draw  [fill={rgb, 255:red, 255; green, 255; blue, 255 }  ,fill opacity=1 ][line width=0.75]  (75.46,144.83) .. controls (75.46,142.36) and (77.46,140.37) .. (79.92,140.37) .. controls (82.39,140.37) and (84.38,142.36) .. (84.38,144.83) .. controls (84.38,147.29) and (82.39,149.29) .. (79.92,149.29) .. controls (77.46,149.29) and (75.46,147.29) .. (75.46,144.83) -- cycle ;
			\draw  [fill={rgb, 255:red, 255; green, 255; blue, 255 }  ,fill opacity=1 ][line width=0.75]  (75.46,165.51) .. controls (75.46,163.05) and (77.46,161.05) .. (79.92,161.05) .. controls (82.39,161.05) and (84.38,163.05) .. (84.38,165.51) .. controls (84.38,167.97) and (82.39,169.97) .. (79.92,169.97) .. controls (77.46,169.97) and (75.46,167.97) .. (75.46,165.51) -- cycle ;
			\draw  [fill={rgb, 255:red, 255; green, 255; blue, 255 }  ,fill opacity=1 ][line width=0.75]  (127.04,102.46) .. controls (127.04,100) and (129.04,98) .. (131.5,98) .. controls (133.96,98) and (135.96,100) .. (135.96,102.46) .. controls (135.96,104.92) and (133.96,106.92) .. (131.5,106.92) .. controls (129.04,106.92) and (127.04,104.92) .. (127.04,102.46) -- cycle ;
			\draw    (174.25,123.14) -- (166.14,145.04) ;
			\draw    (174.25,123.14) -- (182.36,145.04) ;
			\draw    (174.25,102.46) -- (174.25,123.14) ;
			\draw    (174.25,102.46) -- (188.04,123.14) ;
			\draw    (160.46,123.14) -- (174.25,102.46) ;
			\draw  [fill={rgb, 255:red, 255; green, 255; blue, 255 }  ,fill opacity=1 ][line width=0.75]  (169.79,102.46) .. controls (169.79,100) and (171.79,98) .. (174.25,98) .. controls (176.71,98) and (178.71,100) .. (178.71,102.46) .. controls (178.71,104.92) and (176.71,106.92) .. (174.25,106.92) .. controls (171.79,106.92) and (169.79,104.92) .. (169.79,102.46) -- cycle ;
			\draw  [fill={rgb, 255:red, 255; green, 255; blue, 255 }  ,fill opacity=1 ][line width=0.75]  (183.58,123.14) .. controls (183.58,120.68) and (185.58,118.68) .. (188.04,118.68) .. controls (190.5,118.68) and (192.5,120.68) .. (192.5,123.14) .. controls (192.5,125.61) and (190.5,127.61) .. (188.04,127.61) .. controls (185.58,127.61) and (183.58,125.61) .. (183.58,123.14) -- cycle ;
			\draw  [fill={rgb, 255:red, 255; green, 255; blue, 255 }  ,fill opacity=1 ][line width=0.75]  (156,123.14) .. controls (156,120.68) and (158,118.68) .. (160.46,118.68) .. controls (162.92,118.68) and (164.92,120.68) .. (164.92,123.14) .. controls (164.92,125.61) and (162.92,127.61) .. (160.46,127.61) .. controls (158,127.61) and (156,125.61) .. (156,123.14) -- cycle ;
			\draw  [fill={rgb, 255:red, 255; green, 255; blue, 255 }  ,fill opacity=1 ][line width=0.75]  (169.79,123.14) .. controls (169.79,120.68) and (171.79,118.68) .. (174.25,118.68) .. controls (176.71,118.68) and (178.71,120.68) .. (178.71,123.14) .. controls (178.71,125.61) and (176.71,127.61) .. (174.25,127.61) .. controls (171.79,127.61) and (169.79,125.61) .. (169.79,123.14) -- cycle ;
			\draw  [fill={rgb, 255:red, 255; green, 255; blue, 255 }  ,fill opacity=1 ][line width=0.75]  (177.9,145.04) .. controls (177.9,142.58) and (179.9,140.58) .. (182.36,140.58) .. controls (184.82,140.58) and (186.82,142.58) .. (186.82,145.04) .. controls (186.82,147.51) and (184.82,149.51) .. (182.36,149.51) .. controls (179.9,149.51) and (177.9,147.51) .. (177.9,145.04) -- cycle ;
			\draw  [fill={rgb, 255:red, 255; green, 255; blue, 255 }  ,fill opacity=1 ][line width=0.75]  (161.68,145.04) .. controls (161.68,142.58) and (163.68,140.58) .. (166.14,140.58) .. controls (168.6,140.58) and (170.6,142.58) .. (170.6,145.04) .. controls (170.6,147.51) and (168.6,149.51) .. (166.14,149.51) .. controls (163.68,149.51) and (161.68,147.51) .. (161.68,145.04) -- cycle ;
			\draw    (327.5,123.14) -- (319.39,145.04) ;
			\draw    (327.5,123.14) -- (335.61,145.04) ;
			\draw    (327.5,102.46) -- (327.5,123.14) ;
			\draw    (327.5,102.46) -- (341.29,123.14) ;
			\draw    (313.71,123.14) -- (327.5,102.46) ;
			\draw  [fill={rgb, 255:red, 255; green, 255; blue, 255 }  ,fill opacity=1 ][line width=2.25]  (323.04,102.46) .. controls (323.04,100) and (325.04,98) .. (327.5,98) .. controls (329.96,98) and (331.96,100) .. (331.96,102.46) .. controls (331.96,104.92) and (329.96,106.92) .. (327.5,106.92) .. controls (325.04,106.92) and (323.04,104.92) .. (323.04,102.46) -- cycle ;
			\draw  [fill={rgb, 255:red, 255; green, 255; blue, 255 }  ,fill opacity=1 ][line width=0.75]  (336.83,123.14) .. controls (336.83,120.68) and (338.83,118.68) .. (341.29,118.68) .. controls (343.75,118.68) and (345.75,120.68) .. (345.75,123.14) .. controls (345.75,125.61) and (343.75,127.61) .. (341.29,127.61) .. controls (338.83,127.61) and (336.83,125.61) .. (336.83,123.14) -- cycle ;
			\draw  [fill={rgb, 255:red, 255; green, 255; blue, 255 }  ,fill opacity=1 ][line width=0.75]  (309.25,123.14) .. controls (309.25,120.68) and (311.25,118.68) .. (313.71,118.68) .. controls (316.17,118.68) and (318.17,120.68) .. (318.17,123.14) .. controls (318.17,125.61) and (316.17,127.61) .. (313.71,127.61) .. controls (311.25,127.61) and (309.25,125.61) .. (309.25,123.14) -- cycle ;
			\draw  [fill={rgb, 255:red, 255; green, 255; blue, 255 }  ,fill opacity=1 ][line width=0.75]  (323.04,123.14) .. controls (323.04,120.68) and (325.04,118.68) .. (327.5,118.68) .. controls (329.96,118.68) and (331.96,120.68) .. (331.96,123.14) .. controls (331.96,125.61) and (329.96,127.61) .. (327.5,127.61) .. controls (325.04,127.61) and (323.04,125.61) .. (323.04,123.14) -- cycle ;
			\draw  [fill={rgb, 255:red, 255; green, 255; blue, 255 }  ,fill opacity=1 ][line width=0.75]  (331.15,145.04) .. controls (331.15,142.58) and (333.15,140.58) .. (335.61,140.58) .. controls (338.07,140.58) and (340.07,142.58) .. (340.07,145.04) .. controls (340.07,147.51) and (338.07,149.51) .. (335.61,149.51) .. controls (333.15,149.51) and (331.15,147.51) .. (331.15,145.04) -- cycle ;
			\draw  [fill={rgb, 255:red, 255; green, 255; blue, 255 }  ,fill opacity=1 ][line width=0.75]  (314.93,145.04) .. controls (314.93,142.58) and (316.93,140.58) .. (319.39,140.58) .. controls (321.85,140.58) and (323.85,142.58) .. (323.85,145.04) .. controls (323.85,147.51) and (321.85,149.51) .. (319.39,149.51) .. controls (316.93,149.51) and (314.93,147.51) .. (314.93,145.04) -- cycle ;
			\draw    (360.42,143.83) -- (360.42,164.51) ;
			\draw    (360.42,123.14) -- (360.42,143.83) ;
			\draw    (369.75,102.46) -- (378.21,123.14) ;
			\draw    (360.42,123.14) -- (369.75,102.46) ;
			\draw  [fill={rgb, 255:red, 255; green, 255; blue, 255 }  ,fill opacity=1 ][line width=2.25]  (365.29,102.46) .. controls (365.29,100) and (367.29,98) .. (369.75,98) .. controls (372.21,98) and (374.21,100) .. (374.21,102.46) .. controls (374.21,104.92) and (372.21,106.92) .. (369.75,106.92) .. controls (367.29,106.92) and (365.29,104.92) .. (365.29,102.46) -- cycle ;
			\draw  [fill={rgb, 255:red, 255; green, 255; blue, 255 }  ,fill opacity=1 ][line width=0.75]  (355.96,123.14) .. controls (355.96,120.68) and (357.96,118.68) .. (360.42,118.68) .. controls (362.89,118.68) and (364.88,120.68) .. (364.88,123.14) .. controls (364.88,125.61) and (362.89,127.61) .. (360.42,127.61) .. controls (357.96,127.61) and (355.96,125.61) .. (355.96,123.14) -- cycle ;
			\draw  [fill={rgb, 255:red, 255; green, 255; blue, 255 }  ,fill opacity=1 ][line width=0.75]  (373.75,123.14) .. controls (373.75,120.68) and (375.75,118.68) .. (378.21,118.68) .. controls (380.67,118.68) and (382.67,120.68) .. (382.67,123.14) .. controls (382.67,125.61) and (380.67,127.61) .. (378.21,127.61) .. controls (375.75,127.61) and (373.75,125.61) .. (373.75,123.14) -- cycle ;
			\draw  [fill={rgb, 255:red, 255; green, 255; blue, 255 }  ,fill opacity=1 ][line width=0.75]  (355.96,143.83) .. controls (355.96,141.36) and (357.96,139.37) .. (360.42,139.37) .. controls (362.89,139.37) and (364.88,141.36) .. (364.88,143.83) .. controls (364.88,146.29) and (362.89,148.29) .. (360.42,148.29) .. controls (357.96,148.29) and (355.96,146.29) .. (355.96,143.83) -- cycle ;
			\draw  [fill={rgb, 255:red, 255; green, 255; blue, 255 }  ,fill opacity=1 ][line width=0.75]  (355.96,164.51) .. controls (355.96,162.05) and (357.96,160.05) .. (360.42,160.05) .. controls (362.89,160.05) and (364.88,162.05) .. (364.88,164.51) .. controls (364.88,166.97) and (362.89,168.97) .. (360.42,168.97) .. controls (357.96,168.97) and (355.96,166.97) .. (355.96,164.51) -- cycle ;
			\draw  [fill={rgb, 255:red, 255; green, 255; blue, 255 }  ,fill opacity=1 ][line width=2.25]  (407.54,102.46) .. controls (407.54,100) and (409.54,98) .. (412,98) .. controls (414.46,98) and (416.46,100) .. (416.46,102.46) .. controls (416.46,104.92) and (414.46,106.92) .. (412,106.92) .. controls (409.54,106.92) and (407.54,104.92) .. (407.54,102.46) -- cycle ;
			\draw    (454.25,123.14) -- (446.14,145.04) ;
			\draw    (454.25,123.14) -- (462.36,145.04) ;
			\draw    (454.25,102.46) -- (454.25,123.14) ;
			\draw    (454.25,102.46) -- (468.04,123.14) ;
			\draw    (440.46,123.14) -- (454.25,102.46) ;
			\draw  [fill={rgb, 255:red, 255; green, 255; blue, 255 }  ,fill opacity=1 ][line width=2.25]  (449.79,102.46) .. controls (449.79,100) and (451.79,98) .. (454.25,98) .. controls (456.71,98) and (458.71,100) .. (458.71,102.46) .. controls (458.71,104.92) and (456.71,106.92) .. (454.25,106.92) .. controls (451.79,106.92) and (449.79,104.92) .. (449.79,102.46) -- cycle ;
			\draw  [fill={rgb, 255:red, 255; green, 255; blue, 255 }  ,fill opacity=1 ][line width=0.75]  (463.58,123.14) .. controls (463.58,120.68) and (465.58,118.68) .. (468.04,118.68) .. controls (470.5,118.68) and (472.5,120.68) .. (472.5,123.14) .. controls (472.5,125.61) and (470.5,127.61) .. (468.04,127.61) .. controls (465.58,127.61) and (463.58,125.61) .. (463.58,123.14) -- cycle ;
			\draw  [fill={rgb, 255:red, 255; green, 255; blue, 255 }  ,fill opacity=1 ][line width=0.75]  (436,123.14) .. controls (436,120.68) and (438,118.68) .. (440.46,118.68) .. controls (442.92,118.68) and (444.92,120.68) .. (444.92,123.14) .. controls (444.92,125.61) and (442.92,127.61) .. (440.46,127.61) .. controls (438,127.61) and (436,125.61) .. (436,123.14) -- cycle ;
			\draw  [fill={rgb, 255:red, 255; green, 255; blue, 255 }  ,fill opacity=1 ][line width=0.75]  (449.79,123.14) .. controls (449.79,120.68) and (451.79,118.68) .. (454.25,118.68) .. controls (456.71,118.68) and (458.71,120.68) .. (458.71,123.14) .. controls (458.71,125.61) and (456.71,127.61) .. (454.25,127.61) .. controls (451.79,127.61) and (449.79,125.61) .. (449.79,123.14) -- cycle ;
			\draw  [fill={rgb, 255:red, 255; green, 255; blue, 255 }  ,fill opacity=1 ][line width=0.75]  (457.9,145.04) .. controls (457.9,142.58) and (459.9,140.58) .. (462.36,140.58) .. controls (464.82,140.58) and (466.82,142.58) .. (466.82,145.04) .. controls (466.82,147.51) and (464.82,149.51) .. (462.36,149.51) .. controls (459.9,149.51) and (457.9,147.51) .. (457.9,145.04) -- cycle ;
			\draw  [fill={rgb, 255:red, 255; green, 255; blue, 255 }  ,fill opacity=1 ][line width=0.75]  (441.68,145.04) .. controls (441.68,142.58) and (443.68,140.58) .. (446.14,140.58) .. controls (448.6,140.58) and (450.6,142.58) .. (450.6,145.04) .. controls (450.6,147.51) and (448.6,149.51) .. (446.14,149.51) .. controls (443.68,149.51) and (441.68,147.51) .. (441.68,145.04) -- cycle ;
			
			\draw (37,75.4) node [anchor=north west][inner sep=0.75pt]    {$u$};
			\draw (168,75.4) node [anchor=north west][inner sep=0.75pt]    {$v$};
			\draw (102,60.4) node [anchor=north west][inner sep=0.75pt]    {${\displaystyle T}$};
			\draw (358,60.4) node [anchor=north west][inner sep=0.75pt]    {$\sigma ( T,u,v)$};

		\end{tikzpicture}

		\caption{A tree $T$ and a parsing of~$T$. The bold vertices are the roots of the trees.}
		\label{fig:parsing}
	\end{figure}

	Let $T_1, \ldots, T_m$ be a sequence of rooted trees. The \emph{gluing} $\tau(T_1, \ldots, T_m)$  of $T_1, \ldots, T_m$ is the (non-rooted) tree obtained by adding an edge between the root of $T_i$ and the root of $T_{i+1}$, for every $i \in \{1, \ldots, m-1\}$.

	The two following lemmas just follow from the definition of $\sigma$ and $\tau$:
	
	\begin{lemma}
		\label{lem:parsing}
		\begin{itemize}
			\item For every tree $T$ and every vertices $u,v \in V(T)$, we have $\tau(\sigma(T,u,v)) = T$.
			\item For every sequence of rooted trees~$T_1, \ldots, T_m$, let~$u$ (resp.~$v$) be the vertex corresponding to the root of~$T_1$ (resp. of~$T_m$) in $\tau(T_1, \ldots, T_m$). Then: $\sigma(\tau(T_1, \ldots, T_m),u,v) = T_1, \ldots, T_m$.
		\end{itemize}
	\end{lemma}

	\begin{lemma}
		\label{lem:gluing_subsequence}
		Let $T_1, \ldots, T_m$ be a sequence of rooted trees, and let $T'_1, \ldots, T'_{m'}$ be a strict subsequence of $T_1, \ldots, T_m$ (that is, a sequence obtained from $T_1, \ldots, T_m$ by deleting some trees, and at least one). Then, the diameter of $\tau(T'_1, \ldots, T'_{m'})$ is less or equal than the diameter of $\tau(T_1, \ldots, T_m)$, and $\tau(T'_1, \ldots, T'_{m'})$ has strictly less vertices than $\tau(T_1, \ldots, T_m)$.
	\end{lemma}

	\subsection{Warm-up 2: constant size-certification and automata point of view}
	
	The proof of Theorem~\ref{thm:gap-tree-radius-1} will rely on automata. The proof generalizes of the proof of Theorem~\ref{thm:gap-paths-state-complexity} from paths to trees. For every property $\P$ on trees, we define the set $S$ of all trees that satisfy~$\P$. 
	For every $k \in \N$, let us denote by $C_k$ the set of certificates of size~$k$, and by $S_k$ the subset of~$S$ containing the trees that are accepted with certificates in~$C_k$.
	We have: $S = \bigcup_{k \in \N} S_k$.
	In this section, our goal is to construct an automaton having its transitions labeled by rooted trees, that accepts all the parsings of the trees in~$S_k$, to be able to use the same kind of arguments as in the proof of Theorem~\ref{thm:gap-paths-state-complexity}.
	
	In what follows, we will thus define a non-deterministic finite automaton $\A_k$, having its transitions labeled by rooted trees. In other words, the automaton has a finite number of states but there can be an infinite number of transitions (since the set of rooted trees is infinite). A sequence $T_1, \ldots, T_m$ of rooted trees is said to be \emph{accepted by $\A_k$} if there exists a sequence of transitions $p_0 \xrightarrow{T_1} p_1 \xrightarrow{T_2} \ldots \xrightarrow{T_m} p_m$ in~$\A_k$, where $p_0$ is an initial state, $p_m$ is a final state, and for every $i \in \{1, \ldots, m\}$ there is a transition from~$p_{i-1}$ to~$p_i$ labeled by~$T_i$. See Figure~\ref{fig:automata_trees} for an example.
	
	
	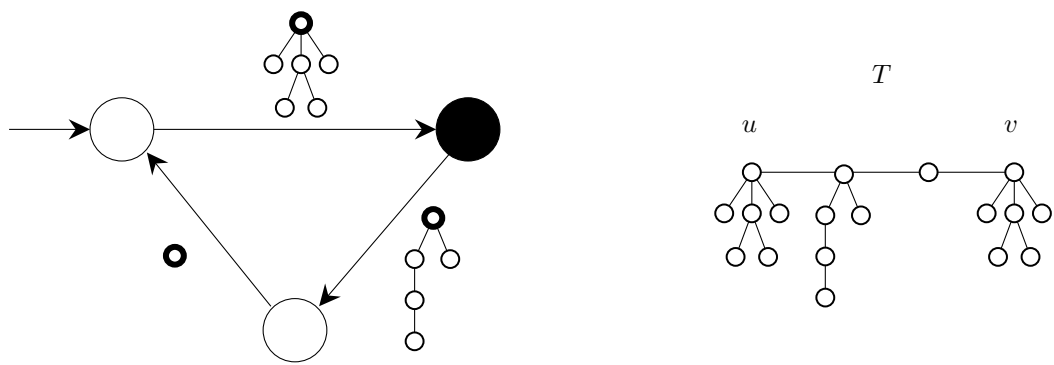
\begin{figure}[h!]
		\centering
		
		\begin{tikzpicture}[x=0.75pt,y=0.75pt,yscale=-1,xscale=1]
			
			\draw    (226.92,199.83) -- (226.92,220.51) ;
			\draw    (226.92,179.14) -- (226.92,199.83) ;
			\draw    (170.25,81.14) -- (162.14,103.04) ;
			\draw    (170.25,81.14) -- (178.36,103.04) ;
			\draw    (170.25,60.46) -- (170.25,81.14) ;
			\draw    (170.25,60.46) -- (184.04,81.14) ;
			\draw    (156.46,81.14) -- (170.25,60.46) ;
			\draw  [fill={rgb, 255:red, 255; green, 255; blue, 255 }  ,fill opacity=1 ] (64.9,113.9) .. controls (64.9,105.12) and (72.02,98) .. (80.81,98) .. controls (89.59,98) and (96.71,105.12) .. (96.71,113.9) .. controls (96.71,122.69) and (89.59,129.81) .. (80.81,129.81) .. controls (72.02,129.81) and (64.9,122.69) .. (64.9,113.9) -- cycle ;
			\draw  [fill={rgb, 255:red, 0; green, 0; blue, 0 }  ,fill opacity=1 ] (237.69,113.9) .. controls (237.69,105.12) and (244.81,98) .. (253.6,98) .. controls (262.38,98) and (269.5,105.12) .. (269.5,113.9) .. controls (269.5,122.69) and (262.38,129.81) .. (253.6,129.81) .. controls (244.81,129.81) and (237.69,122.69) .. (237.69,113.9) -- cycle ;
			\draw  [fill={rgb, 255:red, 255; green, 255; blue, 255 }  ,fill opacity=1 ] (151.3,214.91) .. controls (151.3,206.13) and (158.42,199.01) .. (167.2,199.01) .. controls (175.99,199.01) and (183.11,206.13) .. (183.11,214.91) .. controls (183.11,223.7) and (175.99,230.82) .. (167.2,230.82) .. controls (158.42,230.82) and (151.3,223.7) .. (151.3,214.91) -- cycle ;
			\draw    (96.71,113.9) -- (234.69,113.9) ;
			\draw [shift={(237.69,113.9)}, rotate = 180] [fill={rgb, 255:red, 0; green, 0; blue, 0 }  ][line width=0.08]  [draw opacity=0] (10.72,-5.15) -- (0,0) -- (10.72,5.15) -- (7.12,0) -- cycle    ;
			\draw    (244.57,125.94) -- (181.18,200.59) ;
			\draw [shift={(179.24,202.88)}, rotate = 310.34] [fill={rgb, 255:red, 0; green, 0; blue, 0 }  ][line width=0.08]  [draw opacity=0] (10.72,-5.15) -- (0,0) -- (10.72,5.15) -- (7.12,0) -- cycle    ;
			\draw    (155.17,202.88) -- (95.15,127.85) ;
			\draw [shift={(93.27,125.51)}, rotate = 51.34] [fill={rgb, 255:red, 0; green, 0; blue, 0 }  ][line width=0.08]  [draw opacity=0] (10.72,-5.15) -- (0,0) -- (10.72,5.15) -- (7.12,0) -- cycle    ;
			\draw    (24.5,113.9) -- (61.9,113.9) ;
			\draw [shift={(64.9,113.9)}, rotate = 180] [fill={rgb, 255:red, 0; green, 0; blue, 0 }  ][line width=0.08]  [draw opacity=0] (10.72,-5.15) -- (0,0) -- (10.72,5.15) -- (7.12,0) -- cycle    ;
			\draw  [fill={rgb, 255:red, 255; green, 255; blue, 255 }  ,fill opacity=1 ][line width=2.25]  (165.79,60.46) .. controls (165.79,58) and (167.79,56) .. (170.25,56) .. controls (172.71,56) and (174.71,58) .. (174.71,60.46) .. controls (174.71,62.92) and (172.71,64.92) .. (170.25,64.92) .. controls (167.79,64.92) and (165.79,62.92) .. (165.79,60.46) -- cycle ;
			\draw  [fill={rgb, 255:red, 255; green, 255; blue, 255 }  ,fill opacity=1 ][line width=0.75]  (179.58,81.14) .. controls (179.58,78.68) and (181.58,76.68) .. (184.04,76.68) .. controls (186.5,76.68) and (188.5,78.68) .. (188.5,81.14) .. controls (188.5,83.61) and (186.5,85.61) .. (184.04,85.61) .. controls (181.58,85.61) and (179.58,83.61) .. (179.58,81.14) -- cycle ;
			\draw  [fill={rgb, 255:red, 255; green, 255; blue, 255 }  ,fill opacity=1 ][line width=0.75]  (152,81.14) .. controls (152,78.68) and (154,76.68) .. (156.46,76.68) .. controls (158.92,76.68) and (160.92,78.68) .. (160.92,81.14) .. controls (160.92,83.61) and (158.92,85.61) .. (156.46,85.61) .. controls (154,85.61) and (152,83.61) .. (152,81.14) -- cycle ;
			\draw  [fill={rgb, 255:red, 255; green, 255; blue, 255 }  ,fill opacity=1 ][line width=0.75]  (165.79,81.14) .. controls (165.79,78.68) and (167.79,76.68) .. (170.25,76.68) .. controls (172.71,76.68) and (174.71,78.68) .. (174.71,81.14) .. controls (174.71,83.61) and (172.71,85.61) .. (170.25,85.61) .. controls (167.79,85.61) and (165.79,83.61) .. (165.79,81.14) -- cycle ;
			\draw  [fill={rgb, 255:red, 255; green, 255; blue, 255 }  ,fill opacity=1 ][line width=0.75]  (173.9,103.04) .. controls (173.9,100.58) and (175.9,98.58) .. (178.36,98.58) .. controls (180.82,98.58) and (182.82,100.58) .. (182.82,103.04) .. controls (182.82,105.51) and (180.82,107.51) .. (178.36,107.51) .. controls (175.9,107.51) and (173.9,105.51) .. (173.9,103.04) -- cycle ;
			\draw  [fill={rgb, 255:red, 255; green, 255; blue, 255 }  ,fill opacity=1 ][line width=0.75]  (157.68,103.04) .. controls (157.68,100.58) and (159.68,98.58) .. (162.14,98.58) .. controls (164.6,98.58) and (166.6,100.58) .. (166.6,103.04) .. controls (166.6,105.51) and (164.6,107.51) .. (162.14,107.51) .. controls (159.68,107.51) and (157.68,105.51) .. (157.68,103.04) -- cycle ;
			\draw    (236.25,158.46) -- (244.71,179.14) ;
			\draw    (226.92,179.14) -- (236.25,158.46) ;
			\draw  [fill={rgb, 255:red, 255; green, 255; blue, 255 }  ,fill opacity=1 ][line width=2.25]  (231.79,158.46) .. controls (231.79,156) and (233.79,154) .. (236.25,154) .. controls (238.71,154) and (240.71,156) .. (240.71,158.46) .. controls (240.71,160.92) and (238.71,162.92) .. (236.25,162.92) .. controls (233.79,162.92) and (231.79,160.92) .. (231.79,158.46) -- cycle ;
			\draw  [fill={rgb, 255:red, 255; green, 255; blue, 255 }  ,fill opacity=1 ][line width=0.75]  (240.25,179.14) .. controls (240.25,176.68) and (242.25,174.68) .. (244.71,174.68) .. controls (247.17,174.68) and (249.17,176.68) .. (249.17,179.14) .. controls (249.17,181.61) and (247.17,183.61) .. (244.71,183.61) .. controls (242.25,183.61) and (240.25,181.61) .. (240.25,179.14) -- cycle ;
			\draw  [fill={rgb, 255:red, 255; green, 255; blue, 255 }  ,fill opacity=1 ][line width=0.75]  (222.46,179.14) .. controls (222.46,176.68) and (224.46,174.68) .. (226.92,174.68) .. controls (229.39,174.68) and (231.38,176.68) .. (231.38,179.14) .. controls (231.38,181.61) and (229.39,183.61) .. (226.92,183.61) .. controls (224.46,183.61) and (222.46,181.61) .. (222.46,179.14) -- cycle ;
			\draw  [fill={rgb, 255:red, 255; green, 255; blue, 255 }  ,fill opacity=1 ][line width=0.75]  (222.46,199.83) .. controls (222.46,197.36) and (224.46,195.37) .. (226.92,195.37) .. controls (229.39,195.37) and (231.38,197.36) .. (231.38,199.83) .. controls (231.38,202.29) and (229.39,204.29) .. (226.92,204.29) .. controls (224.46,204.29) and (222.46,202.29) .. (222.46,199.83) -- cycle ;
			\draw  [fill={rgb, 255:red, 255; green, 255; blue, 255 }  ,fill opacity=1 ][line width=2.25]  (102.79,177.46) .. controls (102.79,175) and (104.79,173) .. (107.25,173) .. controls (109.71,173) and (111.71,175) .. (111.71,177.46) .. controls (111.71,179.92) and (109.71,181.92) .. (107.25,181.92) .. controls (104.79,181.92) and (102.79,179.92) .. (102.79,177.46) -- cycle ;
			\draw  [fill={rgb, 255:red, 255; green, 255; blue, 255 }  ,fill opacity=1 ][line width=0.75]  (222.46,220.51) .. controls (222.46,218.05) and (224.46,216.05) .. (226.92,216.05) .. controls (229.39,216.05) and (231.38,218.05) .. (231.38,220.51) .. controls (231.38,222.97) and (229.39,224.97) .. (226.92,224.97) .. controls (224.46,224.97) and (222.46,222.97) .. (222.46,220.51) -- cycle ;
			\draw    (395.25,135.46) -- (441.25,135.46) ;
			\draw    (441.25,135.46) -- (483.5,135.46) ;
			\draw    (483.5,135.46) -- (525.75,135.46) ;
			\draw    (395.25,156.14) -- (387.14,178.04) ;
			\draw    (395.25,156.14) -- (403.36,178.04) ;
			\draw    (395.25,135.46) -- (395.25,156.14) ;
			\draw    (395.25,135.46) -- (409.04,156.14) ;
			\draw    (381.46,156.14) -- (395.25,135.46) ;
			\draw  [fill={rgb, 255:red, 255; green, 255; blue, 255 }  ,fill opacity=1 ][line width=0.75]  (390.79,135.46) .. controls (390.79,133) and (392.79,131) .. (395.25,131) .. controls (397.71,131) and (399.71,133) .. (399.71,135.46) .. controls (399.71,137.92) and (397.71,139.92) .. (395.25,139.92) .. controls (392.79,139.92) and (390.79,137.92) .. (390.79,135.46) -- cycle ;
			\draw  [fill={rgb, 255:red, 255; green, 255; blue, 255 }  ,fill opacity=1 ][line width=0.75]  (404.58,156.14) .. controls (404.58,153.68) and (406.58,151.68) .. (409.04,151.68) .. controls (411.5,151.68) and (413.5,153.68) .. (413.5,156.14) .. controls (413.5,158.61) and (411.5,160.61) .. (409.04,160.61) .. controls (406.58,160.61) and (404.58,158.61) .. (404.58,156.14) -- cycle ;
			\draw  [fill={rgb, 255:red, 255; green, 255; blue, 255 }  ,fill opacity=1 ][line width=0.75]  (377,156.14) .. controls (377,153.68) and (379,151.68) .. (381.46,151.68) .. controls (383.92,151.68) and (385.92,153.68) .. (385.92,156.14) .. controls (385.92,158.61) and (383.92,160.61) .. (381.46,160.61) .. controls (379,160.61) and (377,158.61) .. (377,156.14) -- cycle ;
			\draw  [fill={rgb, 255:red, 255; green, 255; blue, 255 }  ,fill opacity=1 ][line width=0.75]  (390.79,156.14) .. controls (390.79,153.68) and (392.79,151.68) .. (395.25,151.68) .. controls (397.71,151.68) and (399.71,153.68) .. (399.71,156.14) .. controls (399.71,158.61) and (397.71,160.61) .. (395.25,160.61) .. controls (392.79,160.61) and (390.79,158.61) .. (390.79,156.14) -- cycle ;
			\draw  [fill={rgb, 255:red, 255; green, 255; blue, 255 }  ,fill opacity=1 ][line width=0.75]  (398.9,178.04) .. controls (398.9,175.58) and (400.9,173.58) .. (403.36,173.58) .. controls (405.82,173.58) and (407.82,175.58) .. (407.82,178.04) .. controls (407.82,180.51) and (405.82,182.51) .. (403.36,182.51) .. controls (400.9,182.51) and (398.9,180.51) .. (398.9,178.04) -- cycle ;
			\draw  [fill={rgb, 255:red, 255; green, 255; blue, 255 }  ,fill opacity=1 ][line width=0.75]  (382.68,178.04) .. controls (382.68,175.58) and (384.68,173.58) .. (387.14,173.58) .. controls (389.6,173.58) and (391.6,175.58) .. (391.6,178.04) .. controls (391.6,180.51) and (389.6,182.51) .. (387.14,182.51) .. controls (384.68,182.51) and (382.68,180.51) .. (382.68,178.04) -- cycle ;
			\draw    (431.92,177.83) -- (431.92,198.51) ;
			\draw    (431.92,157.14) -- (431.92,177.83) ;
			\draw    (441.25,136.46) -- (449.71,157.14) ;
			\draw    (431.92,157.14) -- (441.25,136.46) ;
			\draw  [fill={rgb, 255:red, 255; green, 255; blue, 255 }  ,fill opacity=1 ][line width=0.75]  (436.79,136.46) .. controls (436.79,134) and (438.79,132) .. (441.25,132) .. controls (443.71,132) and (445.71,134) .. (445.71,136.46) .. controls (445.71,138.92) and (443.71,140.92) .. (441.25,140.92) .. controls (438.79,140.92) and (436.79,138.92) .. (436.79,136.46) -- cycle ;
			\draw  [fill={rgb, 255:red, 255; green, 255; blue, 255 }  ,fill opacity=1 ][line width=0.75]  (445.25,157.14) .. controls (445.25,154.68) and (447.25,152.68) .. (449.71,152.68) .. controls (452.17,152.68) and (454.17,154.68) .. (454.17,157.14) .. controls (454.17,159.61) and (452.17,161.61) .. (449.71,161.61) .. controls (447.25,161.61) and (445.25,159.61) .. (445.25,157.14) -- cycle ;
			\draw  [fill={rgb, 255:red, 255; green, 255; blue, 255 }  ,fill opacity=1 ][line width=0.75]  (427.46,157.14) .. controls (427.46,154.68) and (429.46,152.68) .. (431.92,152.68) .. controls (434.39,152.68) and (436.38,154.68) .. (436.38,157.14) .. controls (436.38,159.61) and (434.39,161.61) .. (431.92,161.61) .. controls (429.46,161.61) and (427.46,159.61) .. (427.46,157.14) -- cycle ;
			\draw  [fill={rgb, 255:red, 255; green, 255; blue, 255 }  ,fill opacity=1 ][line width=0.75]  (427.46,177.83) .. controls (427.46,175.36) and (429.46,173.37) .. (431.92,173.37) .. controls (434.39,173.37) and (436.38,175.36) .. (436.38,177.83) .. controls (436.38,180.29) and (434.39,182.29) .. (431.92,182.29) .. controls (429.46,182.29) and (427.46,180.29) .. (427.46,177.83) -- cycle ;
			\draw  [fill={rgb, 255:red, 255; green, 255; blue, 255 }  ,fill opacity=1 ][line width=0.75]  (427.46,198.51) .. controls (427.46,196.05) and (429.46,194.05) .. (431.92,194.05) .. controls (434.39,194.05) and (436.38,196.05) .. (436.38,198.51) .. controls (436.38,200.97) and (434.39,202.97) .. (431.92,202.97) .. controls (429.46,202.97) and (427.46,200.97) .. (427.46,198.51) -- cycle ;
			\draw  [fill={rgb, 255:red, 255; green, 255; blue, 255 }  ,fill opacity=1 ][line width=0.75]  (479.04,135.46) .. controls (479.04,133) and (481.04,131) .. (483.5,131) .. controls (485.96,131) and (487.96,133) .. (487.96,135.46) .. controls (487.96,137.92) and (485.96,139.92) .. (483.5,139.92) .. controls (481.04,139.92) and (479.04,137.92) .. (479.04,135.46) -- cycle ;
			\draw    (526.25,156.14) -- (518.14,178.04) ;
			\draw    (526.25,156.14) -- (534.36,178.04) ;
			\draw    (526.25,135.46) -- (526.25,156.14) ;
			\draw    (526.25,135.46) -- (540.04,156.14) ;
			\draw    (512.46,156.14) -- (526.25,135.46) ;
			\draw  [fill={rgb, 255:red, 255; green, 255; blue, 255 }  ,fill opacity=1 ][line width=0.75]  (521.79,135.46) .. controls (521.79,133) and (523.79,131) .. (526.25,131) .. controls (528.71,131) and (530.71,133) .. (530.71,135.46) .. controls (530.71,137.92) and (528.71,139.92) .. (526.25,139.92) .. controls (523.79,139.92) and (521.79,137.92) .. (521.79,135.46) -- cycle ;
			\draw  [fill={rgb, 255:red, 255; green, 255; blue, 255 }  ,fill opacity=1 ][line width=0.75]  (535.58,156.14) .. controls (535.58,153.68) and (537.58,151.68) .. (540.04,151.68) .. controls (542.5,151.68) and (544.5,153.68) .. (544.5,156.14) .. controls (544.5,158.61) and (542.5,160.61) .. (540.04,160.61) .. controls (537.58,160.61) and (535.58,158.61) .. (535.58,156.14) -- cycle ;
			\draw  [fill={rgb, 255:red, 255; green, 255; blue, 255 }  ,fill opacity=1 ][line width=0.75]  (508,156.14) .. controls (508,153.68) and (510,151.68) .. (512.46,151.68) .. controls (514.92,151.68) and (516.92,153.68) .. (516.92,156.14) .. controls (516.92,158.61) and (514.92,160.61) .. (512.46,160.61) .. controls (510,160.61) and (508,158.61) .. (508,156.14) -- cycle ;
			\draw  [fill={rgb, 255:red, 255; green, 255; blue, 255 }  ,fill opacity=1 ][line width=0.75]  (521.79,156.14) .. controls (521.79,153.68) and (523.79,151.68) .. (526.25,151.68) .. controls (528.71,151.68) and (530.71,153.68) .. (530.71,156.14) .. controls (530.71,158.61) and (528.71,160.61) .. (526.25,160.61) .. controls (523.79,160.61) and (521.79,158.61) .. (521.79,156.14) -- cycle ;
			\draw  [fill={rgb, 255:red, 255; green, 255; blue, 255 }  ,fill opacity=1 ][line width=0.75]  (529.9,178.04) .. controls (529.9,175.58) and (531.9,173.58) .. (534.36,173.58) .. controls (536.82,173.58) and (538.82,175.58) .. (538.82,178.04) .. controls (538.82,180.51) and (536.82,182.51) .. (534.36,182.51) .. controls (531.9,182.51) and (529.9,180.51) .. (529.9,178.04) -- cycle ;
			\draw  [fill={rgb, 255:red, 255; green, 255; blue, 255 }  ,fill opacity=1 ][line width=0.75]  (513.68,178.04) .. controls (513.68,175.58) and (515.68,173.58) .. (518.14,173.58) .. controls (520.6,173.58) and (522.6,175.58) .. (522.6,178.04) .. controls (522.6,180.51) and (520.6,182.51) .. (518.14,182.51) .. controls (515.68,182.51) and (513.68,180.51) .. (513.68,178.04) -- cycle ;
			
			\draw (389,108.4) node [anchor=north west][inner sep=0.75pt]    {$u$};
			\draw (520,108.4) node [anchor=north west][inner sep=0.75pt]    {$v$};
			\draw (454,80.4) node [anchor=north west][inner sep=0.75pt]    {${\displaystyle T}$};

		\end{tikzpicture}
		
		\caption{An finite automata~$\A$ with transitions labeled by rooted trees, and a tree~$T$ such that $\sigma(T,u,v)$ is accepted by~$\A$. Here the state with the incoming arrow is the only initial state, and the black state is the only final state.}
		\label{fig:automata_trees}
	\end{figure}
	
	Let us now explain how we can define such an automaton with the certification procedure.
	Informally speaking, if a vertex $r$ accepts with two neighbors $u,v$ and a tree~$T$ rooted on~$r$, with certificates $c_1,c_2,c_3$ assigned to $u,r,v$ respectively (and some other certificates assigned to the vertices of~$T$), then we add a transition from the state~$(c_1,c_2)$ to the state~$(c_2,c_3)$ labeled $T$.
	Let us formalize this definition.
	
	For every $k \in \N$, we construct the following non-deterministic automaton $\A_k$. The set of states consists of all pairs of certificates of size~$k$ plus two additional states, $i$ and~$f$. So the set of states is: $C_k^2 \cup \{i,f\}$. There is a single initial state, which is~$i$, and a single final state, which is~$f$. The transitions\footnote{The three cases correspond to the fact that a node can be at the beginning, in the middle or at the end of a path going from a node to another.} are the following:
	\begin{itemize}
		\item \textit{(Middle node transitions)} for every $c_1, c_2, c_3 \in C_k$ and every rooted tree~$T$, consider the tree $T'$ obtained by taking a copy of $T$ and adding two additional neighbors~$u$ and~$v$ to the root~$r$ of~$T$ (and seeing $T'$ as a non-rooted tree). We create a transition from $(c_1,c_2)$ to $(c_2,c_3)$ labeled by~$T$ if and only if there exists a certificate assignment $c : V(T') \to C_k$ such that $c(u) = c_1$, $c(r) = c_2$, $c(v) = c_3$, and all the vertices in $V(T') \setminus \{u,v\}$ accept.
		
		\item \textit{(Initial node transitions)} for every $c_1, c_2 \in C_k$ and every rooted tree~$T$, consider the tree $T'$ obtained by taking a copy of~$T$ and adding one additional neighbor~$v$ to the root $r$ of $T$. We create a transition from $i$ to $(c_1, c_2)$ labeled by~$T$ if and only if there exists a certificate assignment $c : V(T') \to C_k$ such that $c(r) = c_1$, $c(v)=c_2$, and all the vertices in $V(T') \setminus \{v\}$ accept.
		
		\item \textit{(Final node transitions)} for every $c_1, c_2 \in C_k$ and every rooted tree~$T$, consider the tree $T'$ obtained by taking a copy of~$T$ and adding one additional neighbor~$u$ to the root $r$ of $T$. We create a transition from~$(c_1, c_2)$ to~$f$ labeled by~$T$ if and only if there exists a certificate assignment $c : V(T') \to C_k$ such that $c(u) = c_1$, $c(r)=c_2$, and all the vertices in $V(T') \setminus \{u\}$ accept.
		
		\item for every rooted tree~$T$, we put a transition from~$i$ to~$f$ labeled by~$T$ if and only if there exists a certificate assignment to the vertices of $T$ (seen as a non-rooted tree) such that all the vertices accept.
	\end{itemize}

	Now, we prove the following result:
	
	\begin{proposition}
		\label{prop:trees_accepted}
		Let $T$ be a tree. The following are equivalent:
		
		\begin{enumerate}[label=(\roman*)]
			\item $T \in S_k$
			\item For all $u,v \in V(T)$, $\sigma(T,u,v)$ is accepted by~$\A_k$.
			\item There exist $u,v \in V(T)$ such that $\sigma(T,u,v)$ is accepted by~$\A_k$.
		\end{enumerate}
	\end{proposition}
	
	\begin{proof}
		Let us first prove the implication $(i) \Rightarrow (ii)$. Assume that $T \in S_k$. Then, there exists an assignment of certificates in~$C_k$ to the vertices of~$T$ such that every vertex accepts. Let $u,v \in V(T)$, and let $T_1, \ldots, T_m := \sigma(T,u,v)$.
		If $u = v$, then $m=1$ and $T_1$ is the tree $T$ rooted at $u$. By definition of $\A_k$, there is a transition from~$i$ to~$f$ labeled by~$T_1$, and the result follows. If $u \neq v$, for every $i \in \{1, \ldots, m\}$, let~$w_i$ be the vertex in~$T$ corresponding to the root of~$T_i$ (with $w_1 = u$ and $w_m = v$), and let $c_i$ be the certificate assigned to~$w_i$. 
		By definition of~$\A_k$, there is a transition from $i$ to $(c_1, c_2)$ labeled by $T_1$ (since $w_1$ accepts), there is a transition from $(c_{m-1}, c_m)$ to $f$ labeled by $T_m$ (since $w_m$ accepts), and for every $i \in \{2, \ldots, m-1\}$ there is a transition from $(c_{i-1}, c_i)$ to $(c_i, c_{i+1})$ labeled by $T_i$ (since $w_i$ accepts). Thus, $\sigma(T,u,v)$ is accepted by $\A_k$.
		
		The implication $(ii) \Rightarrow (iii)$ is direct. 
		
		Let us finally prove $(iii) \Rightarrow (i)$. 
		Assume that there exist $u,v \in V(T)$ such that $\sigma(T,u,v) = T_1, \ldots, T_m$ is accepted by~$\A_k$. Let us show that there exists an assignment of certificates in~$C_k$ to the vertices of~$T$ such that each vertex accepts. If $m=1$, then $T_1$ is the tree~$T$ rooted at~$u$, and there is a transition $i \xrightarrow{T_1} f$ in $\A_k$. By definition of the transitions, there exists a certificate assignment to the vertices of~$T$ such that all the vertices accept. Assume now that $m \geqslant 2$. Since $\sigma(T,u,v)$ is accepted by~$\A_k$, there exists a sequence of transitions in $\A_k$ of the form:
		
		$$i \xrightarrow{T_1} (c_1, c_2) \xrightarrow{T_2} \;\ldots\; \xrightarrow{T_{m-1}} (c_{m-1}, c_m) \xrightarrow{T_m} f$$
		
		For every $i \in \{1, \ldots, m\}$, let $w_i$ be the vertex in~$T$ corresponding to the root of $T_i$.
		For every $i \in \{1, \ldots, m\}$, we assign the certificate $c_i$ to $w_i$. Then, since there is the transition $i \xrightarrow{T_1} (c_1,c_2)$ in~$\A_k$, there exists an assignment of certificates to the vertices of~$T_1$ such that $w_1$ and all the vertices in~$T_1$ accept. For every $i \in \{2, \ldots, m-1\}$, since there is the transition $(c_{i-1},c_i) \xrightarrow{T_i} (c_i, c_{i+1})$ in $\A_k$, there exists an assignment of certificates to the vertices of~$T_i$ such that $w_i$ and all the vertices in~$T_i$ accept. Finally, since there is the transition $(c_{m-1},c_m) \xrightarrow{T_m} f$ in $\A_k$, there exists an assignment of certificates to the vertices of~$T_m$ such that $w_m$ and all the vertices in~$T_m$ accept. So $T \in S_k$.
	\end{proof}

	\begin{corollary}
		\label{cor:trees_accepted}
		For every $k \in \N$, the language over the alphabet of rooted trees accepted by~$\A_k$ is equal to $\{\sigma(T,u,v) \; | \; T \in S_k, u,v \in V(T)\}$, that is, to the set of all the parsings of the trees in~$S_k$.
	\end{corollary}
	
	\begin{proof}
		If~$T \in S_k$, by Proposition~\ref{prop:trees_accepted} we know that $\sigma(T,u,v)$ is accepted by~$\A_k$. Conversely, let $T_1,\ldots, T_m$ be a sequence of rooted trees accepted by~$\A_k$. By Lemma~\ref{lem:parsing}, there exists two vertices $u,v$ of $\tau(T_1, \ldots, T_m)$ such that $T_1, \ldots, T_m = \sigma(\tau(T_1, \ldots, T_m),u,v)$. Finally, again by Proposition~\ref{prop:trees_accepted}, we get $\tau(T_1, \ldots, T_m) \in S_k$.
	\end{proof}

	
	\subsection{Proof of Theorem~\ref{thm:gap-tree-radius-1}}
	
	We prove Theorem~\ref{thm:gap-tree-radius-1} by contradiction. Assume that $\P$ can not be certified with constant-size certificates. Then, the set $X \subseteq \N$ of integers $k \in \N$ such that $S_k \not \subseteq \bigcup_{1 \leqslant i < k}S_i$ is infinite. For $k \in X$, let $d_k$ be the smallest diameter of a tree in $S_k \setminus \bigcup_{1 \leqslant i < k}S_i$. Since any tree that has diameter~$d_k$ is in the set $S_{s(d_k)}$, we have $s(d_k) \geqslant k$. Let us fix an integer $k \in X$ such that $k > \max_{d \in \{1, \ldots, D-1\}} s(d)$. Since $s(d_k) \geqslant k$ and $k > s(d)$ for all $d < D$, we have $d_k \geqslant D$. Moreover, since $s(d) = \lfloor (\log \log d) / c \rfloor$ for all $d \geqslant D$, the inequality $s(d_k) \geqslant k$ implies that $d_k \geqslant 2^{2^{ck}}$.

	By Corollary~\ref{cor:trees_accepted}, for every $i \in \N$, the language over the alphabet of rooted trees accepted by $\A_i$ is equal to the set of all the parsings of the trees in~$S_i$. Moreover, $\A_i$ has $M_i := 2^{2i}+2$ states. Thus, by Lemma~\ref{lem:automata}, there exists an automaton~$\A_k'$ that accepts the parsings of the trees in $S_k \setminus \bigcup_{i<k} S_i$ has at most $M_k \cdot 2^{\left(\sum_{i=1}^{k-1}M_i\right)} \leqslant 2^{2^{2k}}$ states. Let $T^{(k)} \in S_k \setminus \bigcup_{i < k}S_i$ be a tree such that $T^{(k)}$ has diameter~$d_k$, and among the trees in~$S_k \setminus \bigcup_{i < k}S_i$ having diameter~$d_k$, $T^{(k)}$ has the smallest number of vertices.
	To conclude the proof of Theorem~\ref{thm:gap-tree-radius-1}, it is sufficient to prove the following claim, that contradicts the inequality $d_k \geqslant 2^{2^{kc}}$.
	
	\begin{claim}
		\label{claim:bound d_k trees}
		We have $d_k \leqslant 2^{2^{2k}}$.
	\end{claim}
	
	\begin{proof}
		Let $u, v \in V(T^{(k)})$, and let us consider a run accepting $T_1, \ldots, T_m :=\sigma(T^{(k)},u,v)$ in~$\A_k'$. Let us prove that this run does not use the same state twice. Assume by contradiction that it uses the same state~$p$ twice. Then, consider the run obtained by deleting the part between the first and last times that~$p$ is visited. This run accepts another sequence of rooted trees $T'_1, \ldots, T'_{m'}$. By Lemma~\ref{lem:gluing_subsequence}, the diameter of~$\tau(T'_1, \ldots, T'_{m'})$ is less or equal than the diameter of~$T^{(k)}$ which is $d_k$. So the diameter of~$\tau(T'_1, \ldots, T'_{m'})$ is equal to $d_k$ because $T'_1, \ldots, T'_{m'}$ is accepted by~$\A_k$, so $\tau(T'_1, \ldots, T'_{m'}) \in S_k \setminus \bigcup_{i<k} S_i$ and $d_k$ is the smallest diameter of a tree in this set.
		Again by Lemma~\ref{lem:gluing_subsequence}, $\tau(T'_1, \ldots, T'_{m'})$ has strictly less vertices than~$T^{(k)}$, and this is a contradiction by definition of~$T^{(k)}$. Thus, the length of a run in $\A_k'$ accepting $\sigma(T^{(k)},u,v)$ is at most $2^{2^{2k}}$ for every $u,v \in V(T^{(k)})$.
		In particular, by considering $u, v \in V(T^{(k)})$ such that the distance from~$u$ to~$v$ in $T^{(k)}$ is~$d_k$, we obtain that $d_k \leqslant 2^{2^{2k}}$.
	\end{proof}

		
		
		
		
		
	
	\section{Gap in bounded degree caterpillars with larger radius}
	\label{sec:gap-caterpillar-bounded-deg}
	
	\begin{theorem}
		\label{thm:gap_caterpillar_bounded_degree}
		Let $\Delta, r \geqslant 1$. Assume that the vertices can see at distance~$r$. Let $\P$ be a property on caterpillars of maximum degree at most~$\Delta$, that can be certified with certificates of size $s(n)$, with $s(n)=o(\log \log n)$. Then, $\P$ can be certified with certificates of size~$O(1)$.
	\end{theorem}

	\begin{remark}
		There exists a constant $c_{r,\Delta}$ such that Theorem~\ref{thm:gap_caterpillar_bounded_degree} remains true by replacing $s(n) = o(\log \log n)$ by $s(n) = \left\lfloor\frac{\log \log n}{c_{r,\Delta}}\right\rfloor$ (because such a constant exists in Theorem~\ref{thm:gap_labeled_distance>1}, and the proof of Theorem~\ref{thm:gap_caterpillar_bounded_degree} is a reduction to Theorem~\ref{thm:gap_labeled_distance>1}), but we did not try to determine and optimize this constant (it depends on the encoding used in the reduction, see the proof of Theorem~\ref{thm:gap_caterpillar_bounded_degree}).
		Moreover, the statement of Theorem~\ref{thm:gap_caterpillar_bounded_degree} also holds if we replace~$n$ by~$d$, the diameter: indeed, in caterpillars of maximum degree~$\Delta$, we have $d \leqslant n \leqslant \Delta (d+1)$.
	\end{remark}
	
	\begin{proof}
		We will use Theorem~\ref{thm:gap_labeled_distance>1}.
		Let $P$ be a path labeled by $\Sigma:=\{0, \ldots, \Delta\}$, and $G$ be a caterpillar of maximum degree~$\Delta$. For every $u \in V(P)$, let $\ell(u)$ be its label. Let $f(P)$ be the (unlabeled) caterpillar obtained by adding $\ell(u)$ degree-1 neighbors to every $u \in V(P)$. Let $h(G)$ be the labeled path $P$, where~$P$ is the central path of~$G$ and for every $u \in V(P)$ the label of~$u$ is its number of degree-1 neighbors in~$G$. See Figure~\ref{fig:f_caterpillar_bounded_degree} for an example. The following claim is straightforward from the definition of $f$ and $h$:
		
		\begin{claim}
			\label{claim:fg}
			For every caterpillar~$G$ of maximum degree~$\Delta$, we have $f(h(G))=G$. 
		\end{claim}

		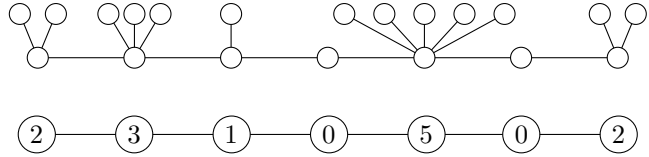
\begin{figure}[h!]
			\centering
			
			\begin{tikzpicture}[x=0.75pt,y=0.75pt,yscale=-1,xscale=1]
				
				\draw    (403.6,164.84) -- (452.24,164.84) ;
				\draw    (258.75,126.25) -- (258.75,104.25) ;
				\draw    (403.5,126.25) -- (451.75,126.25) ;
				\draw    (307.66,164.84) -- (356.29,164.84) ;
				\draw    (259.03,164.84) -- (307.66,164.84) ;
				\draw    (356.29,164.84) -- (404.93,164.84) ;
				\draw    (210.4,164.84) -- (259.03,164.84) ;
				\draw    (161.76,164.84) -- (210.4,164.84) ;
				\draw  [fill={rgb, 255:red, 255; green, 255; blue, 255 }  ,fill opacity=1 ] (298.4,164.84) .. controls (298.4,159.72) and (302.55,155.58) .. (307.66,155.58) .. controls (312.78,155.58) and (316.93,159.72) .. (316.93,164.84) .. controls (316.93,169.96) and (312.78,174.1) .. (307.66,174.1) .. controls (302.55,174.1) and (298.4,169.96) .. (298.4,164.84) -- cycle ;
				\draw  [fill={rgb, 255:red, 255; green, 255; blue, 255 }  ,fill opacity=1 ] (347.03,164.84) .. controls (347.03,159.72) and (351.18,155.58) .. (356.29,155.58) .. controls (361.41,155.58) and (365.56,159.72) .. (365.56,164.84) .. controls (365.56,169.96) and (361.41,174.1) .. (356.29,174.1) .. controls (351.18,174.1) and (347.03,169.96) .. (347.03,164.84) -- cycle ;
				\draw  [fill={rgb, 255:red, 255; green, 255; blue, 255 }  ,fill opacity=1 ] (249.77,164.84) .. controls (249.77,159.72) and (253.91,155.58) .. (259.03,155.58) .. controls (264.14,155.58) and (268.29,159.72) .. (268.29,164.84) .. controls (268.29,169.96) and (264.14,174.1) .. (259.03,174.1) .. controls (253.91,174.1) and (249.77,169.96) .. (249.77,164.84) -- cycle ;
				\draw  [fill={rgb, 255:red, 255; green, 255; blue, 255 }  ,fill opacity=1 ] (152.5,164.84) .. controls (152.5,159.72) and (156.65,155.58) .. (161.76,155.58) .. controls (166.88,155.58) and (171.03,159.72) .. (171.03,164.84) .. controls (171.03,169.96) and (166.88,174.1) .. (161.76,174.1) .. controls (156.65,174.1) and (152.5,169.96) .. (152.5,164.84) -- cycle ;
				\draw  [fill={rgb, 255:red, 255; green, 255; blue, 255 }  ,fill opacity=1 ] (394.34,164.84) .. controls (394.34,159.72) and (398.49,155.58) .. (403.6,155.58) .. controls (408.72,155.58) and (412.87,159.72) .. (412.87,164.84) .. controls (412.87,169.96) and (408.72,174.1) .. (403.6,174.1) .. controls (398.49,174.1) and (394.34,169.96) .. (394.34,164.84) -- cycle ;
				\draw  [fill={rgb, 255:red, 255; green, 255; blue, 255 }  ,fill opacity=1 ] (201.13,164.84) .. controls (201.13,159.72) and (205.28,155.58) .. (210.4,155.58) .. controls (215.51,155.58) and (219.66,159.72) .. (219.66,164.84) .. controls (219.66,169.96) and (215.51,174.1) .. (210.4,174.1) .. controls (205.28,174.1) and (201.13,169.96) .. (201.13,164.84) -- cycle ;
				\draw    (355.25,126.25) -- (315.25,104.25) ;
				\draw    (355.25,126.25) -- (395.25,104.25) ;
				\draw    (355.25,126.25) -- (355.25,104.25) ;
				\draw    (355.25,126.25) -- (375.25,104.25) ;
				\draw    (355.25,126.25) -- (335.25,104.25) ;
				\draw    (210.5,126.25) -- (210.5,104.25) ;
				\draw    (210.5,126.25) -- (223.5,104.25) ;
				\draw    (210.5,126.25) -- (197.5,104.25) ;
				\draw    (162.3,126.25) -- (171.3,104.25) ;
				\draw    (162.25,126.25) -- (153.3,104.25) ;
				\draw    (162.25,126.25) -- (210.5,126.25) ;
				\draw    (210.5,126.25) -- (258.75,126.25) ;
				\draw    (258.75,126.25) -- (307,126.25) ;
				\draw    (307,126.25) -- (355.25,126.25) ;
				\draw    (355.25,126.25) -- (403.5,126.25) ;
				\draw  [fill={rgb, 255:red, 255; green, 255; blue, 255 }  ,fill opacity=1 ] (157,126.25) .. controls (157,123.35) and (159.35,121) .. (162.25,121) .. controls (165.15,121) and (167.5,123.35) .. (167.5,126.25) .. controls (167.5,129.15) and (165.15,131.5) .. (162.25,131.5) .. controls (159.35,131.5) and (157,129.15) .. (157,126.25) -- cycle ;
				\draw  [fill={rgb, 255:red, 255; green, 255; blue, 255 }  ,fill opacity=1 ] (205.25,126.25) .. controls (205.25,123.35) and (207.6,121) .. (210.5,121) .. controls (213.4,121) and (215.75,123.35) .. (215.75,126.25) .. controls (215.75,129.15) and (213.4,131.5) .. (210.5,131.5) .. controls (207.6,131.5) and (205.25,129.15) .. (205.25,126.25) -- cycle ;
				\draw  [fill={rgb, 255:red, 255; green, 255; blue, 255 }  ,fill opacity=1 ] (253.5,126.25) .. controls (253.5,123.35) and (255.85,121) .. (258.75,121) .. controls (261.65,121) and (264,123.35) .. (264,126.25) .. controls (264,129.15) and (261.65,131.5) .. (258.75,131.5) .. controls (255.85,131.5) and (253.5,129.15) .. (253.5,126.25) -- cycle ;
				\draw  [fill={rgb, 255:red, 255; green, 255; blue, 255 }  ,fill opacity=1 ] (301.75,126.25) .. controls (301.75,123.35) and (304.1,121) .. (307,121) .. controls (309.9,121) and (312.25,123.35) .. (312.25,126.25) .. controls (312.25,129.15) and (309.9,131.5) .. (307,131.5) .. controls (304.1,131.5) and (301.75,129.15) .. (301.75,126.25) -- cycle ;
				\draw  [fill={rgb, 255:red, 255; green, 255; blue, 255 }  ,fill opacity=1 ] (350,126.25) .. controls (350,123.35) and (352.35,121) .. (355.25,121) .. controls (358.15,121) and (360.5,123.35) .. (360.5,126.25) .. controls (360.5,129.15) and (358.15,131.5) .. (355.25,131.5) .. controls (352.35,131.5) and (350,129.15) .. (350,126.25) -- cycle ;
				\draw  [fill={rgb, 255:red, 255; green, 255; blue, 255 }  ,fill opacity=1 ] (398.25,126.25) .. controls (398.25,123.35) and (400.6,121) .. (403.5,121) .. controls (406.4,121) and (408.75,123.35) .. (408.75,126.25) .. controls (408.75,129.15) and (406.4,131.5) .. (403.5,131.5) .. controls (400.6,131.5) and (398.25,129.15) .. (398.25,126.25) -- cycle ;
				\draw  [fill={rgb, 255:red, 255; green, 255; blue, 255 }  ,fill opacity=1 ] (148.05,104.25) .. controls (148.05,101.35) and (150.4,99) .. (153.3,99) .. controls (156.2,99) and (158.55,101.35) .. (158.55,104.25) .. controls (158.55,107.15) and (156.2,109.5) .. (153.3,109.5) .. controls (150.4,109.5) and (148.05,107.15) .. (148.05,104.25) -- cycle ;
				\draw  [fill={rgb, 255:red, 255; green, 255; blue, 255 }  ,fill opacity=1 ] (166.05,104.25) .. controls (166.05,101.35) and (168.4,99) .. (171.3,99) .. controls (174.2,99) and (176.55,101.35) .. (176.55,104.25) .. controls (176.55,107.15) and (174.2,109.5) .. (171.3,109.5) .. controls (168.4,109.5) and (166.05,107.15) .. (166.05,104.25) -- cycle ;
				\draw  [fill={rgb, 255:red, 255; green, 255; blue, 255 }  ,fill opacity=1 ] (192.25,104.25) .. controls (192.25,101.35) and (194.6,99) .. (197.5,99) .. controls (200.4,99) and (202.75,101.35) .. (202.75,104.25) .. controls (202.75,107.15) and (200.4,109.5) .. (197.5,109.5) .. controls (194.6,109.5) and (192.25,107.15) .. (192.25,104.25) -- cycle ;
				\draw  [fill={rgb, 255:red, 255; green, 255; blue, 255 }  ,fill opacity=1 ] (205.25,104.25) .. controls (205.25,101.35) and (207.6,99) .. (210.5,99) .. controls (213.4,99) and (215.75,101.35) .. (215.75,104.25) .. controls (215.75,107.15) and (213.4,109.5) .. (210.5,109.5) .. controls (207.6,109.5) and (205.25,107.15) .. (205.25,104.25) -- cycle ;
				\draw  [fill={rgb, 255:red, 255; green, 255; blue, 255 }  ,fill opacity=1 ] (218.25,104.25) .. controls (218.25,101.35) and (220.6,99) .. (223.5,99) .. controls (226.4,99) and (228.75,101.35) .. (228.75,104.25) .. controls (228.75,107.15) and (226.4,109.5) .. (223.5,109.5) .. controls (220.6,109.5) and (218.25,107.15) .. (218.25,104.25) -- cycle ;
				\draw  [fill={rgb, 255:red, 255; green, 255; blue, 255 }  ,fill opacity=1 ] (310,104.25) .. controls (310,101.35) and (312.35,99) .. (315.25,99) .. controls (318.15,99) and (320.5,101.35) .. (320.5,104.25) .. controls (320.5,107.15) and (318.15,109.5) .. (315.25,109.5) .. controls (312.35,109.5) and (310,107.15) .. (310,104.25) -- cycle ;
				\draw  [fill={rgb, 255:red, 255; green, 255; blue, 255 }  ,fill opacity=1 ] (330,104.25) .. controls (330,101.35) and (332.35,99) .. (335.25,99) .. controls (338.15,99) and (340.5,101.35) .. (340.5,104.25) .. controls (340.5,107.15) and (338.15,109.5) .. (335.25,109.5) .. controls (332.35,109.5) and (330,107.15) .. (330,104.25) -- cycle ;
				\draw  [fill={rgb, 255:red, 255; green, 255; blue, 255 }  ,fill opacity=1 ] (350,104.25) .. controls (350,101.35) and (352.35,99) .. (355.25,99) .. controls (358.15,99) and (360.5,101.35) .. (360.5,104.25) .. controls (360.5,107.15) and (358.15,109.5) .. (355.25,109.5) .. controls (352.35,109.5) and (350,107.15) .. (350,104.25) -- cycle ;
				\draw  [fill={rgb, 255:red, 255; green, 255; blue, 255 }  ,fill opacity=1 ] (370,104.25) .. controls (370,101.35) and (372.35,99) .. (375.25,99) .. controls (378.15,99) and (380.5,101.35) .. (380.5,104.25) .. controls (380.5,107.15) and (378.15,109.5) .. (375.25,109.5) .. controls (372.35,109.5) and (370,107.15) .. (370,104.25) -- cycle ;
				\draw  [fill={rgb, 255:red, 255; green, 255; blue, 255 }  ,fill opacity=1 ] (390,104.25) .. controls (390,101.35) and (392.35,99) .. (395.25,99) .. controls (398.15,99) and (400.5,101.35) .. (400.5,104.25) .. controls (400.5,107.15) and (398.15,109.5) .. (395.25,109.5) .. controls (392.35,109.5) and (390,107.15) .. (390,104.25) -- cycle ;
				\draw  [fill={rgb, 255:red, 255; green, 255; blue, 255 }  ,fill opacity=1 ] (253.5,104.25) .. controls (253.5,101.35) and (255.85,99) .. (258.75,99) .. controls (261.65,99) and (264,101.35) .. (264,104.25) .. controls (264,107.15) and (261.65,109.5) .. (258.75,109.5) .. controls (255.85,109.5) and (253.5,107.15) .. (253.5,104.25) -- cycle ;
				\draw    (451.75,126.25) -- (460.75,104.25) ;
				\draw    (451.75,126.25) -- (442.8,104.25) ;
				\draw  [fill={rgb, 255:red, 255; green, 255; blue, 255 }  ,fill opacity=1 ] (446.5,126.25) .. controls (446.5,123.35) and (448.85,121) .. (451.75,121) .. controls (454.65,121) and (457,123.35) .. (457,126.25) .. controls (457,129.15) and (454.65,131.5) .. (451.75,131.5) .. controls (448.85,131.5) and (446.5,129.15) .. (446.5,126.25) -- cycle ;
				\draw  [fill={rgb, 255:red, 255; green, 255; blue, 255 }  ,fill opacity=1 ] (437.55,104.25) .. controls (437.55,101.35) and (439.9,99) .. (442.8,99) .. controls (445.7,99) and (448.05,101.35) .. (448.05,104.25) .. controls (448.05,107.15) and (445.7,109.5) .. (442.8,109.5) .. controls (439.9,109.5) and (437.55,107.15) .. (437.55,104.25) -- cycle ;
				\draw  [fill={rgb, 255:red, 255; green, 255; blue, 255 }  ,fill opacity=1 ] (455.5,104.25) .. controls (455.5,101.35) and (457.85,99) .. (460.75,99) .. controls (463.65,99) and (466,101.35) .. (466,104.25) .. controls (466,107.15) and (463.65,109.5) .. (460.75,109.5) .. controls (457.85,109.5) and (455.5,107.15) .. (455.5,104.25) -- cycle ;
				\draw  [fill={rgb, 255:red, 255; green, 255; blue, 255 }  ,fill opacity=1 ] (442.97,164.84) .. controls (442.97,159.72) and (447.12,155.58) .. (452.24,155.58) .. controls (457.35,155.58) and (461.5,159.72) .. (461.5,164.84) .. controls (461.5,169.96) and (457.35,174.1) .. (452.24,174.1) .. controls (447.12,174.1) and (442.97,169.96) .. (442.97,164.84) -- cycle ;
				
				\draw (157.09,159.58) node [anchor=north west][inner sep=0.75pt]    {$2$};
				\draw (205.73,159.58) node [anchor=north west][inner sep=0.75pt]    {$3$};
				\draw (254.2,159.58) node [anchor=north west][inner sep=0.75pt]    {$1$};
				\draw (303,159.58) node [anchor=north west][inner sep=0.75pt]    {$0$};
				\draw (351.5,159.58) node [anchor=north west][inner sep=0.75pt]    {$5$};
				\draw (399.1,159.58) node [anchor=north west][inner sep=0.75pt]    {$0$};
				\draw (447.57,159.58) node [anchor=north west][inner sep=0.75pt]    {$2$};

			\end{tikzpicture}

			\caption{A caterpillar $G$ and a labeled path~$P$, such that $P = h(G)$ and $G = f(P)$.}
			\label{fig:f_caterpillar_bounded_degree}
		\end{figure}

		Let $\P'$ be following the property on paths labeled by $\Sigma$: a labeled path~$P$ satisfies~$\P'$ if and only if $f(P)$ has maximum degree at most~$\Delta$ and satisfies $\P$.
		
		\begin{claim}
			\label{claim:constant_size_certif}
			The property~$\P'$ can be certified with constant-size certificates, if vertices can see at distance~$r$.
		\end{claim}
		
		\begin{proof}
			By Theorem~\ref{thm:gap_labeled_distance>1}, to prove that~$\P'$ can be certified with constant-size certificates if vertices can see at distance~$r$, it is sufficient to show that $\P'$ can be certified with certificates of size $o(\log \log n)$. 
			Such a certification for~$\P'$ is the following one. Let $P$ be a path labeled with $\{0, \ldots, \Delta\}$. To each vertex~$u$ of~$P$ with label $i$, the prover gives a $(i+1)$-tuple of certificates, containing the certificates that~$u$ and its $i$~additional degree-1 neighbors would receive in the certification scheme for~$\P$ in $f(P)$. The verification of each vertex simply consists checking that its degree plus its label is at most~$\Delta$, and then in reconstructing its view in~$f(P)$ and accepting if and only if it would accept in the verification procedure of $\P$. Finally, the size of these certificates is~$O(\Delta s(n\Delta))$, and since $\Delta$ is a constant and $s(n)=o(\log \log n)$, the size remains $o(\log \log n)$.
		\end{proof}
		
		Let us finally describe the constant-size certification procedure for~$\P$. Let $G$ be a caterpillar of maximum degree~$\Delta$, and let $P$ be its central path. The prover first gives a special certificate to each degree-1 vertex. Then, it writes in the certificate of every vertex~$u \in V(P)$ its number of degree-1 neighbors denoted by~$\delta_1(u)$, and the certificate~$u$ would receive in the certification procedure of~$\P'$ for the labeled path $h(G)$. By Claim~\ref{claim:constant_size_certif}, these certificates have constant size.

		The verification of the vertices consists in the following. First, each vertex checks that it has the special certificate if and only if it has degree~$1$. Then, each vertex~$u$ of~$P$ (that is, each vertex~$u$ that has degree at least two) checks that $\delta_1(u)$ is correctly written in its certificate. Finally, to check if $G$ satisfies~$\P$, since by Claim~\ref{claim:fg} we have $G = f(h(G))$, the vertices will check that $h(G)$ satisfies~$\P'$. To do so, each vertex~$u$ of~$P$ reconstructs its view in~$h(G)$ (note can $u$ can do so, because the both the labels and the certificates of the vertices of $h(G)$ are written in the certificates of the vertices of $G$), and~$u$ accepts if and only if it would accept in the certification procedure for~$\P'$ in~$h(G)$. The correctness of this certification procedure follows from the correctness of the certification procedure for~$\P'$.
	\end{proof}


	\section{Gap in paths with better constants via Chrobak normal form}
	\label{sec:Chrobak}
	
	In this section, we prove our gap theorem on paths with better constants than in Theorem~\ref{thm:gap-paths-state-complexity}, via a different technique.
	
	\ThmGapPathsChrobak*
	
	The proof is based on the following theorem, which is illustrated in Figure~\ref{fig:chrobak}.
	
	\begin{theorem}[Chrobak normal form]
		\label{thm:chrobak}
		Let $\A$ be a nondeterministic finite automaton, over a unary alphabet, having $M$ states. Then, there exists another nondeterministic finite automaton, called the \emph{Chrobak normal form of $\A$}, recognizing the same language, that consists of a path of length at most $\ell(M) = O(M^2)$, followed by a single nondeterministic choice to a set of $r$ disjoint cycles of lengths $p_1, \ldots, p_r$ such that $\sum_{1 \leqslant i \leqslant r} p_i \leqslant M$.
	\end{theorem}
	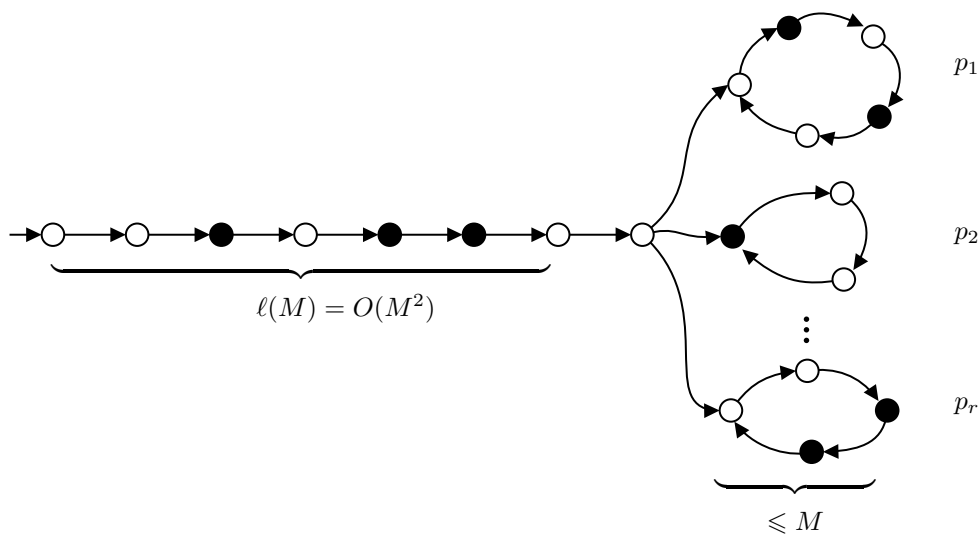
\begin{figure}[h]
		\centering
		
		\tikzset{every picture/.style={line width=0.75pt}} 
		
		\begin{tikzpicture}[x=0.65pt,y=0.65pt,yscale=-1,xscale=1]
			\draw    (19.5,137.5) -- (35,137.5) ;
			\draw [shift={(38,137.5)}, rotate = 180] [fill={rgb, 255:red, 0; green, 0; blue, 0 }  ][line width=0.08]  [draw opacity=0] (8.93,-4.29) -- (0,0) -- (8.93,4.29) -- cycle    ;
			
			\draw    (481.5,263.5) .. controls (470.58,266.83) and (450.17,262.48) .. (438.73,248.32) ;
			\draw [shift={(437,246)}, rotate = 55.49] [fill={rgb, 255:red, 0; green, 0; blue, 0 }  ][line width=0.08]  [draw opacity=0] (8.93,-4.29) -- (0,0) -- (8.93,4.29) -- cycle    ;
			\draw    (525,239.5) .. controls (523.57,257.17) and (511.19,262.52) .. (490.92,263.4) ;
			\draw [shift={(488,263.5)}, rotate = 358.7] [fill={rgb, 255:red, 0; green, 0; blue, 0 }  ][line width=0.08]  [draw opacity=0] (8.93,-4.29) -- (0,0) -- (8.93,4.29) -- cycle    ;
			\draw    (479,216.5) .. controls (495.54,214.14) and (508.95,219.37) .. (518.39,231.76) ;
			\draw [shift={(520,234)}, rotate = 235.84] [fill={rgb, 255:red, 0; green, 0; blue, 0 }  ][line width=0.08]  [draw opacity=0] (8.93,-4.29) -- (0,0) -- (8.93,4.29) -- cycle    ;
			\draw    (435,239.5) .. controls (444.03,226.68) and (451.7,219.72) .. (469.59,216.91) ;
			\draw [shift={(472.5,216.5)}, rotate = 172.87] [fill={rgb, 255:red, 0; green, 0; blue, 0 }  ][line width=0.08]  [draw opacity=0] (8.93,-4.29) -- (0,0) -- (8.93,4.29) -- cycle    ;
			\draw    (384,137.5) .. controls (429.34,177.96) and (391.5,236.96) .. (425.72,239.42) ;
			\draw [shift={(428.5,239.5)}, rotate = 179.27] [fill={rgb, 255:red, 0; green, 0; blue, 0 }  ][line width=0.08]  [draw opacity=0] (8.93,-4.29) -- (0,0) -- (8.93,4.29) -- cycle    ;
			\draw    (500,163.5) .. controls (475.39,167.84) and (462.89,164.27) .. (443.62,147.84) ;
			\draw [shift={(441.5,146)}, rotate = 41.28] [fill={rgb, 255:red, 0; green, 0; blue, 0 }  ][line width=0.08]  [draw opacity=0] (8.93,-4.29) -- (0,0) -- (8.93,4.29) -- cycle    ;
			\draw    (499,113.5) .. controls (510.1,120.74) and (519.8,134.49) .. (506.98,155.67) ;
			\draw [shift={(505.5,158)}, rotate = 303.69] [fill={rgb, 255:red, 0; green, 0; blue, 0 }  ][line width=0.08]  [draw opacity=0] (8.93,-4.29) -- (0,0) -- (8.93,4.29) -- cycle    ;
			\draw    (436,138.5) .. controls (449.99,117.75) and (465.38,112.37) .. (489.81,113.37) ;
			\draw [shift={(492.5,113.5)}, rotate = 183.3] [fill={rgb, 255:red, 0; green, 0; blue, 0 }  ][line width=0.08]  [draw opacity=0] (8.93,-4.29) -- (0,0) -- (8.93,4.29) -- cycle    ;
			\draw    (384,137.5) .. controls (405.6,131.26) and (403.7,140.23) .. (426.52,138.74) ;
			\draw [shift={(429.5,138.5)}, rotate = 174.51] [fill={rgb, 255:red, 0; green, 0; blue, 0 }  ][line width=0.08]  [draw opacity=0] (8.93,-4.29) -- (0,0) -- (8.93,4.29) -- cycle    ;
			\draw    (479,80) .. controls (474.02,79.15) and (470.25,78.27) .. (467.24,77.27) .. controls (464.39,76.31) and (451.85,74.22) .. (441.6,59.44) ;
			\draw [shift={(440,57)}, rotate = 58.3] [fill={rgb, 255:red, 0; green, 0; blue, 0 }  ][line width=0.08]  [draw opacity=0] (8.93,-4.29) -- (0,0) -- (8.93,4.29) -- cycle    ;
			\draw    (520.5,69) .. controls (511.95,77.55) and (505.21,83.39) .. (488.26,80.53) ;
			\draw [shift={(485.5,80)}, rotate = 11.89] [fill={rgb, 255:red, 0; green, 0; blue, 0 }  ][line width=0.08]  [draw opacity=0] (8.93,-4.29) -- (0,0) -- (8.93,4.29) -- cycle    ;
			\draw    (517,24) .. controls (533.54,30.62) and (536.23,47.5) .. (527.2,61.15) ;
			\draw [shift={(525.5,63.5)}, rotate = 308.16] [fill={rgb, 255:red, 0; green, 0; blue, 0 }  ][line width=0.08]  [draw opacity=0] (8.93,-4.29) -- (0,0) -- (8.93,4.29) -- cycle    ;
			\draw    (468.5,17.5) .. controls (481.24,2.94) and (498.95,10.74) .. (508.82,15.27) ;
			\draw [shift={(511.5,16.5)}, rotate = 203.96] [fill={rgb, 255:red, 0; green, 0; blue, 0 }  ][line width=0.08]  [draw opacity=0] (8.93,-4.29) -- (0,0) -- (8.93,4.29) -- cycle    ;
			\draw    (440,50.5) .. controls (439.52,38.56) and (442.69,28.45) .. (459.53,18.85) ;
			\draw [shift={(462,17.5)}, rotate = 152.24] [fill={rgb, 255:red, 0; green, 0; blue, 0 }  ][line width=0.08]  [draw opacity=0] (8.93,-4.29) -- (0,0) -- (8.93,4.29) -- cycle    ;
			\draw    (384,137.5) .. controls (423.2,108.1) and (394.69,81.58) .. (431.18,52.3) ;
			\draw [shift={(433.5,50.5)}, rotate = 143.13] [fill={rgb, 255:red, 0; green, 0; blue, 0 }  ][line width=0.08]  [draw opacity=0] (8.93,-4.29) -- (0,0) -- (8.93,4.29) -- cycle    ;
			\draw  [fill={rgb, 255:red, 255; green, 255; blue, 255 }  ,fill opacity=1 ] (38,137.5) .. controls (38,133.91) and (40.91,131) .. (44.5,131) .. controls (48.09,131) and (51,133.91) .. (51,137.5) .. controls (51,141.09) and (48.09,144) .. (44.5,144) .. controls (40.91,144) and (38,141.09) .. (38,137.5) -- cycle ;
			\draw    (51,137.5) -- (83.5,137.5) ;
			\draw [shift={(86.5,137.5)}, rotate = 180] [fill={rgb, 255:red, 0; green, 0; blue, 0 }  ][line width=0.08]  [draw opacity=0] (8.93,-4.29) -- (0,0) -- (8.93,4.29) -- cycle    ;
			\draw  [fill={rgb, 255:red, 255; green, 255; blue, 255 }  ,fill opacity=1 ] (86.5,137.5) .. controls (86.5,133.91) and (89.41,131) .. (93,131) .. controls (96.59,131) and (99.5,133.91) .. (99.5,137.5) .. controls (99.5,141.09) and (96.59,144) .. (93,144) .. controls (89.41,144) and (86.5,141.09) .. (86.5,137.5) -- cycle ;
			\draw  [fill={rgb, 255:red, 0; green, 0; blue, 0 }  ,fill opacity=1 ] (135,137.5) .. controls (135,133.91) and (137.91,131) .. (141.5,131) .. controls (145.09,131) and (148,133.91) .. (148,137.5) .. controls (148,141.09) and (145.09,144) .. (141.5,144) .. controls (137.91,144) and (135,141.09) .. (135,137.5) -- cycle ;
			\draw    (99.5,137.5) -- (132,137.5) ;
			\draw [shift={(135,137.5)}, rotate = 180] [fill={rgb, 255:red, 0; green, 0; blue, 0 }  ][line width=0.08]  [draw opacity=0] (8.93,-4.29) -- (0,0) -- (8.93,4.29) -- cycle    ;
			\draw    (148,137.5) -- (180.5,137.5) ;
			\draw [shift={(183.5,137.5)}, rotate = 180] [fill={rgb, 255:red, 0; green, 0; blue, 0 }  ][line width=0.08]  [draw opacity=0] (8.93,-4.29) -- (0,0) -- (8.93,4.29) -- cycle    ;
			\draw    (196.5,137.5) -- (229,137.5) ;
			\draw [shift={(232,137.5)}, rotate = 180] [fill={rgb, 255:red, 0; green, 0; blue, 0 }  ][line width=0.08]  [draw opacity=0] (8.93,-4.29) -- (0,0) -- (8.93,4.29) -- cycle    ;
			\draw  [fill={rgb, 255:red, 0; green, 0; blue, 0 }  ,fill opacity=1 ] (232,137.5) .. controls (232,133.91) and (234.91,131) .. (238.5,131) .. controls (242.09,131) and (245,133.91) .. (245,137.5) .. controls (245,141.09) and (242.09,144) .. (238.5,144) .. controls (234.91,144) and (232,141.09) .. (232,137.5) -- cycle ;
			\draw  [fill={rgb, 255:red, 255; green, 255; blue, 255 }  ,fill opacity=1 ] (183.5,137.5) .. controls (183.5,133.91) and (186.41,131) .. (190,131) .. controls (193.59,131) and (196.5,133.91) .. (196.5,137.5) .. controls (196.5,141.09) and (193.59,144) .. (190,144) .. controls (186.41,144) and (183.5,141.09) .. (183.5,137.5) -- cycle ;
			\draw  [fill={rgb, 255:red, 0; green, 0; blue, 0 }  ,fill opacity=1 ] (280.5,137.5) .. controls (280.5,133.91) and (283.41,131) .. (287,131) .. controls (290.59,131) and (293.5,133.91) .. (293.5,137.5) .. controls (293.5,141.09) and (290.59,144) .. (287,144) .. controls (283.41,144) and (280.5,141.09) .. (280.5,137.5) -- cycle ;
			\draw    (293.5,137.5) -- (326,137.5) ;
			\draw [shift={(329,137.5)}, rotate = 180] [fill={rgb, 255:red, 0; green, 0; blue, 0 }  ][line width=0.08]  [draw opacity=0] (8.93,-4.29) -- (0,0) -- (8.93,4.29) -- cycle    ;
			\draw    (245,137.5) -- (277.5,137.5) ;
			\draw [shift={(280.5,137.5)}, rotate = 180] [fill={rgb, 255:red, 0; green, 0; blue, 0 }  ][line width=0.08]  [draw opacity=0] (8.93,-4.29) -- (0,0) -- (8.93,4.29) -- cycle    ;
			\draw    (342,137.5) -- (374.5,137.5) ;
			\draw [shift={(377.5,137.5)}, rotate = 180] [fill={rgb, 255:red, 0; green, 0; blue, 0 }  ][line width=0.08]  [draw opacity=0] (8.93,-4.29) -- (0,0) -- (8.93,4.29) -- cycle    ;
			\draw  [fill={rgb, 255:red, 255; green, 255; blue, 255 }  ,fill opacity=1 ] (329,137.5) .. controls (329,133.91) and (331.91,131) .. (335.5,131) .. controls (339.09,131) and (342,133.91) .. (342,137.5) .. controls (342,141.09) and (339.09,144) .. (335.5,144) .. controls (331.91,144) and (329,141.09) .. (329,137.5) -- cycle ;
			\draw  [fill={rgb, 255:red, 255; green, 255; blue, 255 }  ,fill opacity=1 ] (377.5,137.5) .. controls (377.5,133.91) and (380.41,131) .. (384,131) .. controls (387.59,131) and (390.5,133.91) .. (390.5,137.5) .. controls (390.5,141.09) and (387.59,144) .. (384,144) .. controls (380.41,144) and (377.5,141.09) .. (377.5,137.5) -- cycle ;
			\draw  [fill={rgb, 255:red, 255; green, 255; blue, 255 }  ,fill opacity=1 ] (433.5,50.5) .. controls (433.5,46.91) and (436.41,44) .. (440,44) .. controls (443.59,44) and (446.5,46.91) .. (446.5,50.5) .. controls (446.5,54.09) and (443.59,57) .. (440,57) .. controls (436.41,57) and (433.5,54.09) .. (433.5,50.5) -- cycle ;
			\draw  [color={rgb, 255:red, 0; green, 0; blue, 0 }  ,draw opacity=1 ][fill={rgb, 255:red, 0; green, 0; blue, 0 }  ,fill opacity=1 ] (462,17.5) .. controls (462,13.91) and (464.91,11) .. (468.5,11) .. controls (472.09,11) and (475,13.91) .. (475,17.5) .. controls (475,21.09) and (472.09,24) .. (468.5,24) .. controls (464.91,24) and (462,21.09) .. (462,17.5) -- cycle ;
			\draw  [fill={rgb, 255:red, 255; green, 255; blue, 255 }  ,fill opacity=1 ] (510.5,22.5) .. controls (510.5,18.91) and (513.41,16) .. (517,16) .. controls (520.59,16) and (523.5,18.91) .. (523.5,22.5) .. controls (523.5,26.09) and (520.59,29) .. (517,29) .. controls (513.41,29) and (510.5,26.09) .. (510.5,22.5) -- cycle ;
			\draw  [fill={rgb, 255:red, 0; green, 0; blue, 0 }  ,fill opacity=1 ] (514,69) .. controls (514,65.41) and (516.91,62.5) .. (520.5,62.5) .. controls (524.09,62.5) and (527,65.41) .. (527,69) .. controls (527,72.59) and (524.09,75.5) .. (520.5,75.5) .. controls (516.91,75.5) and (514,72.59) .. (514,69) -- cycle ;
			\draw  [fill={rgb, 255:red, 255; green, 255; blue, 255 }  ,fill opacity=1 ] (472.5,80) .. controls (472.5,76.41) and (475.41,73.5) .. (479,73.5) .. controls (482.59,73.5) and (485.5,76.41) .. (485.5,80) .. controls (485.5,83.59) and (482.59,86.5) .. (479,86.5) .. controls (475.41,86.5) and (472.5,83.59) .. (472.5,80) -- cycle ;
			\draw  [fill={rgb, 255:red, 255; green, 255; blue, 255 }  ,fill opacity=1 ] (492.5,113.5) .. controls (492.5,109.91) and (495.41,107) .. (499,107) .. controls (502.59,107) and (505.5,109.91) .. (505.5,113.5) .. controls (505.5,117.09) and (502.59,120) .. (499,120) .. controls (495.41,120) and (492.5,117.09) .. (492.5,113.5) -- cycle ;
			\draw  [fill={rgb, 255:red, 0; green, 0; blue, 0 }  ,fill opacity=1 ] (429.5,138.5) .. controls (429.5,134.91) and (432.41,132) .. (436,132) .. controls (439.59,132) and (442.5,134.91) .. (442.5,138.5) .. controls (442.5,142.09) and (439.59,145) .. (436,145) .. controls (432.41,145) and (429.5,142.09) .. (429.5,138.5) -- cycle ;
			\draw  [fill={rgb, 255:red, 255; green, 255; blue, 255 }  ,fill opacity=1 ] (493.5,163.5) .. controls (493.5,159.91) and (496.41,157) .. (500,157) .. controls (503.59,157) and (506.5,159.91) .. (506.5,163.5) .. controls (506.5,167.09) and (503.59,170) .. (500,170) .. controls (496.41,170) and (493.5,167.09) .. (493.5,163.5) -- cycle ;
			\draw  [fill={rgb, 255:red, 255; green, 255; blue, 255 }  ,fill opacity=1 ] (472.5,216.5) .. controls (472.5,212.91) and (475.41,210) .. (479,210) .. controls (482.59,210) and (485.5,212.91) .. (485.5,216.5) .. controls (485.5,220.09) and (482.59,223) .. (479,223) .. controls (475.41,223) and (472.5,220.09) .. (472.5,216.5) -- cycle ;
			\draw  [fill={rgb, 255:red, 255; green, 255; blue, 255 }  ,fill opacity=1 ] (428.5,239.5) .. controls (428.5,235.91) and (431.41,233) .. (435,233) .. controls (438.59,233) and (441.5,235.91) .. (441.5,239.5) .. controls (441.5,243.09) and (438.59,246) .. (435,246) .. controls (431.41,246) and (428.5,243.09) .. (428.5,239.5) -- cycle ;
			\draw  [fill={rgb, 255:red, 0; green, 0; blue, 0 }  ,fill opacity=1 ] (518.5,239.5) .. controls (518.5,235.91) and (521.41,233) .. (525,233) .. controls (528.59,233) and (531.5,235.91) .. (531.5,239.5) .. controls (531.5,243.09) and (528.59,246) .. (525,246) .. controls (521.41,246) and (518.5,243.09) .. (518.5,239.5) -- cycle ;
			\draw  [fill={rgb, 255:red, 0; green, 0; blue, 0 }  ,fill opacity=1 ] (475,263.5) .. controls (475,259.91) and (477.91,257) .. (481.5,257) .. controls (485.09,257) and (488,259.91) .. (488,263.5) .. controls (488,267.09) and (485.09,270) .. (481.5,270) .. controls (477.91,270) and (475,267.09) .. (475,263.5) -- cycle ;
			
			\draw (472,174) node [anchor=north west][inner sep=0.75pt]   [align=left] {{\Huge $\vdots$}};
			\draw (562,33.4) node [anchor=north west][inner sep=0.75pt]    {$p_{1}$};
			\draw (562,131.4) node [anchor=north west][inner sep=0.75pt]    {$p_{2}$};
			\draw (562,230.4) node [anchor=north west][inner sep=0.75pt]    {$p_{r}$};
			\draw (42,151.4) node [anchor=north west][inner sep=0.75pt]    {$\underbrace{\ \ \ \ \ \ \ \ \ \ \ \ \ \ \ \ \ \ \ \ \ \ \ \ \ \ \ \ \ \ \ \ \ \ \ \ \ \ \ \ \ \ \ \ \ \ \ \ \ \ \ \ \ \ \ \ }_{}$};
			\draw (160,170.4) node [anchor=north west][inner sep=0.75pt]    {$\ell ( M) = O(M^2)$};
			\draw (424,276.4) node [anchor=north west][inner sep=0.75pt]    {$\underbrace{\ \ \ \ \ \ \ \ \ \ \ \ \ \ \ \ \ \ }_{}$};
			\draw (454,296.4) node [anchor=north west][inner sep=0.75pt]    {$\leqslant M$};

		\end{tikzpicture}
		\caption{\label{fig:chrobak}An automaton in Chrobak normal form. Black states are final states.}
	\end{figure}
	
	Let us make the following observation, which is not crucial in our proof but which completely characterizes the set of paths that admits a constant-size certification, and which is direct from Theorem~\ref{thm:chrobak}.
	
	\begin{observation}
		\label{obs:periodic}
		Let $S \subseteq \N$. Certifying that a path has its length in~$S$ can be done with constant-size certificates if and only if $S$ is eventually periodic.
	\end{observation}
	
	Indeed, if $S$ is eventually periodic, then there are integers $N, i_1, \ldots, i_k, m_1, \ldots, m_k$ such that the set $S$ is equal to $\bigcup_{j=1}^{k} \{n \in \N \; | n \geqslant N \;\; \text{ and } \;\; n \equiv i_j \mod m_j \}$ and for every set of this union, certifying that the length of a path belongs to it can be done with constant size-certificates. Conversely, if the certification for the belonging to~$S$ can be done with constant-size certificates, then there exists a finite automaton~$\A$ (defined as in Section~\ref{sec:gap-path}) that recognizes the set of lengths in~$S$, and its Chrobak normal form given by Theorem~\ref{thm:chrobak} accepts the same set of integers which is eventually periodic.

	To prove Theorem~\ref{thm:gap-paths-chrobak}, we will also use the following result proved by Landau:
	
	\begin{theorem}[Landau]
		\label{thm:landau}
		For $n \in \N$, let us denote by~$g(n)$ the largest least common multiple of any set of positive integers whose sum is equal to~$n$ ($g$ is called the Landau's function).
		Then, we have:
		$$\ln(g(n)) \sim \sqrt{n \ln(n))}$$
	\end{theorem}
	
	We use the notations introduced in Section~\ref{sec:gap-path}.
	
	\begin{proof}[Proof of Theorem~\ref{thm:gap-paths-chrobak}.]
		
		The automaton~$\A_k$ has $M_k := 2^{2k} + 2$ states and recognizes $S_k$. Then, by Theorem~\ref{thm:chrobak}, the set $S_k$ is also recognized by the Chrobak normal form of $\A_k$, which is an automaton that consists in a path of at most $\ell(M_k) = O(M_k^2)$ states, with a single nondeterministic choice to cycles whose sum of lengths is at most $M_k$.
		Let $t > 0$ be such that $\ell(M) \leqslant t M^2$ for every integer $M \geqslant 1$.
		
		\begin{claim}
			\label{claim:proof_chrobak_bound}
			There exists an integer $k_0 \in \N$ such that, for all $k \geqslant k_0$, we have: $$2^{2^{ck}} > g(2^{2k+1}) + \ell(M_k)$$
			
		\end{claim}
		
		\begin{proof}
			Since $M_k = 2^{2k}+2$, we have $\ell(M_k)=o(2^{2^{ck}})$. Moreover, by Theorem~\ref{thm:landau}, we have $\ln(g(n)) = \sqrt{n \ln n}(1+o(1))$, so $g(2^{2k+1}) = 2^{2^{k+1/2}\sqrt{(2k+1) \ln 2}(1+o(1))}$. Since $c>1$, we get $g(2^{2k+1})=o(2^{2^{ck}})$. The result follows.
		\end{proof}
		
		Now, assume by contradiction that $\P$ cannot be certified with constant size certificates, and let $k_0 \in \N$ be the integer given by Claim~\ref{claim:proof_chrobak_bound}. Since $\P$ cannot be certified with constant-size certificates, the set $X \subseteq \N$ containing all the integers $k \in \N$ such that $S_k \nsubseteq \bigcup_{i \leqslant k-1} S_i$ is infinite. For every~$k \in X$, let $n_k$ be the smallest integer in $S_k \setminus \bigcup_{i \leqslant k-1} S_i$.
		
		Let us fix an integer $k \in X$ such that $k \geqslant k_0$ and $n_k \geqslant N$ (note that such an integer exists because $X$ is infinite and for all distinct $k, k' \in X$, we have $n_k \neq n_{k'}$).
		Since a path of length $n_k$ is accepted with certificates of size $s(n_k)$, we have $n_k \in S_{s(n_k)}$.
		Since $n_k \in S_k \setminus \bigcup_{i \leqslant k-1} S_i$, we get $s(n_k) \geqslant k$.
		Now, by definition of $s(n_k)$, we obtain $n_k \geqslant 2^{2^{ck}}$. The run which accepts $n_k$ in the automaton $\A_k$ ends in some accepting state which is inside a loop in the Chrobak normal form of~$\A_k$, because $n_k \geqslant 2^{2^{ck}} > \ell(M_k)$ (the last inequality holds since $k \geqslant k_0$).

		Now, us consider the Chrobak normal forms of $\A_1, \ldots, \A_k$, and let $P_k$ be the set of lengths of all the loops in these Chrobak normal forms. Let $p_k$ be the least common multiple of all the integers in~$P_k$. For every $i \in \{1, \ldots, k\}$, we know that the sum of the lengths of the loops in the Chrobak normal form of $\A_i$ is at most $M_i$. Thus, the sum of the integers in $P_k$ is at most $\sum_{i=1}^k M_i$, and we have:
		$$\sum_{i=1}^kM_i = \sum_{i=1}^k 2^{2i} + 2k \leqslant 2^{2k+1}$$
		
		Since the Landau function $g$ is increasing, we have $p_k \leqslant  g(2^{2k+1})$. Thus, we have $n_k - p_k \geqslant 2^{2^{ck}}-g(2^{2k+1}) > \ell(M_k)$ (the last inequality holds because $k \geqslant k_0$).
		Since the length of the loop in the accepting run of $n_k$ in the Chrobak normal form of $\A_k$ divides~$p_k$, and since $n_k-p_k \geqslant \ell(M_k)$, there is a run of length $n_k - p_k$ in the Chrobak normal form of $\A_k$ which ends in the same accepting state. So \mbox{$n_k - p_k \in S_k$}. By assumption on~$n_k$ (which is the smallest integer in $S_k \setminus \bigcup_{i \leqslant k-1} S_i$), we have \mbox{$n_k - p M_k! \in \bigcup_{i \leqslant k-1} S_i$}. Let $j \leqslant k-1$ be 	such that $n_k - p_k \in S_j$. Since $n_k - p_k \geqslant \ell(M_j)$, the run of $n_k - p_k$ in the Chrobak normal form of $\A_j$ ends in some accepting state which is inside a loop. The length of this loop divides $p_k$. Hence, there exists a run of $n_k$ in the Chrobak normal form of $\A_j$ that ends in the same accepting state as $n_k - p_k$, so $n_k \in S_j$, which is a contradiction and concludes the proof.
	\end{proof}
	
	\section{No gap above $\log n$ in general graphs with identifiers}
	\label{sec:no-gap-above-log-n}
	
	We establish that there is no gap in the local certification complexity in general graphs with identifiers in the $\Omega(\log n)$ regime.
	The proof follows the same route as the proof of the $\Theta(n^2)$ bound for non-trivial automorphism in~\cite{GoosS16} . 
	
	\ThmFolkloreNoGap*
	
	Since the proof is a rather direct adaptation of the one of~\cite{GoosS16} and is not central in this paper, we use a more sketchy style than for our other proofs.
	
	\begin{proof}
		Fix some function $f$ in $\Omega(\log n)$ and $O(n^2)$ consider the following language. A graph is in the language if it is made of two copies of a graph $H$ on $\sqrt{f(n)}$ nodes,where an arbitrary node is linked to its copy by a path of length $n-2\sqrt{f(n)}$. A certification for this language is the following. 
		On a correct instance, every node is given the following pieces of information.  
		\begin{enumerate}
			\item Its part of the certification of the size $n$ of the graph, via a spanning tree (see \cite{Feuilloley21}). 
			\item The adjacency matrix of $H$.
			\item The identifier assignment restricted to the copies of $H$.
			\item Its parts of two spanning trees, pointing to the two nodes that belong both to a copy of $H$ and to the path.  
		\end{enumerate}

		Every node checks the following. 
		\begin{enumerate}
			\item Item 1 and 4 above are consistent (again see \cite{Feuilloley21}).
			\item If its identifier appears in the list of Item 3, it checks that its neighborhood in the graph is consistent with the neighborhood in $H$ as described by the certificates; except if it is the root of a spanning tree of Item 4, in which case it should have exactly one additional neighbor. 
			\item If its identifier does not appear in the certificates it should have degree 2. 
		\end{enumerate}
		
		The correctness is follows from~\cite{GoosS16}, where the same scheme is used except that the central path has constant length (and therefore the identifiers of the nodes on the path can be given to all nodes without overhead). 
		
		The spanning tree certifying the number of nodes can be encoded in $O(\log n)$ bits, the adjacency matrices use $O\left(\left(\sqrt{f(n)}\right)^2\right)=O(f(n))$ bits, and the identifier assignment uses $2\sqrt{f(n)}\log n$ bits. Hence in total $O(f(n))$ bits in the regime of the theorem. 
		
		Again the lower bound is very similar to the one of~\cite{GoosS16}. Basically, if we were to use $o(f(n))$ bits, by pigeon-hole principle, there would be two different correct instances of the same size $n$, for which the same certificates would be used on one edge of the path.
		Then we could consider the graph where we take the right part from one instance and the left part from the other (with their accepting certificates), gluing on the edge with identical certificates. This new graph would be accepted, but it is not in our language since the graphs at the end of the path are different. A contradiction.
	\end{proof}

	\section{No gap in $n$ in caterpillars with radius 1}
	\label{sec:no-gap-in-n-caterpillars}

	The goal of this section is to prove the following: 
	
	\ThmCaterpillarNoGapinNRadiusOne*
	
	To prove this result, we will need the following technical lemma, which basically says that for the function $f$ of the theorem, we can define a sequence of integers such that $f$ applied to the sums of the $d$ first elements is close to $\log d$.  
	
	\begin{lemma}
		\label{lem:sequence a_i}
		Let $f : \N \to \mathbb{R}$ be a non-decreasing function such that $\lim_{n \rightarrow +\infty}f(n) = +\infty$, $f(2)=1$ and for all integers $1 \leqslant s \leqslant t$ we have $f(t)-f(s) \leqslant \frac{\log t - \log s}{2}$.
		Then, there exists an unbounded sequence of positive integers $(a_i)_{i \geqslant 1}$ such that $a_1=2$ and:
		
		\begin{equation}
			\label{eq:lemma a_i}
			\forall d \geqslant 1, \;\; \log d \leqslant f\left(\sum_{i=1}^d a_i\right) \leqslant \log d + 1. 
		\end{equation}
	\end{lemma}
	
	\begin{proof}
		Let us construct the sequence $(a_i)_{i\geqslant 1}$ iteratively: first, we set $a_1 = 2$. Now, let $d \geqslant 2$ and assume that we already defined $a_1, \ldots, a_{d-1}$. We define $a_d$ as being the smallest positive integer such that $\log d \leqslant f\left(\sum_{i=1}^{d} a_i\right)$, which exists because $\lim_{n \rightarrow +\infty}f(n) = +\infty$. Let us prove that this sequence satisfies both inequalities in~(\ref{eq:lemma a_i}).
		
		The first inequality of~(\ref{eq:lemma a_i}) holds by definition of $(a_i)_{i \geqslant 1}$. Let us prove the second inequality by induction on $d$. It trivially holds for $d=1$ because $f(2)=1$. Now, assume that $d \geqslant 2$, and that it holds for $d-1$. There are two cases:
		
		\begin{itemize}
			\item If $a_d = 1$, then we have:
			\begin{align*}
				f\left(\sum_{i=1}^d a_i\right)
				& \leqslant f\left(\sum_{i=1}^{d-1} a_i\right) + \log \left(\sum_{i=1}^d a_i\right) - \log \left(\sum_{i=1}^{d-1} a_i\right) \\
				& \leqslant \log(d-1) + 1 + \log \left(1 + \frac{1}{\sum_{i=1}^{d-1}a_i}\right) \quad \text{\footnotesize{(by~(\ref{eq:lemma a_i}) with the induction hypothesis)}}\\
				& \leqslant \log (d-1) + 1 + \log\left(1 + \frac{1}{d-1}\right) \\
				& \leqslant \log d + 1
			\end{align*}
			
			\item If $a_d \geqslant 2$, by definition of~$a_d$ we know that $f\left(\sum_{i=1}^{d}a_i - 1\right) < \log d$. In this case, we have:
			\begin{align*}
				f\left(\sum_{i=1}^d a_i\right)
				& \leqslant f\left(\sum_{i=1}^{d} a_i - 1\right) + \log \left(\sum_{i=1}^d a_i\right) - \log \left(\sum_{i=1}^{d} a_i - 1\right) \\
				& \leqslant \log d + \log \left(1 + \frac{1}{\sum_{i=1}^{d}a_i - 1}\right) \\
				& \leqslant \log d + 1
			\end{align*}
		\end{itemize}
		
		Thus, the sequence $(a_i)_{i \geqslant 1}$ satisfies indeed both inequalities of~(\ref{eq:lemma a_i}). Let us prove that this sequence is unbounded. By contradiction, assume that there exists $M > 0$ such that $a_i \leqslant M$ for every $i \geqslant 1$. By plugging $s=2$ in the assumption made on~$f$, we have $f(t) \leqslant \frac{\log t}{2}$ for all integer $t \geqslant 2$, and since $f$ is non-decreasing, we get for all $d \geqslant 1$:
		$$\log d \leqslant f\left(\sum_{i=1}^d a_i\right) \leqslant f(Md) \leqslant \frac{\log (Md)}{2}$$
		which finally gives us $d \leqslant \sqrt{Md}$ for every integer~$d \geqslant 1$ and this is a contradiction. Thus, $(a_i)_{i \geqslant 1}$ is unbounded.
	\end{proof}

	We are now able to prove Theorem~\ref{thm:caterpillars-no-gap-in-n}.
	
	\begin{proof}[Proof of Theorem~\ref{thm:caterpillars-no-gap-in-n}]
		Without loss of generality, we can assume that $f(2) = 1$ and for all integers $1 \leqslant s \leqslant t$, we have $f(t)-f(s) \leqslant (\log t - \log s)/2$ (else, we can consider the function~$g$ such that $g(n) := (f(n) - f(2))/2+1$ which satisfies it). Let $(a_i)_{i \geqslant 1}$ be the sequence given by Lemma~\ref{lem:sequence a_i}.
		
		
		

		Let us define the following property~$\P$ on caterpillars. The caterpillars satisfying~$\P$ are such that the central path~$P$ has length~$d$, we have $a_d \geqslant 2$, and if we denote its vertices by $u_1, \ldots, u_d$, then for every $i \in \{1, \ldots, d\}$, there are $a_i-1$ degree-1 vertices attached to~$u_i$ (note that the condition $a_d \geqslant 2$ ensures that $u_d$ has degree at least~$2$ so it belongs to the central path). See Figure~\ref{fig:caterpillar_no_gap_n} for an example. (This type of construction has been used before, for example in \cite{BousquetFP24}.)

		\begin{figure}[h!]
			\centering
			
			\begin{tikzpicture}[x=0.75pt,y=0.75pt,yscale=-1,xscale=1]
				
				\draw    (382.5,125.25) -- (342.5,103.25) ;
				\draw    (382.5,125.25) -- (422.5,103.25) ;
				\draw    (382.5,125.25) -- (382.5,103.25) ;
				\draw    (382.5,125.25) -- (402.5,103.25) ;
				\draw    (382.5,125.25) -- (362.5,103.25) ;
				\draw    (189.5,125.25) -- (189.5,103.25) ;
				\draw    (189.5,125.25) -- (202.5,103.25) ;
				\draw    (189.5,125.25) -- (176.5,103.25) ;
				\draw    (141.25,125.25) -- (141.25,103.25) ;
				\draw    (141.25,125.25) -- (189.5,125.25) ;
				\draw    (189.5,125.25) -- (237.75,125.25) ;
				\draw    (237.75,125.25) -- (286,125.25) ;
				\draw    (286,125.25) -- (334.25,125.25) ;
				\draw    (334.25,125.25) -- (382.5,125.25) ;
				\draw  [fill={rgb, 255:red, 255; green, 255; blue, 255 }  ,fill opacity=1 ] (136,125.25) .. controls (136,122.35) and (138.35,120) .. (141.25,120) .. controls (144.15,120) and (146.5,122.35) .. (146.5,125.25) .. controls (146.5,128.15) and (144.15,130.5) .. (141.25,130.5) .. controls (138.35,130.5) and (136,128.15) .. (136,125.25) -- cycle ;
				\draw  [fill={rgb, 255:red, 255; green, 255; blue, 255 }  ,fill opacity=1 ] (184.25,125.25) .. controls (184.25,122.35) and (186.6,120) .. (189.5,120) .. controls (192.4,120) and (194.75,122.35) .. (194.75,125.25) .. controls (194.75,128.15) and (192.4,130.5) .. (189.5,130.5) .. controls (186.6,130.5) and (184.25,128.15) .. (184.25,125.25) -- cycle ;
				\draw  [fill={rgb, 255:red, 255; green, 255; blue, 255 }  ,fill opacity=1 ] (232.5,125.25) .. controls (232.5,122.35) and (234.85,120) .. (237.75,120) .. controls (240.65,120) and (243,122.35) .. (243,125.25) .. controls (243,128.15) and (240.65,130.5) .. (237.75,130.5) .. controls (234.85,130.5) and (232.5,128.15) .. (232.5,125.25) -- cycle ;
				\draw  [fill={rgb, 255:red, 255; green, 255; blue, 255 }  ,fill opacity=1 ] (280.75,125.25) .. controls (280.75,122.35) and (283.1,120) .. (286,120) .. controls (288.9,120) and (291.25,122.35) .. (291.25,125.25) .. controls (291.25,128.15) and (288.9,130.5) .. (286,130.5) .. controls (283.1,130.5) and (280.75,128.15) .. (280.75,125.25) -- cycle ;
				\draw  [fill={rgb, 255:red, 255; green, 255; blue, 255 }  ,fill opacity=1 ] (329,125.25) .. controls (329,122.35) and (331.35,120) .. (334.25,120) .. controls (337.15,120) and (339.5,122.35) .. (339.5,125.25) .. controls (339.5,128.15) and (337.15,130.5) .. (334.25,130.5) .. controls (331.35,130.5) and (329,128.15) .. (329,125.25) -- cycle ;
				\draw  [fill={rgb, 255:red, 255; green, 255; blue, 255 }  ,fill opacity=1 ] (377.25,125.25) .. controls (377.25,122.35) and (379.6,120) .. (382.5,120) .. controls (385.4,120) and (387.75,122.35) .. (387.75,125.25) .. controls (387.75,128.15) and (385.4,130.5) .. (382.5,130.5) .. controls (379.6,130.5) and (377.25,128.15) .. (377.25,125.25) -- cycle ;
				\draw  [fill={rgb, 255:red, 255; green, 255; blue, 255 }  ,fill opacity=1 ] (136,103.25) .. controls (136,100.35) and (138.35,98) .. (141.25,98) .. controls (144.15,98) and (146.5,100.35) .. (146.5,103.25) .. controls (146.5,106.15) and (144.15,108.5) .. (141.25,108.5) .. controls (138.35,108.5) and (136,106.15) .. (136,103.25) -- cycle ;
				\draw  [fill={rgb, 255:red, 255; green, 255; blue, 255 }  ,fill opacity=1 ] (171.25,103.25) .. controls (171.25,100.35) and (173.6,98) .. (176.5,98) .. controls (179.4,98) and (181.75,100.35) .. (181.75,103.25) .. controls (181.75,106.15) and (179.4,108.5) .. (176.5,108.5) .. controls (173.6,108.5) and (171.25,106.15) .. (171.25,103.25) -- cycle ;
				\draw  [fill={rgb, 255:red, 255; green, 255; blue, 255 }  ,fill opacity=1 ] (184.25,103.25) .. controls (184.25,100.35) and (186.6,98) .. (189.5,98) .. controls (192.4,98) and (194.75,100.35) .. (194.75,103.25) .. controls (194.75,106.15) and (192.4,108.5) .. (189.5,108.5) .. controls (186.6,108.5) and (184.25,106.15) .. (184.25,103.25) -- cycle ;
				\draw  [fill={rgb, 255:red, 255; green, 255; blue, 255 }  ,fill opacity=1 ] (197.25,103.25) .. controls (197.25,100.35) and (199.6,98) .. (202.5,98) .. controls (205.4,98) and (207.75,100.35) .. (207.75,103.25) .. controls (207.75,106.15) and (205.4,108.5) .. (202.5,108.5) .. controls (199.6,108.5) and (197.25,106.15) .. (197.25,103.25) -- cycle ;
				\draw  [fill={rgb, 255:red, 255; green, 255; blue, 255 }  ,fill opacity=1 ] (337.25,103.25) .. controls (337.25,100.35) and (339.6,98) .. (342.5,98) .. controls (345.4,98) and (347.75,100.35) .. (347.75,103.25) .. controls (347.75,106.15) and (345.4,108.5) .. (342.5,108.5) .. controls (339.6,108.5) and (337.25,106.15) .. (337.25,103.25) -- cycle ;
				\draw  [fill={rgb, 255:red, 255; green, 255; blue, 255 }  ,fill opacity=1 ] (357.25,103.25) .. controls (357.25,100.35) and (359.6,98) .. (362.5,98) .. controls (365.4,98) and (367.75,100.35) .. (367.75,103.25) .. controls (367.75,106.15) and (365.4,108.5) .. (362.5,108.5) .. controls (359.6,108.5) and (357.25,106.15) .. (357.25,103.25) -- cycle ;
				\draw  [fill={rgb, 255:red, 255; green, 255; blue, 255 }  ,fill opacity=1 ] (377.25,103.25) .. controls (377.25,100.35) and (379.6,98) .. (382.5,98) .. controls (385.4,98) and (387.75,100.35) .. (387.75,103.25) .. controls (387.75,106.15) and (385.4,108.5) .. (382.5,108.5) .. controls (379.6,108.5) and (377.25,106.15) .. (377.25,103.25) -- cycle ;
				\draw  [fill={rgb, 255:red, 255; green, 255; blue, 255 }  ,fill opacity=1 ] (397.25,103.25) .. controls (397.25,100.35) and (399.6,98) .. (402.5,98) .. controls (405.4,98) and (407.75,100.35) .. (407.75,103.25) .. controls (407.75,106.15) and (405.4,108.5) .. (402.5,108.5) .. controls (399.6,108.5) and (397.25,106.15) .. (397.25,103.25) -- cycle ;
				\draw  [fill={rgb, 255:red, 255; green, 255; blue, 255 }  ,fill opacity=1 ] (417.25,103.25) .. controls (417.25,100.35) and (419.6,98) .. (422.5,98) .. controls (425.4,98) and (427.75,100.35) .. (427.75,103.25) .. controls (427.75,106.15) and (425.4,108.5) .. (422.5,108.5) .. controls (419.6,108.5) and (417.25,106.15) .. (417.25,103.25) -- cycle ;
				
				\draw (253,105.4) node [anchor=north west][inner sep=0.75pt]    {${\displaystyle \dotsc }$};
				\draw (122,77.4) node [anchor=north west][inner sep=0.75pt]    {$a_{1}-1$};
				\draw (170,77.4) node [anchor=north west][inner sep=0.75pt]    {$a_{2}-1$};
				\draw (363,77.4) node [anchor=north west][inner sep=0.75pt]    {$a_{d}-1$};

				\draw    (142,149) -- (379,149) ;
				\draw [shift={(381,149)}, rotate = 180] [color={rgb, 255:red, 0; green, 0; blue, 0 }  ][line width=0.75]    (10.93,-3.29) .. controls (6.95,-1.4) and (3.31,-0.3) .. (0,0) .. controls (3.31,0.3) and (6.95,1.4) .. (10.93,3.29)   ;
				\draw [shift={(140,149)}, rotate = 0] [color={rgb, 255:red, 0; green, 0; blue, 0 }  ][line width=0.75]    (10.93,-3.29) .. controls (6.95,-1.4) and (3.31,-0.3) .. (0,0) .. controls (3.31,0.3) and (6.95,1.4) .. (10.93,3.29)   ;
				\draw (253,156.4) node [anchor=north west][inner sep=0.75pt]    {$d$};
				
			\end{tikzpicture}

			\caption{A caterpillar satisfying~$\P$.}
			\label{fig:caterpillar_no_gap_n}
		\end{figure}
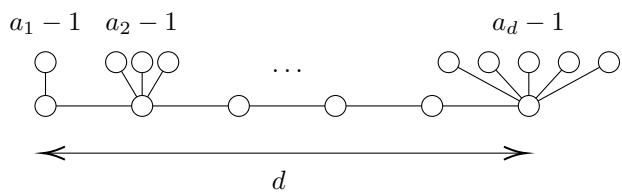
		
		First, let us prove that $\P$ can be certified with certificates of size $O(f(n))$. The certificates given by the prover to the vertices are the following. The prover gives a special certificate to every vertex of degree~$1$ (note that, on correct instances, a vertex belongs to~$P$ if and only if it has degree at least~$2$). Then, for each $i \in \{1, \ldots, d\}$, the prover writes~$i$ in binary in the certificate of~$u_i$. This uses $O(\log d)$ bits. Since the total number of vertices~$n$ is equal to $\sum_{i=1}^{d} a_i$, the size of the certificates is~$O(f(n))$ because $\log d \leqslant f\left(\sum_{i=1}^{d}a_i\right)$.
		The verification of the vertices consists simply in checking that they have the special certificate if and only if they have degree~$1$, and otherwise they check that their certificate is consistent with the certificates of their neighbors of degree at least~$2$ (for instance, if a vertex has certificate~$i$ and has two neighbors of degree at least~$2$ with certificates~$j$ and~$\ell$, it checks that~$\{j,\ell\} = \{i-1,i+1\}$). Finally, every vertex of degree at least~$2$ with certificate~$i$ checks that it has~$a_i-1$ neighbors with the certificate of degree~$1$ vertices.

		Let us now prove that it is not possible to certify~$\P$ with certificates of size~$o(f(n))$. Assume by contradiction that certificates of size~$o(f(n))$ are sufficient.
		For every $d \geqslant 1$ such that $a_d \geqslant 2$ (since $(a_i)_{i \in \N}$ is unbounded, there exists infinitely many such integers~$d$), let us denote by~$G_d$ the caterpillar satisfying~$\P$ having its central path of length~$d$, by $u_1, \ldots, u_d$ the vertices of its central path, and by $n_d := \sum_{i=1}^d a_i$ its number of vertices.
		By inequality~(\ref{eq:lemma a_i}), we have $f(n_d)=\Theta(\log d)$. Thus, the size of the certificates that makes all the vertices of~$G_d$ accept is $o(\log d)$. So there exists $d \geqslant 1$ such that $a_d \geqslant 2$ and the number of different pairs of certificates is strictly smaller than $d-1$.
		
		Now, let us fix such an integer $d$, let $M = \max_{1 \leqslant i \leqslant d} a_i$, and finally let $d' \in \N$ be such that $a_{d'} > M+1$ (such an integer $d'$ exists because $(a_i)_{i \in \N}$ is unbounded). Consider the $(d-1)$ pairs of consecutive certificates $c(u_1), \ldots, c(u_d)$ attributed by the prover to $u_1, \ldots, u_d$ in~$G_d$ in an assignment of certificates such that every vertex accepts. By the pigeonhole principle, there exists $1 \leqslant k < \ell \leqslant d$ such that $(c(u_k),c(u_{k+1})) = (c(u_{\ell}), c(u_{\ell+1}))$. For every $t \in \N$, we can copy $t$ times the part between $u_k$ and $u_\ell$ (with the same certificates) to obtain a new caterpillar in which every vertex accepts because it has the same view as a vertex in~$G_d$ which accepts. See Figure~\ref{fig:pumping_caterpillar} for an example in which $t=2$. In particular, by choosing $t$ large enough, the caterpillar we obtain has its central path of length at least $d'$, and is accepted with these certificates. However, the vertex $u_{d'}$ should have~$a_{d'}-1$ neighbors of degree~1, but it has most $M$ such neighbors (by definition of~$M$) and $M < a_{d'}-1$. This is a contradiction and concludes the proof. 
		\qedhere
		
		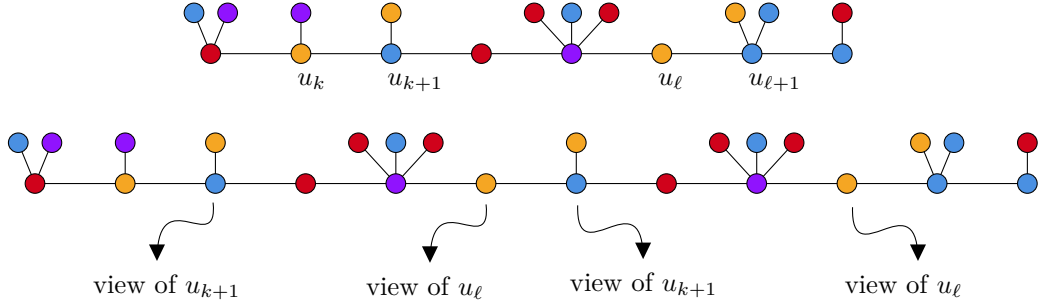
\begin{figure}[h!]
			\centering\begin{tikzpicture}[x=0.7pt,y=0.7pt,yscale=-1,xscale=1]
				
				\draw    (502.5,196.25) -- (550.75,196.25) ;
				\draw    (406,196.25) -- (454.25,196.25) ;
				\draw    (357.75,196.25) -- (357.75,174.25) ;
				\draw    (500,126.25) -- (500,104.25) ;
				\draw    (451.75,126.25) -- (500,126.25) ;
				\draw    (258.75,126.25) -- (258.75,104.25) ;
				\draw    (403.5,126.25) -- (451.75,126.25) ;
				\draw    (355.25,126.25) -- (355.25,104.25) ;
				\draw    (355.25,126.25) -- (375.25,104.25) ;
				\draw    (355.25,126.25) -- (335.25,104.25) ;
				\draw    (210.5,126.25) -- (210.5,104.25) ;
				\draw    (162.3,126.25) -- (171.3,104.25) ;
				\draw    (162.25,126.25) -- (153.3,104.25) ;
				\draw    (162.25,126.25) -- (210.5,126.25) ;
				\draw    (210.5,126.25) -- (258.75,126.25) ;
				\draw    (258.75,126.25) -- (307,126.25) ;
				\draw    (307,126.25) -- (355.25,126.25) ;
				\draw    (355.25,126.25) -- (403.5,126.25) ;
				\draw  [fill={rgb, 255:red, 208; green, 2; blue, 27 }  ,fill opacity=1 ] (157,126.25) .. controls (157,123.35) and (159.35,121) .. (162.25,121) .. controls (165.15,121) and (167.5,123.35) .. (167.5,126.25) .. controls (167.5,129.15) and (165.15,131.5) .. (162.25,131.5) .. controls (159.35,131.5) and (157,129.15) .. (157,126.25) -- cycle ;
				\draw  [fill={rgb, 255:red, 245; green, 166; blue, 35 }  ,fill opacity=1 ] (205.25,126.25) .. controls (205.25,123.35) and (207.6,121) .. (210.5,121) .. controls (213.4,121) and (215.75,123.35) .. (215.75,126.25) .. controls (215.75,129.15) and (213.4,131.5) .. (210.5,131.5) .. controls (207.6,131.5) and (205.25,129.15) .. (205.25,126.25) -- cycle ;
				\draw  [fill={rgb, 255:red, 74; green, 144; blue, 226 }  ,fill opacity=1 ] (253.5,126.25) .. controls (253.5,123.35) and (255.85,121) .. (258.75,121) .. controls (261.65,121) and (264,123.35) .. (264,126.25) .. controls (264,129.15) and (261.65,131.5) .. (258.75,131.5) .. controls (255.85,131.5) and (253.5,129.15) .. (253.5,126.25) -- cycle ;
				\draw  [fill={rgb, 255:red, 208; green, 2; blue, 27 }  ,fill opacity=1 ] (301.75,126.25) .. controls (301.75,123.35) and (304.1,121) .. (307,121) .. controls (309.9,121) and (312.25,123.35) .. (312.25,126.25) .. controls (312.25,129.15) and (309.9,131.5) .. (307,131.5) .. controls (304.1,131.5) and (301.75,129.15) .. (301.75,126.25) -- cycle ;
				\draw  [fill={rgb, 255:red, 144; green, 19; blue, 254 }  ,fill opacity=1 ] (350,126.25) .. controls (350,123.35) and (352.35,121) .. (355.25,121) .. controls (358.15,121) and (360.5,123.35) .. (360.5,126.25) .. controls (360.5,129.15) and (358.15,131.5) .. (355.25,131.5) .. controls (352.35,131.5) and (350,129.15) .. (350,126.25) -- cycle ;
				\draw  [fill={rgb, 255:red, 245; green, 166; blue, 35 }  ,fill opacity=1 ] (398.25,126.25) .. controls (398.25,123.35) and (400.6,121) .. (403.5,121) .. controls (406.4,121) and (408.75,123.35) .. (408.75,126.25) .. controls (408.75,129.15) and (406.4,131.5) .. (403.5,131.5) .. controls (400.6,131.5) and (398.25,129.15) .. (398.25,126.25) -- cycle ;
				\draw  [fill={rgb, 255:red, 74; green, 144; blue, 226 }  ,fill opacity=1 ] (148.05,104.25) .. controls (148.05,101.35) and (150.4,99) .. (153.3,99) .. controls (156.2,99) and (158.55,101.35) .. (158.55,104.25) .. controls (158.55,107.15) and (156.2,109.5) .. (153.3,109.5) .. controls (150.4,109.5) and (148.05,107.15) .. (148.05,104.25) -- cycle ;
				\draw  [fill={rgb, 255:red, 144; green, 19; blue, 254 }  ,fill opacity=1 ] (166.05,104.25) .. controls (166.05,101.35) and (168.4,99) .. (171.3,99) .. controls (174.2,99) and (176.55,101.35) .. (176.55,104.25) .. controls (176.55,107.15) and (174.2,109.5) .. (171.3,109.5) .. controls (168.4,109.5) and (166.05,107.15) .. (166.05,104.25) -- cycle ;
				\draw  [fill={rgb, 255:red, 144; green, 19; blue, 254 }  ,fill opacity=1 ] (205.25,104.25) .. controls (205.25,101.35) and (207.6,99) .. (210.5,99) .. controls (213.4,99) and (215.75,101.35) .. (215.75,104.25) .. controls (215.75,107.15) and (213.4,109.5) .. (210.5,109.5) .. controls (207.6,109.5) and (205.25,107.15) .. (205.25,104.25) -- cycle ;
				\draw  [fill={rgb, 255:red, 208; green, 2; blue, 27 }  ,fill opacity=1 ] (330,104.25) .. controls (330,101.35) and (332.35,99) .. (335.25,99) .. controls (338.15,99) and (340.5,101.35) .. (340.5,104.25) .. controls (340.5,107.15) and (338.15,109.5) .. (335.25,109.5) .. controls (332.35,109.5) and (330,107.15) .. (330,104.25) -- cycle ;
				\draw  [fill={rgb, 255:red, 74; green, 144; blue, 226 }  ,fill opacity=1 ] (350,104.25) .. controls (350,101.35) and (352.35,99) .. (355.25,99) .. controls (358.15,99) and (360.5,101.35) .. (360.5,104.25) .. controls (360.5,107.15) and (358.15,109.5) .. (355.25,109.5) .. controls (352.35,109.5) and (350,107.15) .. (350,104.25) -- cycle ;
				\draw  [fill={rgb, 255:red, 208; green, 2; blue, 27 }  ,fill opacity=1 ] (370,104.25) .. controls (370,101.35) and (372.35,99) .. (375.25,99) .. controls (378.15,99) and (380.5,101.35) .. (380.5,104.25) .. controls (380.5,107.15) and (378.15,109.5) .. (375.25,109.5) .. controls (372.35,109.5) and (370,107.15) .. (370,104.25) -- cycle ;
				\draw  [fill={rgb, 255:red, 245; green, 166; blue, 35 }  ,fill opacity=1 ] (253.5,104.25) .. controls (253.5,101.35) and (255.85,99) .. (258.75,99) .. controls (261.65,99) and (264,101.35) .. (264,104.25) .. controls (264,107.15) and (261.65,109.5) .. (258.75,109.5) .. controls (255.85,109.5) and (253.5,107.15) .. (253.5,104.25) -- cycle ;
				\draw    (451.75,126.25) -- (460.75,104.25) ;
				\draw    (451.75,126.25) -- (442.8,104.25) ;
				\draw  [fill={rgb, 255:red, 74; green, 144; blue, 226 }  ,fill opacity=1 ] (446.5,126.25) .. controls (446.5,123.35) and (448.85,121) .. (451.75,121) .. controls (454.65,121) and (457,123.35) .. (457,126.25) .. controls (457,129.15) and (454.65,131.5) .. (451.75,131.5) .. controls (448.85,131.5) and (446.5,129.15) .. (446.5,126.25) -- cycle ;
				\draw  [fill={rgb, 255:red, 245; green, 166; blue, 35 }  ,fill opacity=1 ] (437.55,104.25) .. controls (437.55,101.35) and (439.9,99) .. (442.8,99) .. controls (445.7,99) and (448.05,101.35) .. (448.05,104.25) .. controls (448.05,107.15) and (445.7,109.5) .. (442.8,109.5) .. controls (439.9,109.5) and (437.55,107.15) .. (437.55,104.25) -- cycle ;
				\draw  [fill={rgb, 255:red, 74; green, 144; blue, 226 }  ,fill opacity=1 ] (455.5,104.25) .. controls (455.5,101.35) and (457.85,99) .. (460.75,99) .. controls (463.65,99) and (466,101.35) .. (466,104.25) .. controls (466,107.15) and (463.65,109.5) .. (460.75,109.5) .. controls (457.85,109.5) and (455.5,107.15) .. (455.5,104.25) -- cycle ;
				\draw  [fill={rgb, 255:red, 74; green, 144; blue, 226 }  ,fill opacity=1 ] (494.75,126.25) .. controls (494.75,123.35) and (497.1,121) .. (500,121) .. controls (502.9,121) and (505.25,123.35) .. (505.25,126.25) .. controls (505.25,129.15) and (502.9,131.5) .. (500,131.5) .. controls (497.1,131.5) and (494.75,129.15) .. (494.75,126.25) -- cycle ;
				\draw    (357.75,196.25) -- (406,196.25) ;
				\draw    (164.75,196.25) -- (164.75,174.25) ;
				\draw    (309.5,196.25) -- (357.75,196.25) ;
				\draw    (261.25,196.25) -- (261.25,174.25) ;
				\draw    (261.25,196.25) -- (281.25,174.25) ;
				\draw    (261.25,196.25) -- (241.25,174.25) ;
				\draw    (116.5,196.25) -- (116.5,174.25) ;
				\draw    (68.3,196.25) -- (77.3,174.25) ;
				\draw    (68.25,196.25) -- (59.3,174.25) ;
				\draw    (68.25,196.25) -- (116.5,196.25) ;
				\draw    (116.5,196.25) -- (164.75,196.25) ;
				\draw    (164.75,196.25) -- (213,196.25) ;
				\draw    (213,196.25) -- (261.25,196.25) ;
				\draw    (261.25,196.25) -- (309.5,196.25) ;
				\draw  [fill={rgb, 255:red, 208; green, 2; blue, 27 }  ,fill opacity=1 ] (63,196.25) .. controls (63,193.35) and (65.35,191) .. (68.25,191) .. controls (71.15,191) and (73.5,193.35) .. (73.5,196.25) .. controls (73.5,199.15) and (71.15,201.5) .. (68.25,201.5) .. controls (65.35,201.5) and (63,199.15) .. (63,196.25) -- cycle ;
				\draw  [fill={rgb, 255:red, 245; green, 166; blue, 35 }  ,fill opacity=1 ] (111.25,196.25) .. controls (111.25,193.35) and (113.6,191) .. (116.5,191) .. controls (119.4,191) and (121.75,193.35) .. (121.75,196.25) .. controls (121.75,199.15) and (119.4,201.5) .. (116.5,201.5) .. controls (113.6,201.5) and (111.25,199.15) .. (111.25,196.25) -- cycle ;
				\draw  [fill={rgb, 255:red, 74; green, 144; blue, 226 }  ,fill opacity=1 ] (159.5,196.25) .. controls (159.5,193.35) and (161.85,191) .. (164.75,191) .. controls (167.65,191) and (170,193.35) .. (170,196.25) .. controls (170,199.15) and (167.65,201.5) .. (164.75,201.5) .. controls (161.85,201.5) and (159.5,199.15) .. (159.5,196.25) -- cycle ;
				\draw  [fill={rgb, 255:red, 208; green, 2; blue, 27 }  ,fill opacity=1 ] (207.75,196.25) .. controls (207.75,193.35) and (210.1,191) .. (213,191) .. controls (215.9,191) and (218.25,193.35) .. (218.25,196.25) .. controls (218.25,199.15) and (215.9,201.5) .. (213,201.5) .. controls (210.1,201.5) and (207.75,199.15) .. (207.75,196.25) -- cycle ;
				\draw  [fill={rgb, 255:red, 144; green, 19; blue, 254 }  ,fill opacity=1 ] (256,196.25) .. controls (256,193.35) and (258.35,191) .. (261.25,191) .. controls (264.15,191) and (266.5,193.35) .. (266.5,196.25) .. controls (266.5,199.15) and (264.15,201.5) .. (261.25,201.5) .. controls (258.35,201.5) and (256,199.15) .. (256,196.25) -- cycle ;
				\draw  [fill={rgb, 255:red, 245; green, 166; blue, 35 }  ,fill opacity=1 ] (304.25,196.25) .. controls (304.25,193.35) and (306.6,191) .. (309.5,191) .. controls (312.4,191) and (314.75,193.35) .. (314.75,196.25) .. controls (314.75,199.15) and (312.4,201.5) .. (309.5,201.5) .. controls (306.6,201.5) and (304.25,199.15) .. (304.25,196.25) -- cycle ;
				\draw  [fill={rgb, 255:red, 74; green, 144; blue, 226 }  ,fill opacity=1 ] (54.05,174.25) .. controls (54.05,171.35) and (56.4,169) .. (59.3,169) .. controls (62.2,169) and (64.55,171.35) .. (64.55,174.25) .. controls (64.55,177.15) and (62.2,179.5) .. (59.3,179.5) .. controls (56.4,179.5) and (54.05,177.15) .. (54.05,174.25) -- cycle ;
				\draw  [fill={rgb, 255:red, 144; green, 19; blue, 254 }  ,fill opacity=1 ] (72.05,174.25) .. controls (72.05,171.35) and (74.4,169) .. (77.3,169) .. controls (80.2,169) and (82.55,171.35) .. (82.55,174.25) .. controls (82.55,177.15) and (80.2,179.5) .. (77.3,179.5) .. controls (74.4,179.5) and (72.05,177.15) .. (72.05,174.25) -- cycle ;
				\draw  [fill={rgb, 255:red, 144; green, 19; blue, 254 }  ,fill opacity=1 ] (111.25,174.25) .. controls (111.25,171.35) and (113.6,169) .. (116.5,169) .. controls (119.4,169) and (121.75,171.35) .. (121.75,174.25) .. controls (121.75,177.15) and (119.4,179.5) .. (116.5,179.5) .. controls (113.6,179.5) and (111.25,177.15) .. (111.25,174.25) -- cycle ;
				\draw  [fill={rgb, 255:red, 208; green, 2; blue, 27 }  ,fill opacity=1 ] (236,174.25) .. controls (236,171.35) and (238.35,169) .. (241.25,169) .. controls (244.15,169) and (246.5,171.35) .. (246.5,174.25) .. controls (246.5,177.15) and (244.15,179.5) .. (241.25,179.5) .. controls (238.35,179.5) and (236,177.15) .. (236,174.25) -- cycle ;
				\draw  [fill={rgb, 255:red, 74; green, 144; blue, 226 }  ,fill opacity=1 ] (256,174.25) .. controls (256,171.35) and (258.35,169) .. (261.25,169) .. controls (264.15,169) and (266.5,171.35) .. (266.5,174.25) .. controls (266.5,177.15) and (264.15,179.5) .. (261.25,179.5) .. controls (258.35,179.5) and (256,177.15) .. (256,174.25) -- cycle ;
				\draw  [fill={rgb, 255:red, 208; green, 2; blue, 27 }  ,fill opacity=1 ] (276,174.25) .. controls (276,171.35) and (278.35,169) .. (281.25,169) .. controls (284.15,169) and (286.5,171.35) .. (286.5,174.25) .. controls (286.5,177.15) and (284.15,179.5) .. (281.25,179.5) .. controls (278.35,179.5) and (276,177.15) .. (276,174.25) -- cycle ;
				\draw  [fill={rgb, 255:red, 245; green, 166; blue, 35 }  ,fill opacity=1 ] (159.5,174.25) .. controls (159.5,171.35) and (161.85,169) .. (164.75,169) .. controls (167.65,169) and (170,171.35) .. (170,174.25) .. controls (170,177.15) and (167.65,179.5) .. (164.75,179.5) .. controls (161.85,179.5) and (159.5,177.15) .. (159.5,174.25) -- cycle ;
				\draw  [fill={rgb, 255:red, 74; green, 144; blue, 226 }  ,fill opacity=1 ] (352.5,196.25) .. controls (352.5,193.35) and (354.85,191) .. (357.75,191) .. controls (360.65,191) and (363,193.35) .. (363,196.25) .. controls (363,199.15) and (360.65,201.5) .. (357.75,201.5) .. controls (354.85,201.5) and (352.5,199.15) .. (352.5,196.25) -- cycle ;
				\draw  [fill={rgb, 255:red, 245; green, 166; blue, 35 }  ,fill opacity=1 ] (352.5,174.25) .. controls (352.5,171.35) and (354.85,169) .. (357.75,169) .. controls (360.65,169) and (363,171.35) .. (363,174.25) .. controls (363,177.15) and (360.65,179.5) .. (357.75,179.5) .. controls (354.85,179.5) and (352.5,177.15) .. (352.5,174.25) -- cycle ;
				\draw  [fill={rgb, 255:red, 208; green, 2; blue, 27 }  ,fill opacity=1 ] (400.75,196.25) .. controls (400.75,193.35) and (403.1,191) .. (406,191) .. controls (408.9,191) and (411.25,193.35) .. (411.25,196.25) .. controls (411.25,199.15) and (408.9,201.5) .. (406,201.5) .. controls (403.1,201.5) and (400.75,199.15) .. (400.75,196.25) -- cycle ;
				\draw  [fill={rgb, 255:red, 208; green, 2; blue, 27 }  ,fill opacity=1 ] (494.75,104.25) .. controls (494.75,101.35) and (497.1,99) .. (500,99) .. controls (502.9,99) and (505.25,101.35) .. (505.25,104.25) .. controls (505.25,107.15) and (502.9,109.5) .. (500,109.5) .. controls (497.1,109.5) and (494.75,107.15) .. (494.75,104.25) -- cycle ;
				\draw    (454.25,196.25) -- (454.25,174.25) ;
				\draw    (454.25,196.25) -- (474.25,174.25) ;
				\draw    (454.25,196.25) -- (434.25,174.25) ;
				\draw    (454.25,196.25) -- (502.5,196.25) ;
				\draw  [fill={rgb, 255:red, 144; green, 19; blue, 254 }  ,fill opacity=1 ] (449,196.25) .. controls (449,193.35) and (451.35,191) .. (454.25,191) .. controls (457.15,191) and (459.5,193.35) .. (459.5,196.25) .. controls (459.5,199.15) and (457.15,201.5) .. (454.25,201.5) .. controls (451.35,201.5) and (449,199.15) .. (449,196.25) -- cycle ;
				\draw  [fill={rgb, 255:red, 208; green, 2; blue, 27 }  ,fill opacity=1 ] (429,174.25) .. controls (429,171.35) and (431.35,169) .. (434.25,169) .. controls (437.15,169) and (439.5,171.35) .. (439.5,174.25) .. controls (439.5,177.15) and (437.15,179.5) .. (434.25,179.5) .. controls (431.35,179.5) and (429,177.15) .. (429,174.25) -- cycle ;
				\draw  [fill={rgb, 255:red, 74; green, 144; blue, 226 }  ,fill opacity=1 ] (449,174.25) .. controls (449,171.35) and (451.35,169) .. (454.25,169) .. controls (457.15,169) and (459.5,171.35) .. (459.5,174.25) .. controls (459.5,177.15) and (457.15,179.5) .. (454.25,179.5) .. controls (451.35,179.5) and (449,177.15) .. (449,174.25) -- cycle ;
				\draw  [fill={rgb, 255:red, 208; green, 2; blue, 27 }  ,fill opacity=1 ] (469,174.25) .. controls (469,171.35) and (471.35,169) .. (474.25,169) .. controls (477.15,169) and (479.5,171.35) .. (479.5,174.25) .. controls (479.5,177.15) and (477.15,179.5) .. (474.25,179.5) .. controls (471.35,179.5) and (469,177.15) .. (469,174.25) -- cycle ;
				\draw  [fill={rgb, 255:red, 245; green, 166; blue, 35 }  ,fill opacity=1 ] (497.25,196.25) .. controls (497.25,193.35) and (499.6,191) .. (502.5,191) .. controls (505.4,191) and (507.75,193.35) .. (507.75,196.25) .. controls (507.75,199.15) and (505.4,201.5) .. (502.5,201.5) .. controls (499.6,201.5) and (497.25,199.15) .. (497.25,196.25) -- cycle ;
				\draw    (599,196.25) -- (599,174.25) ;
				\draw    (550.75,196.25) -- (599,196.25) ;
				\draw    (550.75,196.25) -- (559.75,174.25) ;
				\draw    (550.75,196.25) -- (541.8,174.25) ;
				\draw  [fill={rgb, 255:red, 74; green, 144; blue, 226 }  ,fill opacity=1 ] (545.5,196.25) .. controls (545.5,193.35) and (547.85,191) .. (550.75,191) .. controls (553.65,191) and (556,193.35) .. (556,196.25) .. controls (556,199.15) and (553.65,201.5) .. (550.75,201.5) .. controls (547.85,201.5) and (545.5,199.15) .. (545.5,196.25) -- cycle ;
				\draw  [fill={rgb, 255:red, 245; green, 166; blue, 35 }  ,fill opacity=1 ] (536.55,174.25) .. controls (536.55,171.35) and (538.9,169) .. (541.8,169) .. controls (544.7,169) and (547.05,171.35) .. (547.05,174.25) .. controls (547.05,177.15) and (544.7,179.5) .. (541.8,179.5) .. controls (538.9,179.5) and (536.55,177.15) .. (536.55,174.25) -- cycle ;
				\draw  [fill={rgb, 255:red, 74; green, 144; blue, 226 }  ,fill opacity=1 ] (554.5,174.25) .. controls (554.5,171.35) and (556.85,169) .. (559.75,169) .. controls (562.65,169) and (565,171.35) .. (565,174.25) .. controls (565,177.15) and (562.65,179.5) .. (559.75,179.5) .. controls (556.85,179.5) and (554.5,177.15) .. (554.5,174.25) -- cycle ;
				\draw  [fill={rgb, 255:red, 74; green, 144; blue, 226 }  ,fill opacity=1 ] (593.75,196.25) .. controls (593.75,193.35) and (596.1,191) .. (599,191) .. controls (601.9,191) and (604.25,193.35) .. (604.25,196.25) .. controls (604.25,199.15) and (601.9,201.5) .. (599,201.5) .. controls (596.1,201.5) and (593.75,199.15) .. (593.75,196.25) -- cycle ;
				\draw  [fill={rgb, 255:red, 208; green, 2; blue, 27 }  ,fill opacity=1 ] (593.75,174.25) .. controls (593.75,171.35) and (596.1,169) .. (599,169) .. controls (601.9,169) and (604.25,171.35) .. (604.25,174.25) .. controls (604.25,177.15) and (601.9,179.5) .. (599,179.5) .. controls (596.1,179.5) and (593.75,177.15) .. (593.75,174.25) -- cycle ;
				\draw    (309.5,207.5) .. controls (312.44,237.39) and (280.81,195.73) .. (279.54,236.41) ;
				\draw [shift={(279.5,239)}, rotate = 270] [fill={rgb, 255:red, 0; green, 0; blue, 0 }  ][line width=0.08]  [draw opacity=0] (8.93,-4.29) -- (0,0) -- (8.93,4.29) -- cycle    ;
				\draw    (358.5,208) .. controls (357.52,245.24) and (392.07,194.12) .. (393.46,234.42) ;
				\draw [shift={(393.5,237)}, rotate = 270] [fill={rgb, 255:red, 0; green, 0; blue, 0 }  ][line width=0.08]  [draw opacity=0] (8.93,-4.29) -- (0,0) -- (8.93,4.29) -- cycle    ;
				\draw    (163.5,206.5) .. controls (166.44,236.39) and (134.81,194.73) .. (133.54,235.41) ;
				\draw [shift={(133.5,238)}, rotate = 270] [fill={rgb, 255:red, 0; green, 0; blue, 0 }  ][line width=0.08]  [draw opacity=0] (8.93,-4.29) -- (0,0) -- (8.93,4.29) -- cycle    ;
				\draw    (505.5,209) .. controls (504.52,246.24) and (539.07,195.12) .. (540.46,235.42) ;
				\draw [shift={(540.5,238)}, rotate = 270] [fill={rgb, 255:red, 0; green, 0; blue, 0 }  ][line width=0.08]  [draw opacity=0] (8.93,-4.29) -- (0,0) -- (8.93,4.29) -- cycle    ;
				
				\draw (255.5,135.65) node [anchor=north west][inner sep=0.75pt]    {$u_{k+1}$};
				\draw (207.25,135.65) node [anchor=north west][inner sep=0.75pt]    {$u_{k}$};
				\draw (400.25,135.65) node [anchor=north west][inner sep=0.75pt]    {$u_{\ell}$};
				\draw (448.5,135.65) node [anchor=north west][inner sep=0.75pt]    {$u_{\ell+1}$};
				\draw (244,245) node [anchor=north west][inner sep=0.75pt]   [align=left] {view of $\displaystyle u_{\ell}$};
				\draw (353,243) node [anchor=north west][inner sep=0.75pt]   [align=left] {view of $\displaystyle u_{k+1}$};
				\draw (98,244) node [anchor=north west][inner sep=0.75pt]   [align=left] {view of $\displaystyle u_{k+1}$};
				\draw (500,244) node [anchor=north west][inner sep=0.75pt]   [align=left] {view of $\displaystyle u_{\ell}$};

			\end{tikzpicture}

			\caption{Top: the caterpillar $G_d$ that is accepted with some certificates. The colors represent certificates. Here, the pair of consecutive vertices $(u_k, u_{k+1})$ received the same certificates as $(u_\ell, u_{\ell+1})$. Bottom: the graph obtained from $G_d$ by copying two times the part between~$u_{k+1}$ and~$u_\ell$, which is also accepted with these certificates because every vertex has the same view as a vertex in~$G_d$.}
			\label{fig:pumping_caterpillar}
		\end{figure}

	\end{proof}
	
	\section{No gap in $d$ in caterpillars with larger radius}
	\label{sec:no-gap-in-d-caterpillars}
	
	\ThmCaterpillarNoGapRadiusTwo*
	
	
	\begin{proof}
		Without loss of generality, we can assume that $f(2) = 1$ and that for all integers $1 \leqslant s \leqslant t$ we have $f(t) - f(s) \leqslant (\log t - \log s)/2$ (else, we can consider the function~$g$ such that $g(n):=(f(n)-f(2))/2 + 1$ which satisfies it). Let $(a_i)_{i \geqslant 1}$ be the sequence given by Lemma~\ref{lem:sequence a_i}, and for every $m \geqslant 1$ let $b_m := \sum_{i=1}^m a_i - 2$. Since $a_1 = 2$ and $a_i \geqslant 1$ for every $i \geqslant 2$, we have $b_m \geqslant 1$ for every $m \geqslant 2$, and the sequence $(b_m)_{m \in \N}$ is increasing.
		Moreover, by definition of $(a_i)_{i \in \N}$, we have:
		
		\begin{equation}
			\label{eq:b_i}
			\forall m \geqslant 1, \;\; \log m \leqslant f(b_m + 2) \leqslant \log m + 1
		\end{equation}
		
		Let us define the following property~$\P$ on caterpillars: the caterpillars satisfying~$\P$ are such that there exists $m \geqslant 2$ and a central path $P:=u_1,\ldots,u_{b_m+1}$ of length $b_m+1$ such that there are $b_m$ degree-1 vertices attached to $u_1$, and for every $i \in \{2, \ldots, b_m+1\}$ there are $i-1$ degree-1 vertices attached to $u_i$. See Figure~\ref{fig:caterpillar_no_gap_d} for an illustration.

		
		
		
		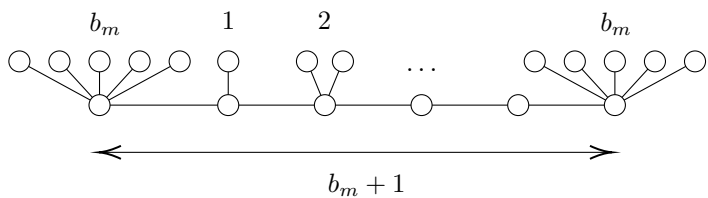
\begin{figure}[h!]
			\centering
			\begin{tikzpicture}[x=0.75pt,y=0.75pt,yscale=-1,xscale=1]
				
				\draw    (125.25,125.25) -- (85.25,103.25) ;
				\draw    (125.25,125.25) -- (165.25,103.25) ;
				\draw    (125.25,125.25) -- (125.25,103.25) ;
				\draw    (125.25,125.25) -- (145.25,103.25) ;
				\draw    (125.25,125.25) -- (105.25,103.25) ;
				\draw  [fill={rgb, 255:red, 255; green, 255; blue, 255 }  ,fill opacity=1 ] (80,103.25) .. controls (80,100.35) and (82.35,98) .. (85.25,98) .. controls (88.15,98) and (90.5,100.35) .. (90.5,103.25) .. controls (90.5,106.15) and (88.15,108.5) .. (85.25,108.5) .. controls (82.35,108.5) and (80,106.15) .. (80,103.25) -- cycle ;
				\draw  [fill={rgb, 255:red, 255; green, 255; blue, 255 }  ,fill opacity=1 ] (100,103.25) .. controls (100,100.35) and (102.35,98) .. (105.25,98) .. controls (108.15,98) and (110.5,100.35) .. (110.5,103.25) .. controls (110.5,106.15) and (108.15,108.5) .. (105.25,108.5) .. controls (102.35,108.5) and (100,106.15) .. (100,103.25) -- cycle ;
				\draw  [fill={rgb, 255:red, 255; green, 255; blue, 255 }  ,fill opacity=1 ] (120,103.25) .. controls (120,100.35) and (122.35,98) .. (125.25,98) .. controls (128.15,98) and (130.5,100.35) .. (130.5,103.25) .. controls (130.5,106.15) and (128.15,108.5) .. (125.25,108.5) .. controls (122.35,108.5) and (120,106.15) .. (120,103.25) -- cycle ;
				\draw  [fill={rgb, 255:red, 255; green, 255; blue, 255 }  ,fill opacity=1 ] (140,103.25) .. controls (140,100.35) and (142.35,98) .. (145.25,98) .. controls (148.15,98) and (150.5,100.35) .. (150.5,103.25) .. controls (150.5,106.15) and (148.15,108.5) .. (145.25,108.5) .. controls (142.35,108.5) and (140,106.15) .. (140,103.25) -- cycle ;
				\draw  [fill={rgb, 255:red, 255; green, 255; blue, 255 }  ,fill opacity=1 ] (160,103.25) .. controls (160,100.35) and (162.35,98) .. (165.25,98) .. controls (168.15,98) and (170.5,100.35) .. (170.5,103.25) .. controls (170.5,106.15) and (168.15,108.5) .. (165.25,108.5) .. controls (162.35,108.5) and (160,106.15) .. (160,103.25) -- cycle ;
				\draw    (382.5,125.25) -- (342.5,103.25) ;
				\draw    (382.5,125.25) -- (422.5,103.25) ;
				\draw    (382.5,125.25) -- (382.5,103.25) ;
				\draw    (382.5,125.25) -- (402.5,103.25) ;
				\draw    (382.5,125.25) -- (362.5,103.25) ;
				\draw    (189.5,125.25) -- (189.5,103.25) ;
				\draw    (237.75,125.25) -- (246.75,103.25) ;
				\draw    (237.75,125.25) -- (228.8,103.25) ;
				\draw    (125.25,125.25) -- (189.5,125.25) ;
				\draw    (189.5,125.25) -- (237.75,125.25) ;
				\draw    (237.75,125.25) -- (286,125.25) ;
				\draw    (286,125.25) -- (334.25,125.25) ;
				\draw    (334.25,125.25) -- (382.5,125.25) ;
				\draw  [fill={rgb, 255:red, 255; green, 255; blue, 255 }  ,fill opacity=1 ] (120,125.25) .. controls (120,122.35) and (122.35,120) .. (125.25,120) .. controls (128.15,120) and (130.5,122.35) .. (130.5,125.25) .. controls (130.5,128.15) and (128.15,130.5) .. (125.25,130.5) .. controls (122.35,130.5) and (120,128.15) .. (120,125.25) -- cycle ;
				\draw  [fill={rgb, 255:red, 255; green, 255; blue, 255 }  ,fill opacity=1 ] (184.25,125.25) .. controls (184.25,122.35) and (186.6,120) .. (189.5,120) .. controls (192.4,120) and (194.75,122.35) .. (194.75,125.25) .. controls (194.75,128.15) and (192.4,130.5) .. (189.5,130.5) .. controls (186.6,130.5) and (184.25,128.15) .. (184.25,125.25) -- cycle ;
				\draw  [fill={rgb, 255:red, 255; green, 255; blue, 255 }  ,fill opacity=1 ] (232.5,125.25) .. controls (232.5,122.35) and (234.85,120) .. (237.75,120) .. controls (240.65,120) and (243,122.35) .. (243,125.25) .. controls (243,128.15) and (240.65,130.5) .. (237.75,130.5) .. controls (234.85,130.5) and (232.5,128.15) .. (232.5,125.25) -- cycle ;
				\draw  [fill={rgb, 255:red, 255; green, 255; blue, 255 }  ,fill opacity=1 ] (280.75,125.25) .. controls (280.75,122.35) and (283.1,120) .. (286,120) .. controls (288.9,120) and (291.25,122.35) .. (291.25,125.25) .. controls (291.25,128.15) and (288.9,130.5) .. (286,130.5) .. controls (283.1,130.5) and (280.75,128.15) .. (280.75,125.25) -- cycle ;
				\draw  [fill={rgb, 255:red, 255; green, 255; blue, 255 }  ,fill opacity=1 ] (329,125.25) .. controls (329,122.35) and (331.35,120) .. (334.25,120) .. controls (337.15,120) and (339.5,122.35) .. (339.5,125.25) .. controls (339.5,128.15) and (337.15,130.5) .. (334.25,130.5) .. controls (331.35,130.5) and (329,128.15) .. (329,125.25) -- cycle ;
				\draw  [fill={rgb, 255:red, 255; green, 255; blue, 255 }  ,fill opacity=1 ] (377.25,125.25) .. controls (377.25,122.35) and (379.6,120) .. (382.5,120) .. controls (385.4,120) and (387.75,122.35) .. (387.75,125.25) .. controls (387.75,128.15) and (385.4,130.5) .. (382.5,130.5) .. controls (379.6,130.5) and (377.25,128.15) .. (377.25,125.25) -- cycle ;
				\draw  [fill={rgb, 255:red, 255; green, 255; blue, 255 }  ,fill opacity=1 ] (223.55,103.25) .. controls (223.55,100.35) and (225.9,98) .. (228.8,98) .. controls (231.7,98) and (234.05,100.35) .. (234.05,103.25) .. controls (234.05,106.15) and (231.7,108.5) .. (228.8,108.5) .. controls (225.9,108.5) and (223.55,106.15) .. (223.55,103.25) -- cycle ;
				\draw  [fill={rgb, 255:red, 255; green, 255; blue, 255 }  ,fill opacity=1 ] (184.25,103.25) .. controls (184.25,100.35) and (186.6,98) .. (189.5,98) .. controls (192.4,98) and (194.75,100.35) .. (194.75,103.25) .. controls (194.75,106.15) and (192.4,108.5) .. (189.5,108.5) .. controls (186.6,108.5) and (184.25,106.15) .. (184.25,103.25) -- cycle ;
				\draw  [fill={rgb, 255:red, 255; green, 255; blue, 255 }  ,fill opacity=1 ] (241.5,103.25) .. controls (241.5,100.35) and (243.85,98) .. (246.75,98) .. controls (249.65,98) and (252,100.35) .. (252,103.25) .. controls (252,106.15) and (249.65,108.5) .. (246.75,108.5) .. controls (243.85,108.5) and (241.5,106.15) .. (241.5,103.25) -- cycle ;
				\draw  [fill={rgb, 255:red, 255; green, 255; blue, 255 }  ,fill opacity=1 ] (337.25,103.25) .. controls (337.25,100.35) and (339.6,98) .. (342.5,98) .. controls (345.4,98) and (347.75,100.35) .. (347.75,103.25) .. controls (347.75,106.15) and (345.4,108.5) .. (342.5,108.5) .. controls (339.6,108.5) and (337.25,106.15) .. (337.25,103.25) -- cycle ;
				\draw  [fill={rgb, 255:red, 255; green, 255; blue, 255 }  ,fill opacity=1 ] (357.25,103.25) .. controls (357.25,100.35) and (359.6,98) .. (362.5,98) .. controls (365.4,98) and (367.75,100.35) .. (367.75,103.25) .. controls (367.75,106.15) and (365.4,108.5) .. (362.5,108.5) .. controls (359.6,108.5) and (357.25,106.15) .. (357.25,103.25) -- cycle ;
				\draw  [fill={rgb, 255:red, 255; green, 255; blue, 255 }  ,fill opacity=1 ] (377.25,103.25) .. controls (377.25,100.35) and (379.6,98) .. (382.5,98) .. controls (385.4,98) and (387.75,100.35) .. (387.75,103.25) .. controls (387.75,106.15) and (385.4,108.5) .. (382.5,108.5) .. controls (379.6,108.5) and (377.25,106.15) .. (377.25,103.25) -- cycle ;
				\draw  [fill={rgb, 255:red, 255; green, 255; blue, 255 }  ,fill opacity=1 ] (397.25,103.25) .. controls (397.25,100.35) and (399.6,98) .. (402.5,98) .. controls (405.4,98) and (407.75,100.35) .. (407.75,103.25) .. controls (407.75,106.15) and (405.4,108.5) .. (402.5,108.5) .. controls (399.6,108.5) and (397.25,106.15) .. (397.25,103.25) -- cycle ;
				\draw  [fill={rgb, 255:red, 255; green, 255; blue, 255 }  ,fill opacity=1 ] (417.25,103.25) .. controls (417.25,100.35) and (419.6,98) .. (422.5,98) .. controls (425.4,98) and (427.75,100.35) .. (427.75,103.25) .. controls (427.75,106.15) and (425.4,108.5) .. (422.5,108.5) .. controls (419.6,108.5) and (417.25,106.15) .. (417.25,103.25) -- cycle ;
				\draw    (127,149) -- (379,149) ;
				\draw [shift={(381,149)}, rotate = 180] [color={rgb, 255:red, 0; green, 0; blue, 0 }  ][line width=0.75]    (10.93,-3.29) .. controls (6.95,-1.4) and (3.31,-0.3) .. (0,0) .. controls (3.31,0.3) and (6.95,1.4) .. (10.93,3.29)   ;
				\draw [shift={(125,149)}, rotate = 0] [color={rgb, 255:red, 0; green, 0; blue, 0 }  ][line width=0.75]    (10.93,-3.29) .. controls (6.95,-1.4) and (3.31,-0.3) .. (0,0) .. controls (3.31,0.3) and (6.95,1.4) .. (10.93,3.29)   ;
				
				\draw (277,105.4) node [anchor=north west][inner sep=0.75pt]    {${\displaystyle \dotsc }$};
				\draw (374,77.4) node [anchor=north west][inner sep=0.75pt]    {$b_{m}$};
				\draw (238,158.4) node [anchor=north west][inner sep=0.75pt]    {$b_{m} +1$};
				\draw (119,77.4) node [anchor=north west][inner sep=0.75pt]    {$b_{m}$};
				\draw (185,77.4) node [anchor=north west][inner sep=0.75pt]    {$1$};
				\draw (233,77.4) node [anchor=north west][inner sep=0.75pt]    {$2$};

			\end{tikzpicture}
			
			\caption{A caterpillar satisfying~$\P$.}
			\label{fig:caterpillar_no_gap_d}
		\end{figure}
		
		First, let us prove that $\P$ can be certified with certificates of size~$O(f(d))$. First, the prover gives a special certificate to each degree-1 vertex. Then, the prover writes the integer~$m$ in the certificate of every vertex.
		The certificates have size $O(\log m)$, and since the diameter~$d$ of a caterpillar in~$\P$ is equal to $b_m + 2$, the size of the certificates is $O(f(d))$ because $\log m \leqslant f(b_m + 2)$.
		Let us now explain how the vertices check their certificates. First, every vertex checks that the integer~$m \geqslant 2$ written in its certificate is the same as in the certificates of its neighbors.
		Every vertex also checks that it has the special certificate if and only if it has degree~$1$. Then, note that since $b_m \geqslant 1$ for every $m \geqslant 2$, a vertex belongs to the path~$P$ if and only if it has degree at least~$2$.
		Let~$u$ be a vertex having degree at least~$2$.
		Let us note by~$\delta_1(u)$ the number of neighbors of~$u$ that have degree~1.
		For every neighbor~$v$ of~$u$, since the vertices can view at distance $r \geqslant 2$, $u$ can determine $\delta_1(v)$.
		The vertex~$u$ checks that one of the three following cases holds, and rejects if it is not the case:
		\begin{enumerate}
			\item $\delta_1(u)=b_m$, and~$u$ has exactly one neighbor with degree at least~$2$, denoted by~$v$, such that $\delta_1(v) \in \{1,b_{m}-1\}$
			\item $\delta_1(u)=1$, and~$u$ has exactly two neighbors with degree at least~$2$, denoted by~$v$ and~$w$, such that $\delta_1(v)=b_m$ and $\delta_1(w)=2$
			\item there exists $i \in \{2, \ldots, b_m-1\}$ such that $\delta_1(u)=i$, and $u$ has exactly two neighbors with degree at least~$2$, denoted by~$v$ and~$w$, such that $\delta_1(v) = i-1$ and~$\delta_1(w) = i+1$
		\end{enumerate}
		
		Then, it is straightforward that if all the vertices accept, the caterpillar satisfies~$\P$, and conversely that if it satisfies~$\P$ then the verification of every vertex accepts with the certificates described above.
		
		Let us now prove that $\P$ cannot be certified with certificates of size~$o(f(d))$. Assume by contradiction that certificates of size~$o(f(d))$ are sufficient. For every $m \geqslant 2$, let us denote by~$G_m$ the caterpillar satisfying~$\P$ having its central path of length~$b_m+1$, by $u_1, \ldots, u_{b_m+1}$ the vertices of its central path, and by~$d_m:=b_m+2$ its diameter. By equation~(\ref{eq:b_i}), we have $f(d_m) = \Theta(\log m)$.
		Thus, if we denote by~$s(m)$ the size of the certificates that makes all the vertices of~$G_m$ accept, we have $s(m)=o(\log m)$. So there exists $m_0 \geqslant 2$ such that, for all $m \geqslant m_0$, $b_m > 2r$, and the two following conditions are all satisfied:
		\begin{equation}
			\label{eq:technical inequality gap d}
			s(m) \leqslant \frac{\log \sqrt{m}}{4r^2} \qquad \text{and} \qquad \sqrt{\frac{m}{2}} > \frac{\log \sqrt{2m}}{4r^2}
		\end{equation}

		For every $m \geqslant m_0$, let us denote by $X$ the set of vertices in~$G_m$ that contains $u_2, \ldots, u_{2r+1}$ and all their neighbors of degree~1, and let us also denote by~$G_m[X]$ the subgraph of $G_m$ induced by~$X$ (note that, by definition of~$G_m$, the structure of~$G_m[X]$ does not depend on~$m$, because the length $b_m$ of the central path of~$G_m$ is at least $2r+1$). Let us also denote by~$c_m$ an assignment of certificates by the prover to the vertices of~$G_m$ such that every vertex accepts.
		Let us prove the following claim:
		
		\begin{claim}
			\label{claim:pigeonhole no gap d}
			There exists two distinct integers $m,m' \geqslant 2$ such that $G_m$ and $G_{m'}$ are certified with certificates of the same size, and such that $c_m$ and $c_{m'}$ are equal on the vertices of~$X$ in~$G_m$ and~$G_{m'}$.
		\end{claim}
		
		\begin{proof}
			For every $m \in \{m_0+1, \ldots, 2m_0\}$, the size of the certificates that make all the vertices of $G_m$ accept is $s(m)$ and we have $s(m) \leqslant s(2m_0) \leqslant \frac{\log \sqrt{2m_0}}{4r^2}$. By the second inequality of~(\ref{eq:technical inequality gap d}), we know that $m_0 > \sqrt{2m_0}\frac{\log \sqrt{2m_0}}{4r^2}$. Thus, by the pigeonhole principle, there exists a set $I \subseteq \{m_0+1, \ldots, 2m_0\}$ of size at least $\sqrt{2m_0}+1$ such that $s(m)$ is the same for all the integers $m \in I$. Moreover, simply by counting, the total number of vertices in $X$ is at most~$4r^2$. So, for $m \in I$, the total number of possible ways to attribute certificates of size~$s(m)$ to the vertices of $G_m[X]$ is at most $2^{4r^2s(m)} \leqslant \sqrt{m} \leqslant \sqrt{2m_0}$. Finally, again by the pigeonhole principle, there exist two distinct integers $m, m' \in I$ such $c_m$ and $c_{m'}$ are equal on the vertices of $X$.
		\end{proof}
		
		Let $m,m'$ be the two integers given by Claim~\ref{claim:pigeonhole no gap d}.
		Consider the following caterpillar~$G_{m,m'}$: the central path has length~$b_m+1$, the vertices of the central path are denoted by~$u_1, \ldots, u_{b_m+1}$, $u_1$ has $b_{m'}$ neighbors of degree~$1$, and for all $i \in \{2, \ldots, b_m+1\}$, $u_i$ has $i-1$ neighbors of degree~$1$. Let us consider the following assignment of certificates: $u_1$ and its neighbors receive the certificates given by~$c_{m'}$, and for every $i \geqslant 2$, $u_i$ and its neighbors receive the certificates given by~$c_m$. Then, all the vertices accept with this assignment of certificates: indeed, $u_1$ and its neighbors have the same view as in~$G_{m'}$, and for every $i \geqslant 2$, $u_i$ and its neighbors have the same view as in~$G_{m}$. It is finally a contradiction and concludes the proof, because $G_{m,m'}$ does not satisfy~$\P$ (indeed, $m \neq m'$ implies that $b_m \neq b_{m'}$ since $(b_m)_{m\geqslant 1}$ is increasing). \qedhere

	\end{proof}

	
	
	
	\section{No gap in general graphs via caterpillars}
	\label{sec:no-gap-general-graphs}
	
	In this section, we transfer the results we got in Section~\ref{sec:no-gap-in-n-caterpillars} and~\ref{sec:no-gap-in-d-caterpillars} from caterpillars to general graphs. 
	
	\begin{theorem}
		\label{thm:general-graphs}
		Let $f : \N \to \mathbb{R}$ be a non-decreasing function such that $\lim_{n \rightarrow +\infty}f(n) = +\infty$ and for all integers $1 \leqslant s \leqslant t$ we have $f(t)-f(s) \leqslant \log t - \log s$. Then:
		\begin{enumerate}
			\item For verification radius 1, there exists a property of general graphs that can be certified with certificates of size~$O(f(n))$ and not with certificates of size~$o(f(n))$.
			\item For verification radius $r>1$, there exists a property of general graphs that can be certified with certificates of size~$O(f(d))$ and not with certificates of size~$o(f(d))$, where $d$ is the diameter.
		\end{enumerate}
	\end{theorem}
	
	\begin{proof}
		For both items, we use the same strategy, and use Theorems~\ref{thm:caterpillars-no-gap-in-n} and~\ref{thm:caterpillar-no-gap-radius-2} as black-boxes. Let us consider the first item for concreteness, and a property $\mathcal{P}$ of caterpillars having a local certification of size $O(f(n))$ but not of size $o(f(n))$, for radius 1. We consider the language of graphs that can be created by concatenating caterpillars having property $\mathcal{P}$, by edges between endpoints, in a circular fashion. 
		
		Let us first prove that the fact that the graph is a circular caterpillar (a cycle with pending vertices) can be certified with a constant number of bits. The certification on a correct instances consists in assigning a label \emph{cycle} on the vertices of the cycles, and \emph{leaf} on the others. Nodes with label \emph{leaf} check that they have exactly one neighbor, and that it is labeled \emph{cycle}. Nodes labeled \emph{cycle} check that they have exactly two neighbors labeled with \emph{cycle}. The only thing to prove is that there cannot be two disjoint cycles with label \emph{cycle}. This cannot happen since if there were at least two disjoint cycles, then a pair of such cycles would be connected by a path of vertices labeled $\emph{leaf}$, which is impossible.
		
		We now want to certify that the circular caterpillar is made of a concatenation of graphs having property $\mathcal{P}$. On a correct instance, in addition to the certificates for checking the circular caterpillar structure, every node if given the following pieces of information. 
		\begin{enumerate}
			\item A label saying whether it is the endpoint of a connecting edge or not.
			\item The certificate it would have got if only its own caterpillar existed (without the concatenation). 
		\end{enumerate}
		Every node then checks (in addition to the structure check) that:
		\begin{enumerate}
			\item If it has not been labeled as a connecting node, then the verifier from the caterpillar version accepts. 
			\item Otherwise, it should have a unique neighbor also labeled as a connecting node, and the verifier from the caterpillar version should accept, when forgetting the edge to that node. 
		\end{enumerate}
		This scheme is clearly correct. 
		It size is in $O(f(n))$ since every concatenated part of size $k$ uses $O(f(k))$ , with $k\leq n$, and $f$ is non-decreasing. 
		Now this cannot be certified with $o(f(n))$ bits because otherwise we could contradict our previous theorem. Indeed, in that case, the graph made by taking one caterpillar and connecting both endpoints by an edge would be certified with $o(f(n))$ bits. Then in the non-circular caterpillars, we could have a $o(f(n))$ certification, by using the exact same scheme, except that the certificates of the endpoints would be copied all along the path, to allow the checks on the virtual concatenating edge. 
		
		The proofs works exactly the same for the other item of the theorem. 
	\end{proof}

	\section{Extensions for identifiers and global knowledge}
	\label{sec:extensions-ID}
	
	All the results so far were for anonymous networks, without any knowledge of $n$. In this section we prove a few results about settings where there are identifiers or some knowledge of $n$ is given. 
	For simplicity, we focus on paths. 
	We first prove in Subsection~\ref{subsec:no-gap-exact-n} that in labeled paths, the gap disappears when $n$ is given (either explicitly, or implicitly via identifiers). Then for unlabeled paths, where exact knowledge of $n$ makes the problem trivial, we prove that having a sharp estimate on $n$ allows to break the $\log \log n$ barrier (Subsection~\ref{subsec:approx-n}). 
	Finally, we prove in Subsection~\ref{subsec:large-ID} that the gap result still holds if the nodes are equipped with arbitrarily large identifiers.

	\subsection{No gap for labeled paths with identifiers in $[1,n]$ or exact knowledge of $n$}
	\label{subsec:no-gap-exact-n}

	\begin{theorem}\label{thm:nogap_exactID}
		Let $f : \N \to \N$ be a non decreasing function such that $\lim_{n \rightarrow +\infty}f(n) = +\infty$. Then, there exists a property on labeled paths that can be certified with certificates of size~$O(f(n))$ and not with certificates of size~$o(f(n))$ if the vertices received unique IDs in the range $[1,n]$ or the exact knowledge of~$n$.
	\end{theorem}
	
	\begin{proof}
		The language we consider is the following set of labeled paths. All nodes at distance at most $f(n)$ from an endpoint are given a bit, and all other nodes are given an empty label, and the strings corresponding to the concatenation of the bits of each endpoints should be equal. 
		
		Let us start with the case where $n$ is known. In a correct instance, a node received the following certificate:
		\begin{enumerate}
			\item The string that appears on both endpoints. 
			\item If it has a bit, its distance to the endpoint. 
		\end{enumerate}
		The nodes perform the following check. 
		\begin{enumerate}
			\item The string is the same as the one given to its neighbors.
			\item If it has a bit, the distances are consistent (in particular with the value of $f(n)$ that it can compute) and the string is consistent with its distance and its input bit. 
		\end{enumerate}
		This certification is clearly correct and uses $O(f(n))$ bits. Also $o(f(n))$ bits is not possible, since it would directly contradict the communication complexity of EQUALITY (see \cite{Feuilloley21} for this kind of proof, or simply use a counting argument like in Section~\ref{sec:no-gap-general-graphs} or \cite{GoosS16}).
		
		Now for the case where nodes only have IDs in $[1,n]$ but do not know $n$, the trick is the following. We give every node its distance modulo 3 to the node with the largest identifier. This is typ of counter is known to certify a unique leader in trees (see \emph{e.g.} \cite{FeuilloleyBP22}). Then every node checks that when it applies $f$ to its ID, the result is at most the size of the string, and the leader checks that it is exactly the size of the string. 
	\end{proof}
	
	

	\subsection{Breaking the $\log\log n$ barrier with sharp estimates of $n$ in unlabeled paths}
	\label{subsec:approx-n}
	
	
	In this subsection we prove that if we have sharp estimate of $n$ then the gap result does not hold anymore. 
	
	\begin{theorem}
		Let $g(n) \in o(\log n)$, a non-decreasing function such that for all integers $1 \leqslant s \leqslant t$ we have $g(t)-g(s) \leqslant \log t - \log s$.
		Consider a model where the nodes are given an approximation of $n$, denoted $\hat{n}$, with the promise that $|\hat{n}-n|\leq g(n)$. In this model, 
		there are properties that can be certified with $O(\log g(n))$ bits (which is in $o(\log \log n)$) but not with $O(1)$ bits.
	\end{theorem}
	
	Note that the approximation needed is much more precise than the usual linear or even polynomial upper bounds on $n$ that are more common in distributed computing. 
	
	\begin{proof}
		The idea of the upper bound is easy. The node know that the real value of $n$ is close to $\hat{n}$, more precisely, $n$ must be in $[\hat{n}-g(\hat n)-1,\hat{n}+g(\hat n)+1]$ (thanks to the constraints on the growth of $g$, for large enough $n$). The certification simply consists in a counter modulo $4g(n)$. Then the last node can deduce the exact value of $n$ because only one such value is consistent with both the interval and the congruence. This uses $O(\log g(n))$ bits.
		
		For the lower bound, suppose that for any property there would exist a constant size certification, when an approximation of $n$ as in the theorem statement.
		Consider an infinite set of integers $S$, such that for any integer $k$ there exists at most one element of $S$ in $[k-\log k, k+\log k]$ (such a set can be built greedily). Consider the language of the paths whose lengths are in $S$, and suppose that it has a constant size certification. For $n\in S$, consider the certificates given to the nodes when $\hat{n}=n$. 
		Since it is a constant size certification, there exists a pair of edge having exactly the same certificates, which are at constant distance one from the other. Now consider the instance where we have removed the part in between these edges, including the two internal nodes of the pairs, keeping the certificates of the remaining nodes. In this new instance, of size $n-c$ for a constant $c$, $\hat{n}$ is still a legitimate approximation of the length, and all nodes still accept. But this instance cannot be in the language by construction, which is a contradiction. 
	\end{proof}
	
	\subsection{Arbitrarily large identifiers}
	\label{subsec:large-ID}
	
	Another setting that is stronger than the anonymous case but does not leak the knowledge of $n$ is the one where the nodes are given distinct identifiers, which can be arbitrarily large. Here a node can see its own certificate and identifier and the certificates of its neighbors. Not seeing the identifiers of the neighbors is the classic assumption in the state model of self-stabilization. In this setting, we show that we are basically back to the anonymous case, in particular the gap result still holds. 
	
	\begin{theorem}
		The gap result of Theorem~\ref{thm:gap-paths-chrobak} still holds if we assume that the nodes are given distinct identifiers from an arbitrarily large set. 
	\end{theorem}
	
	The proof follows the same path as the ones in~\cite{BlinDF20}, which extend the ones of~\cite{BeauquierGJ99}. We sketch it here. 
	
	\begin{proof}
		For the sake of contradiction, suppose there exists a property with complexity $f(n)$ in $o(\log \log n)$ but not in $O(1)$ in paths, when the nodes are given distinct identifiers.  
		We claim that for every certificate size $s$, we can find an infinite set of identifiers $I_s$, such that the mapping from the neighborhood's certificates  to the accept/reject decision is the same for all identifiers of $I_s$. 
		Indeed, for every $s$ the number of such mappings is bounded by a function of $s$ 
		(a upper bound is $2^{2^{2s}}$ 
		taking into account the fact 
		that the mappings are special at the extremities), hence we can apply pigeon-hole principle, and get this infinite set.
		Now we can define a certification in anonymous paths in the following way. The prover picks the relevant certificate size $s$, and assigns certificates as if all nodes had an ID from $I_s$. This is always possible because $I_s$ is infinite. Each node computes $I_s$ from $s$, and simulates the local verification of the original scheme, imagining that it has an ID from $I_s$. 
		The correctness follows from the definition of $I_s$ and from the correctness of the original definition. And this certification uses $f(n)$ bits, which contradicts Theorem~\ref{thm:gap-paths-chrobak}.
	\end{proof}
	
	This proof can easily be extended to bounded degree trees.

	\newpage{}

	\bibliography{biblio}
	
\end{document}